\documentclass{article}

\usepackage[
backend=bibtex,
style=alphabetic,
sorting=nyt
]{biblatex}
\usepackage{amsmath}
\usepackage{amsfonts}
\usepackage{amsthm}
\usepackage{mathrsfs}
\usepackage{amssymb}
\usepackage{bbm} 
\usepackage{mathtools}
\usepackage{MnSymbol} 

\usepackage[svgnames]{xcolor}

\usepackage{breakurl}

\usepackage[colorlinks,linkcolor=blue,citecolor=Green,breaklinks=true]{hyperref}

\usepackage{tikz}
\usepackage{tikz-cd}
\usepackage{graphicx}
\usepackage{float}

\usetikzlibrary{decorations.markings}

\usepackage{subcaption}

\usepackage{color}

\usepackage{authblk}

\theoremstyle{plain}
\newtheorem{theorem}{Theorem}[section]
\newtheorem{proposition}{Proposition}[section]
\newtheorem{corollary}{Corollary}[section]
\newtheorem{lemma}{Lemma}[section]

\newtheorem*{ackn}{Acknowledgements}

\theoremstyle{definition}

\theoremstyle{remark}
\newtheorem{remark}{Remark}[section]
\newtheorem{notation}{Notation}[section]

\numberwithin{equation}{subsection}

\title{ODE/IM correspondence in the semiclassical limit: Large degree asymptotics of the spectral determinants for the ground state potential}

\author[1]{Gabriele Degano \thanks{gabriele.degano@outlook.it}}
\affil[1]{\small{Grupo de Física Matemática, Departamento de Matemática, Instituto Superior
Técnico, Av. Rovisco Pais, Lisboa 1049-001, Portugal}}
\affil[1]{\small{Departamento de Matemática, Faculdade de Ciências da Universidade de Lisboa,
Campo Grande Edifício C6, Lisboa 1749-016, Portugal}}

\date{}

\begin{document}

\maketitle

\begin{abstract}
We study a Schrödinger-like equation for the anharmonic potential $x^{2 \alpha}+\ell(\ell+1) x^{-2}-E$ when the anharmonicity $\alpha$ goes to $+\infty$. When $E$ and $\ell$ vary in bounded domains, we show that the spectral determinant for the central connection problem converges to a special function written in terms of a Bessel function of order $\ell+\frac{1}{2}$ and its zeros converge to the zeros of that Bessel function. We then  study the regime in which $E$ and $\ell$ grow large as well, scaling as $E\sim \alpha^2 \varepsilon^2$ and $\ell\sim \alpha p$. When $\varepsilon$ is greater than $1$ we show that the spectral determinant for the central connection problem is a rapidly oscillating function whose zeros tend to be distributed according to the continuous density law $\frac{2p}{\pi}\frac{\sqrt{\varepsilon^2-1}}{\varepsilon}$. When $\varepsilon$ is close to $1$ we show that the spectral determinant converges to a function expressed in terms of the Airy function $\operatorname{Ai}(-)$ and its zeros converge to the zeros of that function.   This work is motivated by and has applications to the ODE/IM correspondence for the quantum KdV model. 
\end{abstract}

\tableofcontents

\section{Introduction}

In this paper, we study the following Schrödinger equation
\begin{equation}
\label{eqn}
\frac{d^2\psi}{dx^2}=\left(x^{2 \alpha}+\frac{\ell(\ell+1)}{ x^2}-E\right)\psi,\quad x\in\widetilde{\mathbb{C}^*},
\end{equation}
where $\widetilde{\mathbb{C}^*}$ denotes the universal covering space of the puctured complex plane $\mathbb{C}^*$, $(E,\ell)\subset\mathbb{C}^2$ are complex parameters called the \textit{energy} and the \textit{angular momentum}, respectively, and $\alpha>0$ is a large positive real parameter, called the \textit{anharmonicity degree}. The function
\begin{equation}
\label{15-maggio-2024-1}
V^{(\alpha)}(x;E,\ell):=x^{2 \alpha}+\frac{\ell(\ell+1)}{ x^2}-E
\end{equation}
is referred to as the \textit{potential} of the Schrödinger equation~\eqref{eqn}. It has singularities at $x=0$, which is called a \textit{(formal) regular singularity} and at $x=\infty$ in $\widetilde{\mathbb{C}^*}$, which is called a \textit{(formal) irregular singularity}. This family of anharmonic oscillators appears in the theory of quantum integrable systems, in particular, in the ODE/IM correspondence for the quantum KdV model, which constitutes the main motivation for the study of this paper. 
\\

Our main objects of study are entire functions of the energy parameter called \textit{spectral determinants} for the anharmonic oscillator~\eqref{eqn}. These functions are defined through pairs of solutions to~\eqref{eqn} which  belong to the space $\mathcal{A}_\ell^{(\alpha)}$ of functions that solve~\eqref{eqn} and are analytic for $(x,E)\in \widetilde{\mathbb{C}^*}\times \mathbb{C}$. The space $\mathcal{A}_\ell^{(\alpha)}$ is a free module of rank $2$ over the ring $\mathcal{O}(E)$ of entire functions with variable $E\in\mathbb{C}$. The pairs of solutions to~\eqref{eqn} we are going to introduce are distinguished $\mathcal{O}(E)$-basis of $\mathcal{A}_\ell^{(\alpha)}$. 

In order to introduce the first distinguished $\mathcal{O}(E)$-basis of $\mathcal{A}_\ell^{(\alpha)}$ we consider the \textit{quantum monodromy} or \textit{Dorey-Tateo symmetry} $\mathcal{M}^{(\alpha)}$, which is an operator on $\mathcal{A}_\ell^{(\alpha)}$ whose action is defined by
\begin{equation}
\label{11-maggio-2024-2}
\begin{aligned}
\mathcal{M}^{(\alpha)}\colon\, \mathcal{A}_\ell^{(\alpha)}\ni \psi^{(\alpha)}(x;E,\ell)\mapsto \psi^{(\alpha)}\left(x e^{\frac{i \pi}{\alpha+1}};E e^{-\frac{2 i  \pi}{\alpha+1}},\ell\right)\in\mathcal{A}_\ell^{(\alpha)}.
\end{aligned}
\end{equation}
For generic values of $\ell$, it has distinct eigenvalues. For such values of $\ell$, we denote by $\chi_+^{(\alpha)}(x;E,\ell), \chi_-^{(\alpha)}(x;E,\ell)\in\mathcal{A}_\ell^{(\alpha)}$ the eigenvectors of $\mathcal{M}^{(\alpha)}$ with normalizations 
\begin{equation}
\label{12-maggio-2024-1}
\lim_{x\to 0} \Gamma\left(1+\frac{\ell+\frac{1}{2}}{\alpha+1}\right) x^{-\ell-1}\chi_+^{(\alpha)}(x;E,\ell)=1
\end{equation}
and
\begin{equation}
\label{12-maggio-2024-2}
\lim_{x\to 0}\Gamma\left(1-\frac{\ell+\frac{1}{2}}{\alpha+1}\right) x^\ell \chi_-^{(\alpha)}(x;E,\ell)=1.
\end{equation}
The function $\chi_+^{(\alpha)}(x;E,\ell)$ is said to be subdominant as $x\to 0$, while the function $\chi_-^{(\alpha)}(x;E,\ell)$ is said to be dominant as $x\to 0$. They are referred to as distinguished \textit{Frobenius solutions}.

To introduce the next distinguished $\mathcal{O}(E)$-basis of $\mathcal{A}_\ell^{(\alpha)}$, for any $k\in\mathbb{Z}$ we let $\Sigma_k^{(\alpha)}$ and $\mathcal{S}_k^{(\alpha)}$ to be the sectors
\begin{equation}
\label{3-aprile-2024-2-bis}
\Sigma_k^{(\alpha)}:=\left\{\frac{2k-1}{\alpha+1}\frac{\pi}{2}<\arg(x)<\frac{2k+1}{\alpha+1}\frac{\pi}{2}\right\}\subset \widetilde{\mathbb{C}^*},
\end{equation}
and
\begin{equation}
\label{3-aprile-2024-2}
\mathcal{S}_k^{(\alpha)}:=\left\{\frac{2k-3}{\alpha+1}\frac{\pi}{2}<\arg(x)<\frac{2k+3}{\alpha+1}\frac{\pi}{2}\right\}=\Sigma_{k-1}^{(\alpha)}\cup\overline{\Sigma_k^{(\alpha)}}\cup\Sigma_{k+1}^{(\alpha)}\subset \widetilde{\mathbb{C}^*}.
\end{equation}
For $\alpha>1$, we denote by $\psi_k^{(\alpha)}(x;E,\ell)$ the unique element of $\mathcal{A}_\ell^{(\alpha)}$ satisfying the normalization
\begin{equation}
\label{3-aprile-2024-1}
\lim_{\substack{x\to\infty \\ x\in S}} \sqrt{\frac{2\alpha+2}{\pi}}x^{\frac{\alpha}{2}} e^{\frac{x^{\alpha+1}}{\alpha+1}e^{-ik\pi}}\psi_k^{(\alpha)}(x;E,\ell)=1,
\end{equation}  
for any proper subsector $S$ of $\mathcal{S}_k^{(\alpha)}$. The functions $\psi_k^{(\alpha)}(x;E,\ell)$ are referred to as \textit{Sibuya's solutions}. For any $k\in\mathbb{Z}$, the pair $\psi_k^{(\alpha)}(x;E,\ell)$, $\psi_{k+1}^{(\alpha)}(x;E,\ell)$ is a $\mathcal{O}(E)$-basis of $\mathcal{A}_\ell^{(\alpha)}$. The function $\psi_k^{(\alpha)}(x;E,\ell)$ is said to be subdominant at $x=\infty$ in $\mathcal{S}_k^{(\alpha)}$, while the function $\psi_{k+1}^{(\alpha)}(x;E,\ell)$ is said to be dominant at $x=\infty$ in $\mathcal{S}_k^{(\alpha)}$. 

We now have all the elements to define our spectral determinants. Let us start considering the pair of Frobenius solutions $\chi_+^{(\alpha)}(x;E,\ell)$ and $\chi_-^{(\alpha)}(x;E,\ell)$. Since they form a $\mathcal{O}(E)$-basis of $\mathcal{A}_\ell^{(\alpha)}$, for generic values of $\ell$ we can write
\begin{equation}
\label{31-agosto-2024-1}
\begin{aligned}
 \psi_0^{(\alpha)}(x;E,\ell)=& -\frac{\mathcal{Q}_-^{(\alpha)}(E;\ell)}{\frac{2\alpha+2}{\pi}\sin\left(\frac{\pi}{\alpha+1}\left(\ell+\frac{1}{2}\right)\right)} \chi_+^{(\alpha)}(x;E,\ell) \\ & +\frac{\mathcal{Q}_+^{(\alpha)}(E;\ell)}{\frac{2\alpha+2}{\pi}\sin\left(\frac{\pi}{\alpha+1}\left(\ell+\frac{1}{2}\right)\right)} \chi_-^{(\alpha)}(x;E,\ell) \\[2ex]
\end{aligned}
\end{equation}
The elements $\mathcal{Q}_+^{(\alpha)}(E;\ell)$ and $\mathcal{Q}_-^{(\alpha)}(E;\ell)$ are the \textit{spectral determinants for the central connection problem}. This name is because their zeros are precisely the values of the energy parameter $E$ for which there exists a solution to~\eqref{eqn} which is subdominant at $x=0$ and at $x=\infty$ in $\mathcal{S}_0^{(\alpha)}$, or there exists a solution to~\eqref{eqn} which is dominant at $x=0$ and subdominant at $x=\infty$ in $\mathcal{S}_0^{(\alpha)}$, respectively. In~\eqref{31-agosto-2024-1} we have extracted a factor $\left[\frac{2\alpha+2}{\pi}\sin\left(\frac{\pi}{\alpha+1}\left(\ell+\frac{1}{2}\right)\right)\right]^{-1}$ for latter convenience. Let us now fix an integer $k\in\mathbb{Z}$ and let us consider the pair of Sibuya's solutions $\psi_k^{(\alpha)}(x;E,\ell)$ and $\psi_{k+1}^{(\alpha)}(x;E,\ell)$. As before, since they are a $\mathcal{O}(E)$-basis of $\mathcal{A}_\ell^{(\alpha)}$ and due to normalization 
\begin{equation}
\label{16-gennaio-2024-1}
\lim_{\substack{ x\to\infty \\ x\in S}}\frac{\psi_{k-1}^{(\alpha)}(x;E,\ell)}{\psi_{k+1}^{(\alpha)}(x;E,\ell)}=1,
\end{equation}
for any proper subsector $S$ of $\Sigma_k^{(\alpha)}$, we can write
\begin{equation}
\label{31-agosto-2024-2}
\psi_{k-1}^{(\alpha)}(x;E,\ell)=\sigma_k^{(\alpha)}(E;\ell)\psi_k^{(\alpha)}(x;E,\ell)+  \psi_{k+1}^{(\alpha)}(x;E,\ell),
\end{equation}
for some $\sigma_k^{(\alpha)}(E;\ell)\in\mathcal{O}(E)$. The coefficient $\sigma_k^{(\alpha)}(E;\ell)$ is called the \textit{k-th Stokes multiplier} and it is the \textit{spectral determinant for the lateral connection problem between $\Sigma_{k-1}^{(\alpha)}$ and $\Sigma_{k+1}^{(\alpha)}$}. As before, this name is because the zeros of $\sigma_k^{(\alpha)}(E;\ell)$ are precisely the values of the energy parameter $E$ for which there exists a solution to~\eqref{eqn} which is subdominant as $x\to \infty$ in $\Sigma_{k-1}^{(\alpha)}$ and as $x\to \infty$ in $\Sigma_{k+1}^{(\alpha)}$ (the limits being nontangential).
\\

The main results of this paper are the following three theorems where we characterize the large $\alpha$ limit of the spectral determinants in three asymptotic regimes:
\begin{itemize}
\item[1)]$\alpha\to+\infty$ with $E$ and $\ell$ taking values on compact subsets;
\item[2)] $\alpha\to+\infty$ and $E=4p^2(\alpha+1)^2\varepsilon^2$, $\ell=2p(\alpha+1)-\frac{1}{2}$, with  $\varepsilon$ and $p$ taking values on compact subsets and $|\varepsilon|>1$, $p\ne 0$;
\item[3)] $\alpha\to+\infty$ and $E=4p^2(\alpha+1)^2\left(1+\frac{\eta}{(\alpha+1)^{\frac{2}{3}}}\right)^2$, $\ell =2p(\alpha+1)-\frac{1}{2}$, with $\eta$ and $p$ taking values on compact subsets, $p\ne 0$.  
\end{itemize}

Before stating the first theorem, we introduce the following discrete set in the $\ell$-space:
\begin{equation}
\label{14-maggio-2024-10}
\Lambda^{(\alpha)}:=\left\{i+\alpha j,\,(i,j)\in\mathbb{Z}^2_{\ge 0}\right\}
\end{equation}
(as explained in Section \ref{review} these are the non-generic values of $\ell$ for which there is no dominant solution at $x=0$ in $\mathcal{A}_\ell^{(\alpha)}$). Finally, we will denote by $\mathbb{H}_+$ the right half plane
\begin{equation}
\label{22-aprile-2024-9}
\mathbb{H}_+:=\left\{\operatorname{Re}(\ell)>-\frac{1}{2}\right\}
\end{equation}
and we will use the following
\begin{notation}
For any two complex valued functions $f(x)$, $g(x)$ we write $|f(x)|\lesssim |g(x)|$ (resp. $|f(x)|\gtrsim |g(x)|$) if there exists a consant $\mathcal{C}>0$ such that $|f(x)|\le C |g(x)|$ (resp. $|f(x)|\ge C |g(x)|$). If the functions $f(x)$ and $g(x)$ depend on some parameters $\underline{u}\in K\subset \mathbb{C}^n$ we write $|f(x)|\lesssim_{\underline{u}_0} |g(x)|$ (resp. $|f(x)|\gtrsim_{\underline{u}_0} |g(x)|$) or $|f(x)|\lesssim_K |g(x)|$ (resp. $|f(x)|\gtrsim_K |g(x)|$) if the constant $C>0$ depends on a fixed value $\underline{u}_0\in K$ of the parameters or on the domain $K$ in which the parameters take values, respectively. 
\end{notation}

\begin{theorem}
\label{theorem-23-aprile-2024-1}
The following holds:
\begin{itemize}
\item[i)] Let $W_+\subset \mathbb{C}\times \mathbb{H}_+ $ and $W_-\subset \mathbb{C}\times \left(\mathbb{H}_+\setminus \Lambda^{(\alpha)}\right)$ (with $\Lambda^{(\alpha)}$ as in~\eqref{14-maggio-2024-10}) be compact subsets. There exists a constant $\alpha_{W_\pm}>0$ depending only on $W_\pm$ such that
\begin{equation}
\label{uniform-convergence-spectral-determinants}
\left|\mathcal{Q}_\pm^{(\alpha)}(E;\ell)-\Gamma\left(1\pm\left(\frac{1}{2}+\ell\right)\right)\left(\frac{E^{\frac{1}{2}}}{2}\right)^{\mp\left(\ell+\frac{1}{2}\right)} J_{\pm\left(\ell+\frac{1}{2}\right)}\left( E^{\frac{1}{2}} \right)\right|\lesssim_{W_\pm} \frac{\log(\alpha+1)}{\alpha+1},
\end{equation}
holds for all $\alpha\ge\alpha_{W_\pm}$ and all $(E,\ell)\in W_\pm$. Here $J_{\pm\left(\ell+\frac{1}{2}\right)}(-)$ denotes the Bessel function of the first kind of order $\pm\left(\ell+\frac{1}{2}\right)$. 

Notice that $\left(\frac{E^{\frac{1}{2}}}{2}\right)^{\mp\left(\ell+\frac{1}{2}\right)} J_{\pm\left(\ell+\frac{1}{2}\right)}\left( E^{\frac{1}{2}} \right)$ is a well-defined entire function of $E$:
\[
\left(\frac{E^{\frac{1}{2}}}{2}\right)^{\mp \left(\ell+\frac{1}{2}\right)} J_{\pm\left(\ell+\frac{1}{2}\right)}\left( E^{\frac{1}{2}} \right)= \sum_{k\ge 0} \frac{1}{k! \Gamma\left(1\pm\left(\ell+\frac{1}{2}\right)+k\right)}\left(-\frac{E}{4}\right)^k;
\]

\item[ii)] Let $W_1\subset\mathbb{C}$ and $W_{2,+}\subset \mathbb{H}_+$, $W_{2,-}\subset\mathbb{H}_+\setminus\Lambda^{(\alpha)}$ be compact subsets. If  $J_{\pm\left(\ell+\frac{1}{2}\right)}\left(E^{\frac{1}{2}}\right)\ne 0$ for all $(E,\ell)\in W_1 \times W_{2,\pm}$, then $\mathcal{Q}_{\pm}^{(\alpha)}(E;\ell)\ne 0$ for all $(E,\ell)\in W_1 \times W_{2,\pm}$ and all sufficiently big values of $\alpha$ (depending on $W_1$ and $W_{2,\pm}$);
\item[iii)]  Let $W_{2,+}\subset \mathbb{H}_+$, $W_{2,-}\subset\mathbb{H}_+\setminus\Lambda^{(\alpha)}$ be compact subsets and let $E^{(\infty)}_\pm(\ell)$ be a complex (simple) zero of $J_{\pm\left(\ell+\frac{1}{2}\right)}\left(E^{\frac{1}{2}}\right)$, if any. There exists precisely one simple zero $E_\pm^{(\alpha)}(\ell)$ of $\mathcal{Q}_{\pm}^{(\alpha)}(E;\ell)$ such  that
\begin{equation}
\label{23-aprile-2024-16}
\left|\frac{E_{\pm}^{(\alpha)}(\ell)}{E^{(\infty)}_{\pm}(\ell)}-1\right|\lesssim_{W_{2,\pm}} \frac{\log( \alpha+1)}{\alpha+1}
\end{equation}
holds for all $\ell\in W_{2,\pm}$ and all sufficiently big values of $\alpha$ (depending on $W_{2,\pm}$);

\item[iv)] Let us fix $\ell>-\frac{1}{2}$. The zeros of $\mathcal{Q}_+^{(\alpha)}(E;\ell)$ and $J_{\ell+\frac{1}{2}}\left(E^{\frac{1}{2}}\right)$ are simple, real and positive, we denote them by $E_{s,+}^{(\alpha)}(\ell)$ and $j^2_{\left(\ell+\frac{1}{2}\right),s}$, $s\in\mathbb{Z}_{\ge 0}$, respectively, and they are ordered so that $E_{s,+}^{(\alpha)}(\ell)<E_{s+1,+}^{(\alpha)}(\ell)$, and $j^2_{\left(\ell+\frac{1}{2}\right),s}<j^2_{\left(\ell+\frac{1}{2}\right),s+1}$.  

For any fixed $k\in\mathbb{Z}_{\ge 0}$, inequality
\begin{equation}
\label{1-settembre-2024-1}
\left|\frac{E_{k,+}^{(\alpha)}(\ell)}{j_{\left(\ell+\frac{1}{2}\right),k}^2}-1\right|\lesssim_{\ell,k}\frac{\log(\alpha+1)}{\alpha+1},
\end{equation}
holds for all sufficiently big values of $\alpha$ (depending on $\ell$ and $k$).

\end{itemize}
\end{theorem}
\noindent
The proof of this theorem, which is given in Section \ref{subsection-18-maggio-2024-1}, is based on the convergence of Sibuya's solution $\psi_0^{(\alpha)}(x;E,\ell)$ to the solution of the spherical rigid well with specific Cauchy data (we prove this in Section \ref{fixed-e-l}). This fact was conjectured in \cite{article:BLZ2001}.

The subsequent theorems address the asymptotic regime of large degree with energy $E$ and momentum $\ell$ scaling as
\[
E=4p^2(\alpha+1)^2\varepsilon^2,\quad \ell=2p(\alpha+1)-\frac{1}{2},
\]
for bounded complex parameters $\varepsilon$, $|\varepsilon|\ge 1$, and $p\ne 0$. When $|\varepsilon|$ is strictly greater than $1$, the spectral determinant is a rapidly oscillating function of $\varepsilon$. For real $\varepsilon,p$, for any bounded interval $I_1$ in the $\varepsilon$-line we denote by $n^{(\alpha)}_{I_1}(p)$ the number of zeros of the spectral determinant $\mathcal{Q}_+^{(\alpha)}\left(4p^2(\alpha+1)^2\varepsilon^2;2p(\alpha+1)-\frac{1}{2}\right)$ in $I_1$. Due to the rapidly oscillatory character of the spectral determinant, $n^{(\alpha)}_{I_1}(p)$ grows large in the large degree limit. It turns out that the zeros tend to accumulate in $I_1$ according to the continuous distribution $\frac{2p}{\pi}\frac{\sqrt{\varepsilon^2-1}}{\varepsilon}$. More precisely, we have the following

\begin{theorem}
\label{theorem-13-agosto-2024-1}
The following holds:
\begin{itemize}
\item[i)] Let $I\subset\mathbb{R}_{>1}\times \mathbb{R}_{>0}$ and let us fix a number $0<\Theta<\frac{\pi}{2}$. There exists a constant $\alpha_{I}>0$ depending only on $I$ (and $\Theta$) such that
\begin{equation}
\label{13-agosto-2024-2}
\begin{aligned}
& \left|\mathcal{Q}_+^{(\alpha)}\left(4p^2(\alpha+1)^2\varepsilon^2,2p(\alpha+1)-\frac{1}{2}\right)\right. \\
& \left. -\frac{2 \varepsilon^{-2p(\alpha+1)}}{\Gamma(1+2p)(\varepsilon^2-1)^{\frac{1}{4}}}\cos\left(2p(\alpha+1)\left[\sqrt{\varepsilon^2-1}-\operatorname{arctan}\sqrt{\varepsilon^2 -1}\right] -\frac{\pi}{4}\right)\right| \\ 
& \lesssim_{I,\Theta} \frac{\log(\alpha+1)}{(\alpha+1)^{\frac{1}{6}}} \left| \frac{2 \sqrt{\pi} \left[2p(\alpha+1)\right]^{\frac{1}{6}}}{\Gamma\left(1+2p\right)}\varepsilon^{-2 p (\alpha+1)} \right|
\end{aligned}
\end{equation}
holds for all $\alpha\ge\alpha_{I}$ and all $(|\varepsilon|,|p|)\in I$, $|\arg(\varepsilon)|, |\arg(p)|\le\frac{\Theta}{\alpha+1}$;
\item[ii)] Let $I_1\subset \mathbb{R}_{>1}$ and $I_2\subset\mathbb{R}_{>0}$ be  compact intervals. There exists a constant $\alpha_{1,2}>0$ depending only on $I_1$ and $I_2$ such that 
\begin{equation}
\label{16-agosto-2024-1}
\left|\frac{ n^{(\alpha)}_{I_1}(p)}{\alpha+1}-\frac{2p}{\pi}\int_{I_1}\frac{\sqrt{\varepsilon^2-1}}{\varepsilon} d\varepsilon\right|\lesssim_{I_2}\frac{1}{\alpha+1}
\end{equation}
for all $\alpha\ge\alpha_{1,2}$ and all real $\varepsilon\in I_1$, $p\in I_2$.
\end{itemize}
\end{theorem}
\noindent
The proof is given in Section \ref{high-energy-case}.

When the parameter $\varepsilon$ is close to $1$, scaling as
\[
\varepsilon=1+\frac{\eta}{(\alpha+1)^{\frac{2}{3}}},
\]
for a bounded complex parameter $\eta\ne 0$, the spectral determinant 
\[
\mathcal{Q}_+^{(\alpha)}\left(4p^2(\alpha+1)^2\left(1+\frac{\eta}{(\alpha+1)^{\frac{2}{3}}}\right)^2;2p(\alpha+1)-\frac{1}{2}\right)
\]
and its zeros are related to the Airy function $\operatorname{Ai}(-)$ and its zeros:

\begin{theorem}
\label{theorem-airy-det}
The following holds:
\begin{itemize}
\item[i)] Let $J\subset\mathbb{R}_{>0}\times \mathbb{R}_{>0}$ be a compact subset and fix a number $0<\Theta<\frac{\pi}{2}$. There exists a constant $\alpha_{J}>0$ depending only on $J$ (and $\Theta$) such that
\begin{equation}
\label{23-marzo-2024-1}
\begin{aligned}
& \left| \mathcal{Q}_+^{(\alpha)}\left(4p^2(\alpha+1)^2\left(1+\frac{\eta}{(\alpha+1)^{\frac{2}{3}}}\right)^2;2p(\alpha+1)-\frac{1}{2}\right) \right. \\
& \left.- \frac{2 \sqrt{\pi} \left[p(\alpha+1)\right]^{\frac{1}{6}}}{\Gamma\left(1+2p\right)}\left(1+\frac{\eta}{(\alpha+1)^{\frac{2}{3}}}\right)^{-2 p (\alpha+1)} \operatorname{Ai}\left(-2 p^{\frac{2}{3}}\eta\right)\right| \\
& \lesssim_{J,\Theta} \frac{\log(\alpha+1)}{(\alpha+1)^{\frac{1}{3}}} \frac{2 \sqrt{\pi} \left[2p(\alpha+1)\right]^{\frac{1}{6}}}{\Gamma\left(1+2p\right)}\left(1+\frac{\eta}{(\alpha+1)^{\frac{2}{3}}}\right)^{-2 p (\alpha+1)}
\end{aligned}
\end{equation}
for all $\alpha\ge\alpha_{J}$ and all $(|\eta|,|p|)\in J$, $|\arg(\eta)|\le\frac{\Theta}{(\alpha+1)^{\frac{1}{3}}}$, $|\arg(p)|\le\frac{\Theta}{\alpha+1}$. Here $\operatorname{Ai}(-)$ denotes the standard Airy function;

\item[ii)] Let $J_1,J_2\subset \mathbb{R}_{>0}$ be compact subsets and let $0<\Theta<\frac{\pi}{2}$. Let us denote by $a_s$, $s\in\mathbb{Z}_{\ge 0}$, the (simple and negative) zeros of $\operatorname{Ai}(-)$, which are ordered so that $|a_s|<|a_{s+1}|$. If $\frac{|a_s|}{2|p|^{\frac{2}{3}}}\notin J_1$ for all $s\in\mathbb{Z}_{\ge 0}$ and all $|p|\in J_2$, then 
\[
\mathcal{Q}_+^{(\alpha)}\left(4p^2(\alpha+1)^2\left(1+\frac{\eta}{(\alpha+1)^{\frac{2}{3}}}\right)^2;2p(\alpha+1)-\frac{1}{2}\right)\ne 0
\]
for all sufficiently big values of $\alpha$ (depending on $J_1,J_2$ and $\Theta$) and all $(|\eta|,|p|)\in J_1\times J_2$, $|\arg(\eta)|\le\frac{\Theta}{(\alpha+1)^{\frac{1}{3}}}$ , $|\arg(p)|\le\frac{\Theta}{\alpha+1}$;

\item[iii)] Let $J_2\subset\mathbb{R}_{>0}$ be a compact subset, let $p$ be real and positive, $p\in J_2$. The zeros of 
\[
\mathcal{Q}_+^{(\alpha)}\left(4p^2(\alpha+1)^2\left(1+\frac{\eta}{(\alpha+1)^{\frac{2}{3}}}\right)^2;2p(\alpha+1)-\frac{1}{2}\right)
\]
are simple, real and positive, they are denoted by $\eta_s^{(\alpha)}(p)$, $s\in\mathbb{Z}_{\ge 0}$ and they are ordered so that $\eta_{s}^{(\alpha)}(p)<\eta_{s+1}^ {(\alpha)}(p)$.

For any fixed $k\in\mathbb{Z}_{\ge 0}$, inequality
\begin{equation}
\label{2-settembre-2024-1}
\left|\frac{2p^{\frac{2}{3}}\eta^{(\alpha)}_k(p)}{|a_k|}-1\right|\lesssim_{p,k}\frac{\log(\alpha+1)}{(\alpha+1)^{\frac{1}{3}}},
\end{equation}
holds for all sufficiently big values of $\alpha$ (depending on $J_2$ and $k$) and all $|p|\in J_2$;
\end{itemize}
\end{theorem}
\noindent
The proof of this theorem is given in Section \ref{low-energy-case}.

\subsection{Organization of the paper}
In Section \ref{review} we review known results of the literature: we give more details about Sibuya's and Frobenius solutions and the characterization of the spectral determinants as Wronskians of these solutions. We also recall certain functional equations satisfied by the spectral determinants and we give a brief review of the ODE/IM correspondence for the quantum KdV model. 

Section \ref{large-alpha-fixed-e-l} is devoted to the large $\alpha$ limit with the parameters $(E,\ell)$ varying on compact subsets $W\subset\mathbb{C}^2$. In subsection \ref{large-alpha-fixed-e-l-infinity} we show that when $\alpha$ is sufficiently big, Sibuya's solutions $\psi_k^{(\alpha)}(x;E,\ell)$ are (uniformly) asymptotic to solutions to a Cauchy problem for the spherical rigid well and we compute the $(E,\ell)$-uniform, large $\alpha$ asymptotics for the Stokes multipliers $\sigma_k^{(\alpha)}(E;\ell)$ (the main results are Proposition \ref{proposition-9-aprile-2024-1} and Proposition \ref{proposition-16-maggio-2024-10}).  In subsection \ref{subsection-18-maggio-2024-1} we show that the Frobenius solutions $\chi_\pm^{(\alpha)}(x;E,\ell)$ are asymptotic to certain special functions expressed in terms of Bessel functions and we prove Theorem \ref{theorem-23-aprile-2024-1}. 

Section \ref{large-deg-en-mom} is devoted to the study of the large $\alpha$ limit with the energy and the angular momentum parameters growing large as well. In subsection \ref{subsection-all-large-infinity} we study Sibuya's solutions in the new parameterization and we obtain analogous results to the previous ones (the main results are Proposition \ref{proposition-9-aprile-2024-1-bis} and Proposition \ref{proposition-16-maggio-2024-10-bis}). In subsection \ref{subsection-3.2} we study the subdominant Frobenius solution using the complex WKB method, obtaining a uniform approximation in terms of a semiclassical WKB function (the main result is Proposition \ref{proposition-8-agosto-2024-1}). Through this approximation, we study the spectral determinant in two different regimes, which we call \textquotedblleft large rescaled energy regime \textquotedblright and \textquotedblleft small rescaled energy regime \textquotedblright\, and we prove Theorems \ref{theorem-13-agosto-2024-1} and \ref{theorem-airy-det}.

\section{Review of the functional relations for the spectral determinants}
\label{review}
In this section, following closely \cite{masoero2024qfunctions}, we give a more detailed description of Frobenius and Sibuya's solution and of the spectral determinants we defined in the Introduction.
 
Let us start with the (formal) regular singularity $x=0$\footnote{\label{footnote-12-maggio-2024-1}The point $x=0$ is formally a regular singularity because, besides its behavior $O\left(\frac{1}{x^2}\right)$ as $x\to 0$, it is a branch point of the potential, due to the fact that $\alpha$ is allowed to be any positive real number. Nevertheless, the usual Frobenius technique to obtain convergent series for pairs of dominant and subdominant solutions may be extended, as it is shown in \cite{masoero2024qfunctions}. The solutions obtained in this way are still called \textit{Frobenius solutions} and behave as $O\left(x^{\ell+1}\right)$ and $O\left(x^{-\ell}\right)$ as $x\to 0$. The exponents $\ell+1$ and $-\ell$ are still referred to as the \textit{Frobenius indeces}.}. First of all, we notice that, since the substitution $\ell\to -\ell-1$ is a symmetry for equation~\eqref{eqn}, we can assume without loss of generality that $\operatorname{Re}(\ell)>-\frac{1}{2}$ (we will not consider the case $\operatorname{Re}(\ell)=-\frac{1}{2}$). The distinguished subdominant Frobenius solution $\chi_+^{(\alpha)}(x;E,\ell)$ is defined by the series
\begin{equation}
\label{19-agosto-2024-1}
\chi_+^{(\alpha)}(x;E,\ell):= E^{-\frac{1}{2\alpha}\left(\ell+\frac{1}{2}\right)+\frac{1}{2}} x^{\ell+1}\sum_{j\ge 0} g_{j,+}\left(\xi\right)\left(\frac{x^{2 \alpha}}{E}\right)^j,
\end{equation}
where $\xi:=\frac{E x^2}{4 \alpha^2}$ and $\left\{g_{j,+}(\xi)\right\}_{j\ge 0}$ are the unique entire functions solving the recurrence
\begin{equation}
\label{19-agosto-2024-2}
\begin{aligned}
& \left(D^2_\xi +(2(\ell+j)+1)D_\xi+j(2\ell+j+1)+\xi \right)g_{j,+}=\xi g_{j-1,+},\quad D_\xi:=\frac{\xi}{\alpha}\partial_\xi \\
& g_{-1,+}(\xi)\equiv 0,\, g_{0,+}(0)=1,
\end{aligned}
\end{equation}
$j\ge 0$, for all $\ell$ with $\operatorname{Re}(\ell)>-\frac{1}{2}$. The distinguished dominant Frobenius solution $\chi_-^{(\alpha)}(x;E,\ell)$ is defined by the series 
\begin{equation}
\label{19-agosto-2024-3}
\chi_-^{(\alpha)}(x;E,\ell):= E^{\frac{1}{2\alpha}\left(\ell+\frac{1}{2}\right)-\frac{1}{2}} x^{-\ell}\sum_{j\ge 0} g_{j,-}\left(\xi\right)\left(\frac{x^{2 \alpha}}{E}\right)^j,
\end{equation}
where $\xi$ is as before and $\left\{g_{j,-}(\xi)\right\}_{j\ge 0}$ are the unique entire functions solving the recurrence
\begin{equation}
\label{19-agosto-2024-4}
\begin{aligned}
& \left(D^2_\xi +(2(j-\ell)-1)D_\xi+j(j-2\ell-1)+\xi \right)g_{j,-}=\xi g_{j-1,-},\quad D_\xi:=\frac{\xi}{\alpha}\partial_\xi \\
& g_{-1,-}(\xi)\equiv 0,\, g_{0,-}(0)=1,
\end{aligned}
\end{equation}
$j\ge 0$, for all $\ell$ with $\operatorname{Re}(\ell)>-\frac{1}{2}$ and $\ell+\frac{1}{2}\notin\Lambda^{(\alpha)}$, where $\Lambda^{(\alpha)}$ is as in~\eqref{14-maggio-2024-10}. Both $\chi_+^{(\alpha)}(x;E,\ell)$ and $\chi_-^{(\alpha)}(x;E,\ell)$ are elements of the module $\mathcal{A}_\ell^{(\alpha)}$ defined in the Introduction and they are normalized so that~\eqref{12-maggio-2024-1} and~\eqref{12-maggio-2024-2} hold. The elements $\chi_+^{(\alpha)}(x;E,\ell)$ and $\chi_-^{(\alpha)}(x;E,\ell)$ are eigenfunctions of the quantum monodromy $\mathcal{M}^{(\alpha)}$. It can be readily checked that the action of $\mathcal{M}^{(\alpha)}$ on the Frobenius solutions is
\begin{equation}
\label{13-maggio-2024-4}
\begin{aligned}
& \mathcal{M}^{(\alpha)}\chi_+^{(\alpha)}(x;E,\ell)=e^{\frac{i\pi}{\alpha+1}(\ell+1)}\chi_+^{(\alpha)}(x;E,\ell),\quad \operatorname{Re}(\ell)>-\frac{1}{2}, \\
& \mathcal{M}^{(\alpha)}\chi_-^{(\alpha)}(x;E,\ell)=e^{-\frac{i\pi}{\alpha+1}\ell}\chi_-^{(\alpha)}(x;E,\ell),\quad \operatorname{Re}(\ell)>-\frac{1}{2},\,\ell+\frac{1}{2}\notin\Lambda^{(\alpha)}.
\end{aligned}
\end{equation}
It is important to remark that $\chi_+^{(\alpha)}(x;E,\ell)$ is uniquely determined by the normalization~\eqref{12-maggio-2024-1}, while $\chi_-^{(\alpha)}(x;E,\ell)$ is determined by the normalization~\eqref{12-maggio-2024-2} together with the requirement of being an eigenvector of the quantum monodromy operator $\mathcal{M}^{(\alpha)}$ (see~\eqref{11-maggio-2024-2}). Their Wronskian is
\begin{equation}
\label{13-maggio-2024-7}
\mathcal{W}\left[\chi_-^{(\alpha)}(x;E,\ell),\chi_+^{(\alpha)}(x;E,\ell)\right]=\frac{2\alpha+2}{\pi}\sin\left(\frac{\pi}{\alpha+1}\left(\ell+\frac{1}{2}\right)\right),
\end{equation}
where the Wronskian of a differentiable function $f(x)$ with a differentiable function $g(x)$ is defined as
\begin{equation}
\label{13-maggio-2024-2}
\mathcal{W}\left[f(x),g(x)\right]:=f(x)\frac{dg}{dx}(x)-\frac{df}{dx}(x) g(x).
\end{equation}

Let us consider now the (formal) irregular singularity $x=\infty$\footnote{As pointed out in footnote \ref{footnote-12-maggio-2024-1}, the singularity at infinity is said to be formally irregular due to the fact that $\alpha$ is allowed to take any positive real value. See \cite{cotti2023asymptoticsolutionslinearodes} for a general theory of ODE's with non necessarily meromorphic coefficients.}.  Sibuya's solutions transform under the quantum monodromy operator as 
\begin{equation}
\label{13-maggio-2024-5}
\mathcal{M}^{(\alpha)}\psi_k^{(\alpha)}(x;E,\ell)=e^{-\frac{i  \pi}{2}}e^{\frac{i\pi}{2\alpha+2}}\psi_{k-1}^{(\alpha)}(x;E,\ell).
\end{equation}
The first derivatives of Sibuya's solutions are such that
\begin{equation}
\label{13-maggio-2024-11}
\lim_{\substack{x\to\infty \\ x\in S}}(-1)^{k+1}\sqrt{\frac{2\alpha+2}{\pi}} x^{-\frac{\alpha}{2}} e^{\frac{x^{\alpha+1}}{\alpha+1}e^{-i k \pi}}\frac{d}{dx}\psi_k^{(\alpha)}(x;E,\ell)=1,
\end{equation}
for any proper subsector $S$ of $\mathcal{S}_k^{(\alpha)}$. As a consequence, the Wronskian (see~\eqref{13-maggio-2024-2}) of $\psi_k^{(\alpha)}(x;E,\ell)$ with $\psi_{k+1}^{(\alpha)}(x;E,\ell)$ is
\begin{equation}
\label{13-maggio-2024-12}
\mathcal{W}\left[\psi_k^{(\alpha)}(x;E,\ell),\psi_{k+1}^{(\alpha)}(x;E,\ell)\right]=\frac{(-1)^k \pi}{\alpha+1}.
\end{equation}
The Stokes multipliers can be expressed in terms of Wronskians as
\begin{equation}
\label{13-maggio-2024-1}
\begin{aligned}
\sigma_k^{(\alpha)}(E;\ell) & =\frac{\mathcal{W}\left[\psi_{k-1}^{(\alpha)}(x;E,\ell),\psi_{k+1}^{(\alpha)}(x;E,\ell)\right]}{\mathcal{W}\left[\psi_{k}^{(\alpha)}(x;E,\ell),\psi_{k+1}^{(\alpha)}(x;E,\ell)\right]} \\ & =(-1)^k\frac{\alpha+1}{\pi} \mathcal{W}\left[\psi_{k-1}^{(\alpha)}(x;E,\ell),\psi_{k+1}^{(\alpha)}(x;E,\ell)\right],
\end{aligned}
\end{equation}
while the spectral determinants $\mathcal{Q}_\pm^{(\alpha)}(E;\ell)$ (see~\eqref{31-agosto-2024-1}) can be written as
\begin{equation}
\label{13-maggio-2024-3}
\mathcal{Q}_\pm^{(\alpha)}(E;\ell):=\mathcal{W}\left[\psi_0^{(\alpha)}(x;E,\ell),\chi_\pm^{(\alpha)}(x;E,\ell)\right].
\end{equation}

As anticipated at the beginning of the Introduction, the relevance of the family of anharmonic oscillators~\eqref{eqn} is given by the theory of quantum integrable systems. The quantum field theory that lies behind equation~\eqref{eqn} is known as quantum KdV model, which is a quantization of the second Poisson structure of the classical KdV model that leads to a two dimensional conformal field theory which is integrable by Bethe Ansatz, see \cite{article:BLZ1996} and \cite{article:BLZ1997}. Given a free field representation of the Heisenberg algebra with quasi-momentum $p$ and Planck's constant $\frac{1}{2\alpha+2}$, the space of states of this conformal field theory is the induced, generically irreducible, Virasoro module $\mathcal{V}_{c,h}=\bigoplus_{N\ge 0} \left(\mathcal{V}_{c,h}\right)_N$ with central charge $c$ and conformal dimension $h$ given by
\begin{equation}
\label{7-maggio-2024-10}
c=c^{(\alpha)}:=13-6\frac{\alpha^2+2\alpha+2}{\alpha+1},\quad h=h^{(\alpha)}(p):=\frac{(\alpha+1)^2p^2-\alpha^2}{\alpha+1}
\end{equation}
(see \cite{key1104280m} for the precise definitions). The main objects of interest of the theory are graded operator-valued functions\footnote{Here we use a different notation from \cite{article:BLZ1997}: our parameter $E$ corresponds to the parameter $\lambda^2$ of \cite{article:BLZ1997}, while the operators that we denote $\hat{\mathbf{Q}}_\pm(E)$ correspond to the operators $\lambda^{\mp 2(\alpha+1) P} \mathbf{Q}_\pm(\lambda)$ of  \cite{article:BLZ1997}, where $P$ is an operator that acts on the Heisenberg highest state $|p\rangle$ as $P|p\rangle=p|p\rangle$.}
\[
\hat{\mathbf{T}}(E),\,\hat{\mathbf{Q}}_\pm (E)\colon\, \left(\mathcal{V}_{c,h}\right)_N \to \left(\mathcal{V}_{c,h}\right)_N
\]
of an auxiliary complex variable $E$ (called the spectral parameter of the theory), which commute among themselves and whose eigenvalues, denoted here $\mathcal{T}^{(\alpha)}_j(E;p)$ and $\mathcal{Q}_{\pm,j}^{(\alpha)}(E;p)$, $j\in\mathbb{Z}_{\ge 0}$, satisfy the following relations:
\begin{equation}
\label{6-maggio-2024-1}
\begin{aligned}
\mathcal{T}_j^{(\alpha)}(E;p) \mathcal{Q}_{\pm,j}^{(\alpha)}(E;p)=e^{\pm 2 i \pi p} \mathcal{Q}_{\pm,j}^{(\alpha)}\left(E e^{\frac{2i \pi}{\alpha+1}};p\right)+e^{\mp 2 i \pi p}\mathcal{Q}_{\pm,j}\left(E e^{-\frac{2i \pi}{\alpha+1}};p\right) \\\mbox{\textquotedblleft Baxter's TQ relation\textquotedblright}
\end{aligned}
\end{equation}
and
\begin{equation}
\label{6-maggio-2024-2}
\begin{aligned}
e^{2 i \pi p} \mathcal{Q}_{+,j}^{(\alpha)}\left(E e^{\frac{i \pi}{\alpha+1}};p\right)\mathcal{Q}_{-,j}^{(\alpha)}\left(E e^{-\frac{i\pi}{\alpha+1}};p\right) - e^{-2 i \pi p}\mathcal{Q}_{+,j}^{(\alpha)}\left(E e^{-\frac{i \pi}{\alpha+1}};p\right)\mathcal{Q}_{-,j}^{(\alpha)}\left(E e^{\frac{i\pi}{\alpha+1}};p\right) \\ =2 i \sin(2 \pi p) \\
\mbox{\textquotedblleft Quantum Wronskian condition\textquotedblright}.
\end{aligned}
\end{equation}
In \cite{article:BLZ1996} it is claimed that the $\mathcal{T}_j^{(\alpha)}(E;p)$'s are entire functions of $E$, while in \cite{article:BLZ1997}, as well as in \cite{article:BLZ2003}, the authors conjecture that the $\mathcal{Q}_{\pm,j}^{(\alpha)}(E;p)$'s enjoy the same analytic property together with the normalization $\mathcal{Q}_{\pm,j}^{(\alpha)}(0;p)=1$\footnote{Entirety and normalization of the $\mathcal{Q}_{\pm,j}^{(\alpha)}(E;p)$'s are not sufficient to uniquely characterize them, but we will not enter into this problem here. See \cite{article:BLZ1997} or \cite{article:BLZ2003} for the conjectural properties the $\mathcal{Q}_{\pm,j}^{(\alpha)}(E;p)$'s should have and that should be sufficient to uniquely characterize them.} . 

As pointed out in \cite{article:BLZ1997}, the true fundamental object encoding the integrable structure of the quantum KdV model is the operator $\hat{\mathbf{Q}}_+(E)$, being $\hat{\mathbf{T}}(E)$ and $\hat{\mathbf{Q}}_-(E)$ expressible in terms of it. Indeed, any eigenvalue $\mathcal{Q}_{+,j}^{(\alpha)}(E;p)$ of $\hat{\mathbf{Q}}_+(E)$ satisfies the Bethe Ansatz equation (BAE)
\begin{equation}
\label{7-maggio-2024-1}
e^{4 i \pi p}\frac{\mathcal{Q}_{+,j}^{(\alpha)}\left(E e^{\frac{2i \pi}{\alpha+1}};p\right)}{\mathcal{Q}_{+,j}^{(\alpha)}\left(E e^{-\frac{2i \pi}{\alpha+1}};p\right)}=-1,\quad \mbox{for all }E\in\mathbb{C}\mbox{ such that }\mathcal{Q}_{+,j}^{(\alpha)}(E;p)=0
\end{equation}
and characterize the physical observables of the theory (see \cite{article:conti-masoero-2023} together with \cite{article:conti-masoero-2021} for the precise definition of the functional space on which such equations should be understood and the classification of the solutions in the large $p$ limit). 

In the seminal papers \cite{PatrickDorey1999} and \cite{DOREY1999573}, followed by \cite{article:BLZ2001}, the authors discovered that the required solution of the BAE~\eqref{7-maggio-2024-1} for the ground state eigenvalues of $\hat{\mathbf{Q}}_\pm(E)$, $\hat{\mathbf{T}}(E)$ can be represented as a suitably defined spectral determinants for equation~\eqref{eqn}  provided that the parameters $p$ and $\ell$ are related as
\begin{equation}
\label{13-maggio-2024-10}
\ell+\frac{1}{2}=2p(\alpha+1).
\end{equation}
It turned out that the spectral determinants~\eqref{13-maggio-2024-1} and~\eqref{13-maggio-2024-3} make the job. This is the first and simplest instance of the celebrated \textit{ODE/IM correspondence} (e.g. see the review \cite{Dorey2007} or \cite{FIORAVANTI2023137706} and references therein). Since its discovery, this correspondence has attracted the attention of more and more physicists and mathematicians, and has been generalized and extended to a vast variety of models (e.g. see \cite{article:BLZ2003} for the extension to the excited states of the quantum KdV model, \cite{masoero-raimondo-2020} \cite{masoero-raimondo-2023}, \cite{masoero-raimondo-valeri-2016}, \cite{masoero-raimondo-valeri-2017}, \cite{feigin-frenkel-2011}, \cite{frenkel-hernandez-2018}, \cite{frenkel-hernandez-2024}, \cite{frenkel-koroteev-2023}, \cite{ekhammar2021extendedsystemsbaxterqfunctions}, \cite{kotousov-2024} for generalised quantum Drinfeld-Sokolov models, \cite{10.1007/978-3-030-57000-25} for the special case of quantum Boussinesq model, \cite{kotousov-lukyanov-2023} for affine $\mathfrak{sl}(2)$ Gaudin model, \cite{article:kotousov2018} for the $O(3)$ non-linear sigma model, and references therein), suggesting the existence of a deep duality between the theory of (quantum) integrable systems (see e.g. \cite{procházka2024walgebrasintegrability} for a review of recent achievements in quantum integrability) and the theory of linear differential operators (see also \cite{aramini-2023} for examples of deformations of the correspondence). This strongly motivates our study of equation~\eqref{eqn}. The large $\alpha$ regime in which we perform our analysis is motivated by a remark in \cite{article:BLZ2001}, in which the authors point out that the limit $\alpha\to+\infty$ corresponds, through the parameterization~\eqref{7-maggio-2024-10} of the central charge, to the limit $c\to -\infty$, which in turn corresponds to the classical limit in conformal field theory, which is known to have at least some formal resemblance to the second Hamiltonian structure of the classical KdV theory (see e.g. \cite{Gervais:1985fc}, \cite{difrancesco}, \cite{babelon} and references therein, and \cite{dymarky-2018}, \cite{dymarky-2022} for recent works on the semiclassical limit of the quantum KdV charges).

The correspondence with the spectral determinants~\eqref{13-maggio-2024-1} and~\eqref{13-maggio-2024-3} is realized as follows. Starting from equation~\eqref{31-agosto-2024-2} with $k=0$, taking the Wronskian $\mathcal{W}\left[-,\chi_+^{(\alpha)}(x;E,\ell)\right]$ of both sides and using~\eqref{13-maggio-2024-4} and~\eqref{13-maggio-2024-5} specialized to $k=0$ and $k=1$, we obtain
\begin{equation}
\label{13-maggio-2024-6}
i\sigma_0^{(\alpha)}(E;\ell) \mathcal{Q}_\pm^{(\alpha)}(E;\ell)=e^{-\frac{i\pi}{\alpha+1}\left(\ell+\frac{1}{2}\right)}\mathcal{Q}_\pm^{(\alpha)}\left(E e^{-\frac{2\pi i}{\alpha+1}};\ell\right)+ e^{\frac{i\pi}{\alpha+1}\left(\ell+\frac{1}{2}\right)}\mathcal{Q}_\pm^{(\alpha)}\left(E e^{\frac{2i\pi}{\alpha+1}};\ell\right)
\end{equation} 
(the \textquotedblleft$+$\textquotedblright\, case for all $\ell$ with $\operatorname{Re}(\ell)>-\frac{1}{2}$, the \textquotedblleft$-$\textquotedblright\, case for $\ell$ with $\operatorname{Re}(\ell)>-\frac{1}{2}$ and $\ell+\frac{1}{2}\notin\Lambda^{(\alpha)}$). 

Using~\eqref{13-maggio-2024-7} we can write
\begin{equation}
\label{14-maggio-2024-1}
\begin{aligned}
& \mathcal{W}\left[\chi_-^{(\alpha)}(x;E,\ell),\chi_+^{(\alpha)}(x;E,\ell)\right]\psi_0^{(\alpha)}(x;E,\ell)\\ & = \mathcal{Q}_+^{(\alpha)}(E;\ell)\chi_-^{(\alpha)}(x;E,\ell)-\mathcal{Q}_-^{(\alpha)}(E;\ell)\chi_+^{(\alpha)}(x;E,\ell).
\end{aligned}
\end{equation}
Using~\eqref{13-maggio-2024-7}, tanking the Wronskian $\mathcal{W}\left[-,\psi_1^{(\alpha)}(x;E,\ell)\right]$ on both sides of~\eqref{14-maggio-2024-1}, using~\eqref{13-maggio-2024-4},~\eqref{13-maggio-2024-5} specialized to $k=0$ and $k=1$, together with~\eqref{13-maggio-2024-12} specialized to $k=0$, we obtain
\begin{equation}
\label{14-maggio-2024-2}
\begin{aligned}
& e^{\frac{i\pi}{\alpha+1}\left(\ell+\frac{1}{2}\right)} \mathcal{Q}_+^{(\alpha)}\left(E e^{\frac{i\pi}{\alpha+1}};\ell\right) \mathcal{Q}_-^{(\alpha)}\left(E e^{-\frac{i\pi}{\alpha+1}};\ell\right)-e^{-\frac{i\pi}{\alpha+1}\left(\ell+\frac{1}{2}\right)} \mathcal{Q}_-^{(\alpha)}\left(E e^{\frac{i \pi}{\alpha+1}};\ell\right) \mathcal{Q}_+^{(\alpha)}\left(E e^{-\frac{i\pi}{\alpha+1}};\ell\right) \\ &  =2 i \sin\left(\frac{\pi}{\alpha+1}\left(\ell+\frac{1}{2}\right)\right),
\end{aligned}
\end{equation}
for $\ell+\frac{1}{2}\notin\Lambda^{(\alpha)}$. 

Comparing~\eqref{13-maggio-2024-6} with~\eqref{6-maggio-2024-1} and~\eqref{14-maggio-2024-2} with~\eqref{6-maggio-2024-2} we can read the correspondence:
\begin{equation}
\label{14-maggio-2024-11}
\begin{array}{lll}
\mathcal{Q}_\pm^{(\alpha)}(E;\ell) & \longleftrightarrow & \mbox{ground state eigenvalue of }\hat{\mathbf{Q}}_\pm(E) \\
i \sigma_0^{(\alpha)}(E;\ell) & \longleftrightarrow & \mbox{ground state eigenvalue of }\hat{\mathbf{T}}(E) 
\end{array}
\end{equation}
provided $\ell$ is as in~\eqref{13-maggio-2024-10}\footnote{In order to obtain the functional equations~\eqref{13-maggio-2024-6} and~\eqref{14-maggio-2024-2} a prominent role is played by the action of the quantum monodromy $\mathcal{M}^{(\alpha)}$ defined in~\eqref{11-maggio-2024-2}. This transformation was first used by P-F- Hsieh and Y. Sibuya in \cite{HSIEH196684} as a procedure to obtain new solutions from a given one. P. Dorey and R. Tateo (see \cite{PatrickDorey1999}, \cite{DOREY1999573}) used it in order to find the distinguished solutions satisfying the covariance properties~\eqref{13-maggio-2024-4} and~\eqref{13-maggio-2024-5} so that they could be employed to define the correct spectral determinants to exploit the correspondence. For this reason, in the ODE/IM correspondence literature,~\eqref{13-maggio-2024-4} and~\eqref{13-maggio-2024-5} are sometimes called \textit{Dorey-Tateo symmetry}.}.
\\

This correspondence is the main motivation for the study of this paper. Via the ODE/IM correspondence, equation~\eqref{eqn} in the large degree, energy and angular momentum regime is relevant for the classical limit of the ground state of the quantum KdV model, while the large degree regime with fixed energy and angular momentum has a twofold relevance: it is related to the classical limit of the quantum KdV model in the small quasi momentum $p$ regime (with $p$ scaling as $O\left(\frac{1}{\alpha+1}\right)$) and it is related also to the ODE/IM correspondence for a scaling limit of the lattice six vertex model or the scaling limit of the Heisenberg XXZ model (see \cite{Dorey2007}, \cite{BAZHANOV2021115337} for more details). Remarkably, as pointed out in \cite{article:BLZ2001},  the case with fixed $(E,\ell)$ has also some relevance for the description of certain quantum Brownian particles (see \cite{PhysRevLett.46.211} and \cite{BAZHANOV1999529}).

\section{Large degree limit with fixed energy and angular momentum}
\label{large-alpha-fixed-e-l}

In this section we study equation~\eqref{eqn} in the limit $\alpha\to+\infty$ with $(E,\ell)$ varying in a compact subset of $\mathbb{C}^2$.

By looking at the potential $V^{(\alpha)}(x;E,\ell)$ (see~\eqref{15-maggio-2024-1}) of the differential equation, we expect that Sibuya's solutions converge, in a precise meaning to be explained in the subsequent sections, to some solution to the spherical rigid well as $\alpha$ grows large, namely to some solution to equation
\begin{equation}
\label{spherical-potential-well}
\begin{aligned}
& \frac{d^2\psi^{(\infty)}}{dx^2}=V^{(\infty)}(x;E,\ell)\psi^{(\infty)},\quad V^{(\infty)}(x;E,\ell):=\frac{\ell(\ell+1)}{x^2}-E,\quad 0<|x|\le1 \\
& \left( \psi^{(\infty)}(x)\equiv 0,\quad |x|\ge 1 \right),
\end{aligned}
\end{equation}
with some boundary conditions to be determined. Analogously, we expect Frobenius solutions to be asymptotic to a pair of dominant and subdominant solutions to~\eqref{spherical-potential-well}.

\subsection{Local theory at the (formal) irregular singularity - Sibuya's solutions}
\label{large-alpha-fixed-e-l-infinity}

Since as $|x|$ grows large the dominant term of the potential $V^{(\alpha)}(x;E,\ell)$ is $x^{2\alpha}$, we expect that dominant/subdominant solutions at $x=\infty$ can be (uniformly) approximated by dominant/subdominant solutions to equation
\begin{equation}
\label{modified-bessel}
\frac{d^2\psi}{dx^2}=x^{2\alpha}\psi,
\end{equation}
which can be expressed in terms of modified Bessel functions of order $\frac{1}{2\alpha+2}$ (see equation~\eqref{modified-bessel-equation-second-form} of Appendix \ref{appendix-14-dez-2023}). Let us introduce some notations: for any integer $k\in\mathbb{Z}$ and any $\alpha>0$, we denote by $\Psi_k^{(\alpha)}(x)$, $\Xi_k^{(\alpha)}(x)$ and $\widetilde{\Xi}_k^{(\alpha)}(x)$ the solutions to~\eqref{modified-bessel} defined by
\begin{equation}
\label{14-dez-2023-1}
\Psi_{k}^{(\alpha)}(x):= \frac{e^{-\frac{ik\pi}{2}}x^{\frac{1}{2}}}{\alpha+1} K_{\frac{1}{2\alpha+2}}\left(\frac{x^{\alpha+1}}{\alpha+1} e^{-i k \pi}\right)
\end{equation}
and
\begin{equation}
\label{14-dez-2023-2}
\Xi_k^{(\alpha)}(x):= e^{\frac{i k \pi}{2}}x^{\frac{1}{2}}\left[ I_{\frac{1}{2\alpha+2}}\left(\frac{x^{\alpha+1}}{\alpha+1} e^{-ik\pi}\right)+\frac{i}{\pi}e^{-\frac{i\pi}{2\alpha+2}} K_{\frac{1}{2\alpha+2}}\left(\frac{x^{\alpha+1}}{\alpha+1} e^{-ik\pi}\right) \right] ,
\end{equation}
\begin{equation}
\label{14-dez-2023-2-bis}
\widetilde{\Xi}_k^{(\alpha)}(x):=e^{\frac{i k \pi}{2}}x^{\frac{1}{2}}\left[ I_{\frac{1}{2\alpha+2}}\left(\frac{x^{\alpha+1}}{\alpha+1} e^{-ik\pi}\right)-\frac{i}{\pi}e^{\frac{i\pi}{2\alpha+2}} K_{\frac{1}{2\alpha+2}}\left(\frac{x^{\alpha+1}}{\alpha+1} e^{-ik\pi}\right) \right] ,
\end{equation}
where $K_{\frac{1}{2\alpha+2}}(-)$ and $I_{\frac{1}{2\alpha+2}}(-)$ denote the standard modified Bessel functions of the third and first kind, respectively, of order $\frac{1}{2\alpha+2}$ (see Appendix \ref{appendix-14-dez-2023}).
By means of the analytic continuation formulas~\eqref{8-gennaio-2024-2} and~\eqref{8-gennaio-2024-3} and the representations~\eqref{8-gennaio-2024-4} and~\eqref{8-gennaio-2024-5} of Appendix \ref{appendix-14-dez-2023} the following limits are readily proved:
\[
\lim_{\substack{x\to \infty \\ x\in\Sigma_k^{(\alpha)}}} \Psi_k^{(\alpha)}(x)=0
\]
and
\[
\lim_{\substack{x\to \infty \\ x\in\Sigma_k^{(\alpha)}}} \Xi_k^{(\alpha)}(x)=\lim_{\substack{x\to \infty \\ x\in\Sigma_k^{(\alpha)}}} \widetilde{\Xi}_k^{(\alpha)}(x)=\infty,
\]
where the limits are nontangential and $\Sigma_k^{(\alpha)}$ is the sector~\eqref{3-aprile-2024-2-bis}. The normalization constants in~\eqref{14-dez-2023-1},~\eqref{14-dez-2023-2} and~\eqref{14-dez-2023-2-bis} are chosen so that
\[
\mathcal{W}\left[\Psi_k^{(\alpha)}(x),\Xi_k^{(\alpha)}(x)\right]=\mathcal{W}\left[\Psi_k^{(\alpha)}(x),\widetilde{\Xi}_k^{(\alpha)}(x)\right]=1,
\]
where $\mathcal{W}[-,-]$ denotes the Wronskian (see~\eqref{13-maggio-2024-2}), and
\[
\lim_{\substack{x\to \infty \\ x\in \Sigma_k^{(\alpha)}}} \frac{\Psi_{k-1}^{(\alpha)}(x)}{\Psi_{k+1}^{(\alpha)}(x)}=1,
\]
the limit being nontangential. 
\begin{proposition}
\label{proposition-3-aprile-2024-1}
Let $W\subset\mathbb{C}^2$ be a compact subset and let $\psi_k^{(\alpha)}(x;E,\ell)$ be the $k$-th Sibuya's solution to~\eqref{eqn} defined by~\eqref{3-aprile-2024-1}. There exists a constant $\alpha_W>0$ depending only on $W$ and a number $r_0>1$ (independent of all the parameters) such that inequality
\begin{equation}
\label{3-aprile-2024-3}
\left|\frac{\psi_k^{(\alpha)}(x;E,\ell)}{\Psi_k^{(\alpha)}(x)}-1\right|\lesssim_{W}\frac{1}{(\alpha-1)r_0^{\alpha-1}}
\end{equation}
holds for all $x$ in the closed annular sector
\begin{equation}
\label{12-agosto-2024-1}
\left\{\frac{(k-1)\pi}{\alpha+1}\le \arg(x)\le\frac{(k+1)\pi}{\alpha+1},\,|x|\ge 1\right\}\subset\mathcal{S}_k^{(\alpha)},
\end{equation}
$\alpha\ge\alpha_W$ and $(E,\ell)\in W$ (here $\Psi_k^{(\alpha)}(x)$ is the function~\eqref{14-dez-2023-1}, while $\mathcal{S}_k^{(\alpha)}$ is the sector~\eqref{3-aprile-2024-2}).
\end{proposition}

\begin{proof}
From Proposition \ref{theorem-10-gennaio-2024-1} it follows that the function $\Psi_k^{(\alpha)}(x)$ defined in~\eqref{14-dez-2023-1} has no zeros in the whole sector $\left\{\frac{(k-1)\pi}{\alpha+1}\le\arg(x)\le\frac{(k+1)\pi}{\alpha+1}\right\}$, thus the ratio 
\begin{equation}
\label{11-agosto-2024-1}
Z^{(\alpha)}_k(x;E,\ell):=\frac{\psi_k^{(\alpha)}(x;E,\ell)}{\Psi_k^{(\alpha)}(x)}
\end{equation} 
is well-defined and we can check by direct inspection that it satisfies the Volterra integral equation
\begin{equation}
\label{10-gennaio-2024-1}
Z_k^{(\alpha)}(x;E,\ell)=1+\int_{\gamma_k,\infty}^{x} \frac{\Psi_k^{(\alpha)}(t)}{\Psi_k^{(\alpha)}(x)}  G^{(\alpha)}(x,t) F(t;E,\ell)Z_k^{(\alpha)}(t;E,\ell) dt,
\end{equation}
where the kernel $G^{(\alpha)}(x,t)$ is defined as
\begin{equation}
\label{8-gennaio-2024-1}
\begin{aligned}
G^{(\alpha)}(x,t):&=\Psi_k^{(\alpha)}(t) \Xi_k^{(\alpha)}(x)-\Psi_k^{(\alpha)}(x) \Xi_k^{(\alpha)}(t) \\
&= \Psi_k^{(\alpha)}(t) \widetilde{\Xi}_k^{(\alpha)}(x)-\Psi_k^{(\alpha)}(x) \widetilde{\Xi}_k^{(\alpha)}(t)
\end{aligned}
\end{equation}
(see~\eqref{14-dez-2023-1},~\eqref{14-dez-2023-2} and~\eqref{14-dez-2023-2-bis} for the notations), the forcing term is $F(t;E,\ell):=\frac{\ell(\ell+1)}{t^2}-E$, and $\gamma_k=\gamma_k(s)$, $s\in(0,1)$, is an oriented path joining $\infty$ to $x$ such that
\begin{equation}
\label{4-aprile-2024-7}
\lim_{s\to 0^+}\left|\gamma_k(s)\right|=\infty,\quad\left|\lim_{s\to 0^+}\arg(\gamma_k(s))-\frac{k\pi}{\alpha+1}\right|<\frac{\pi}{\alpha+1}.
\end{equation}
We start proving the statement for points $x$ with $|x|\ge 2$ and $\frac{2k-1}{\alpha+1}\frac{\pi}{2}<\arg(x)\le \frac{\pi k}{\alpha+1}$. In this case we can choose the path $\gamma_k$ to be the half line
\[
-\left\{r e^{i \arg(x)},\,|x|\le r<+\infty\right\}.
\] 
We write the kernel $G^{(\alpha)}(x,t)$ as
\[
\begin{aligned}
& \frac{\alpha+1}{(t x)^{\frac{1}{2}}}G^{(\alpha)}(x,t)= \\ & K_{\frac{1}{2\alpha+2}}\left(\frac{t^{\alpha+1}}{\alpha+1}e^{-i k \pi}\right)I_{\frac{1}{2\alpha+2}}\left(\frac{x^{\alpha+1}}{\alpha+1}e^{-i k \pi}\right)-K_{\frac{1}{2\alpha+2}}\left(\frac{x^{\alpha+1}}{\alpha+1}e^{-i k \pi}\right)I_{\frac{1}{2\alpha+2}}\left(\frac{t^{\alpha+1}}{\alpha+1}e^{-i k \pi}\right);
\end{aligned}
\]
since by~\eqref{15-aprile-2024-1} and Proposition \ref{proposition-15-aprile-2024-1} the function $I_{\frac{1}{2\alpha+2}}\left(\frac{x^{\alpha+1}}{\alpha+1}\right)$ has no zeros along $\gamma_k$, by making use of representations~\eqref{8-gennaio-2024-4} and~\eqref{8-gennaio-2024-5} (with the choice of the \textquotedblleft$-$\textquotedblright\,sign) with the corresponding bounds for the remainder terms~\eqref{8-gennaio-2024-7} and~\eqref{8-gennaio-2024-8} we find
\begin{equation}
\label{11-agosto-2024-2}
\left|\frac{\Psi_k^{(\alpha)}(x)}{\Psi_k^{(\alpha)}(t)} G^{(\alpha)}(x,t)\right|\lesssim \left|\Psi_k^{(\alpha)}(t) t^{\frac{1}{2}} I_{\frac{1}{2\alpha+2}}\left(\frac{t^{\alpha+1}}{\alpha+1}\right)\right|,\quad x,t\in\gamma_k,\,2\le|x|\le |t|.
\end{equation}
As a consequence, we have
\[
\left|\int_{\gamma_k,\infty}^x\frac{\Psi_k^{(\alpha)}(t)}{\Psi_k^{(\alpha)}(x)}G^{(\alpha)}(x,t)F(t;E,\ell) dt\right|\lesssim_W \frac{1}{(\alpha-1)2^{\alpha-1}}
\]
for all $\alpha>1$; using these results together with the integral relation~\eqref{10-gennaio-2024-1} we can find a constant $\alpha_W>1$ depending only on $W$ such that
\[
\sup_{\alpha\ge\alpha_W} \sup_{(E,\ell)\in W} \sup_{\substack{|x|\ge 2 \\ \frac{2k-1}{\alpha+1}\frac{\pi}{2}<\arg(x)\le \frac{k\pi}{\alpha+1}}} \left|Z_k^{(\alpha)}(x;E,\ell)\right|\le 2
\]
and, using again equation~\eqref{10-gennaio-2024-1}, inequality~\eqref{3-aprile-2024-3} follows.

For points $x$ with $1\le|x|\le 2$ and $\frac{2k-1}{\alpha+1}\frac{\pi}{2}\le \arg(x)\le \frac{k\pi}{\alpha+1}$, we choose again $\gamma_k$ to be the oriented half line joining $\infty$ to $x$. Inequality~\eqref{11-agosto-2024-2} still holds and we have
\[
\begin{aligned}
& \left|\int_{\gamma_k,\infty}^x\frac{\Psi_k^{(\alpha)}(t)}{\Psi_k^{(\alpha)}(x)}G^{(\alpha)}(x,t)F(t;E,\ell) dt\right| \lesssim_W \\
& \int_{1}^{2} \left|\Psi_k^{(\alpha)}(t) t^{\frac{1}{2}} I_{\frac{1}{2\alpha+2}}\left(\frac{t^{\alpha+1}}{\alpha+1}\right)\right| d|t|+\frac{1}{(\alpha-1) 2^{\alpha-1}}.
\end{aligned}
\]
Applying the mean value theorem for definite integrals on the right hand side of the last inequality and the representations of the modified Bessel functions we have used before, we receive a constant $1<\eta<2$ (independent of $\alpha$ and $W$) such that
\begin{equation}
\label{11-agosto-2024-10}
 \left|\int_{\gamma_k,\infty}^x\frac{\Psi_k^{(\alpha)}(t)}{\Psi_k^{(\alpha)}(x)}G^{(\alpha)}(x,t)F(t;E,\ell) dt\right| \lesssim_W \frac{1}{(\alpha-1) \eta^{\alpha-1}}
\end{equation}
and the conclusion follows from the integral equation~\eqref{10-gennaio-2024-1}.

For points $x$ with $|x|\ge 1$ and $\frac{k\pi}{\alpha+1} \le \arg(x)< \frac{2k+1}{\alpha+1}\frac{\pi}{2}$ the proof is similar (the only difference consists in choosing representation~\eqref{8-gennaio-2024-5} for the modified Bessel function of the first kind with the choice of the \textquotedblleft$+$\textquotedblright\, sign).

For points $x$ with $|x|\ge 2$ and $\frac{k-1}{\alpha+1}\le \arg(x)\le \frac{2k-1}{\alpha+1}\frac{\pi}{2}$ we can take $\gamma_k:=\gamma_k^{(1)}*\gamma_k^{(2)}$, where
\begin{equation}
\label{11-agosto-2024-3}
\begin{aligned}
& \gamma_k^{(1)}:=-\left\{r e^{\frac{i k \pi}{\alpha+1}},\, |x|\le r <+\infty\right\} \\[2ex]
& \gamma_k^{(2)}:=-\left\{|x| e^{i\phi},\,\arg(x)\le \phi \le \frac{k\pi}{\alpha+1}\right\},
\end{aligned}
\end{equation}
see Figure \ref{figure-11-agosto-2024-1} below.
\begin{figure}[H]
\centering
\begin{tikzpicture}[decoration={markings, mark= at position 0.5 with {\arrow{stealth}}},scale=0.35, every node/.style={scale=0.65}] 

\filldraw[color=green!90, fill=green!10, dashed] (-11.33:4) arc[start angle=-11.33, end angle=0, radius=4] -- (19,0) -- (0:19) arc[start angle=0, end angle=-11.33, radius=19] -- (-11.33:4);

\draw (0,0)  node[left] {$0$};

\fill (-8.7:10) circle[radius=3pt];

\draw (23,0) node[right] {\mbox{\LARGE $\mathcal{S}_k^{(\alpha)}$}};

\draw (-8.7:10) node[right] {$x$};
\draw (-4.35:10) node[left] {$\gamma_k^{(2)}$};
\draw (0:14.5) node[above] {$\gamma_k^{(1)}$};

\draw[dashed] (0,0) -- (17:19);
\draw[dashed] (0,0) -- (-17:19);

\draw (19,0) node[right] {$\arg(t)=\frac{k\pi}{\alpha+1}$};

\draw[color=green!40!black!70, postaction=decorate] (0:19) -- (0:10);

\draw[color=green!40!black!70, postaction=decorate] (0:10) arc[start angle=0, end angle=-8.7, radius=10];

\end{tikzpicture}
\caption{\small The paths $\gamma_k^{(1)}$ and $\gamma_k^{(2)}$ (dark green lines) defined in~\eqref{11-agosto-2024-3}. The green area represents the sector $\left\{|t|\ge 2,\,\frac{(k-1)\pi}{\alpha+1}\le\arg(t)\le\frac{k\pi}{\alpha+1}\right\}$.}
\label{figure-11-agosto-2024-1}
\end{figure}
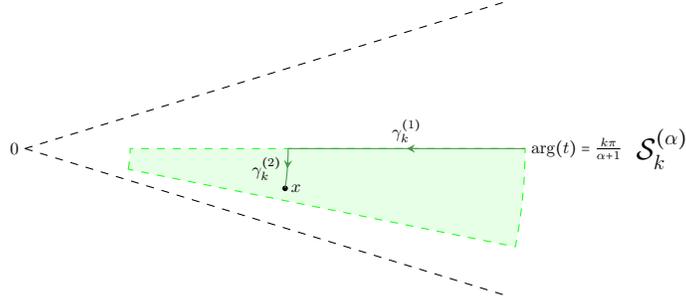
In this case, we write the kernel $G^{(\alpha)}(x,t)$ as
\[
G^{(\alpha)}(x,t)=\Psi_k^{(\alpha)}(t)\Xi_k^{(\alpha)}(x)-\Psi_k^{(\alpha)}(x)\Xi_k^{(\alpha)}(t);
\]
from the connection formulas~\eqref{15-aprile-2024-2},~\eqref{15-aprile-2024-3} (first line) and from Proposition \ref{proposition-15-aprile-2024-1} point ii), we see that $\Xi_k^{(\alpha)}(t)$ does not have zeros for $\frac{(k-1)\pi}{\alpha+1}\le \arg(t)\le \frac{k\pi}{\alpha+1}$ an we have
\begin{equation}
\label{11-agosto-2024-4}
\left|\frac{\Psi_k^{(\alpha)}(t)}{\Psi_k^{(\alpha)}(x)} G^{(\alpha)}(x,t)\right| \lesssim \left|\Psi_k^{(\alpha)}(t)\Xi_k^{(\alpha)}(t)\right|,\quad x,t\in \gamma_k^{(2)},\,\arg(x)\le \arg(t).
\end{equation}
Writing
\[
\begin{aligned}
& \int_{\gamma_k,\infty}^x \frac{\Psi_k^{(\alpha)}(t)}{\Psi_k^{(\alpha)}(x)} G^{(\alpha)}(x,t) F(t;E,\ell) dt= \\ & \left(\int_{\gamma_k^{(1)}}dt+\int_{\gamma_k^{(2)}}dt\right) \frac{\Psi_k^{(\alpha)}(t)}{\Psi_k^{(\alpha)}(x)} G^{(\alpha)}(x,t) F(t;E,\ell),
\end{aligned}
\]
we follow the same procedure as before and we get the result. 

For points $x$ with $1\le |x|\le 2$ and $\frac{(k-1)\pi}{\alpha+1}\le \arg(x)\le \frac{2k-1}{\alpha+1}\frac{\pi}{2}$ inequality~\eqref{11-agosto-2024-4} still holds and we proceed in the same way of the previous analogous case: we choose $\gamma_k$ to be the oriented path consisting of an half line joining $\infty$ to $|x|e^{i\frac{k\pi}{\alpha+1}}$ and an arc joining $|x|e^{i\frac{k\pi}{\alpha+1}}$ to $x$, we split the integral along $\gamma_k$ of the function $\frac{\Psi_k^{(\alpha)}(t)}{\Psi_k^{(\alpha)}(x)} G^{(\alpha)}(x,t) F(t;E,\ell)$ into an integral along the piece of $\gamma_k$ with $|t|\ge 2$ and an integral along the piece of $\gamma_k$ with $1\le|t|\le 2$, we apply the mean value theorem to obtain an inequality similar to~\eqref{11-agosto-2024-10} and we use the integral equation~\eqref{10-gennaio-2024-1} to get the result.

For points $|x|\ge 1$ with $\frac{k\pi}{\alpha+1}\le \arg(x)\le \frac{(k+1)\pi}{\alpha+1} $ we proceed in a similar way (the only difference consists in writing the kernel in terms of $\Psi_k^{(\alpha)}$ and $\widetilde{\Xi}_k^{(\alpha)}$).

\end{proof}

Before going to the study of the Stokes phenomenon, we want to remark that by making use of the Volterra integral equation~\eqref{10-gennaio-2024-1} we have considered in the proof of Proposition \ref{proposition-3-aprile-2024-1} we can also check that Sibuya's solutions $\psi_k^{(\alpha)}(x;E,\ell)$ enjoy the Dorey-Tateo symmetry~\eqref{13-maggio-2024-5}:

\begin{proposition}
\label{proposition-15-maggio-2024-1}
With the same assumptions and notations of Proposition \ref{proposition-3-aprile-2024-1}, the solution $\psi_k^{(\alpha)}(x;E,\ell)$ satisfy relation~\eqref{13-maggio-2024-5} for all $k\in\mathbb{Z}$.
\end{proposition}
\begin{proof}
Let $x$ be a point in the sector $\Sigma_k^{(\alpha)}$ (see~\eqref{3-aprile-2024-2-bis}) with $|x|\ge 1$, let $\gamma_k$ be the oriented half line joining $\infty$ to $x$ and let us consider the integral relation~\eqref{10-gennaio-2024-1}.  Notice that the kernel $G^{(\alpha)}(x,t)$ defined in~\eqref{8-gennaio-2024-1} is independent of $k\in\mathbb{Z}$, in particular
\begin{equation}
\label{14-dez-2023-4}
G^{(\alpha)}(x,t)=\frac{(tx)^{\frac{1}{2}}}{\alpha+1}\left[K_{\frac{1}{2\alpha+2}}\left(\frac{t^{\alpha+1}}{\alpha+1}\right) I_{\frac{1}{2\alpha+2}}\left(\frac{x^{\alpha+1}}{\alpha+1}\right)-K_{\frac{1}{2\alpha+2}}\left(\frac{x^{\alpha+1}}{\alpha+1}\right) I_{\frac{1}{2\alpha+2}}\left(\frac{t^{\alpha+1}}{\alpha+1}\right)\right],
\end{equation}
and notice also that
\begin{equation}
\label{16-gennaio-2024-3}
G^{(\alpha)}\left(xe^{\frac{i m \pi}{\alpha+1}},te^{\frac{i m \pi}{\alpha+1}}\right)=e^{\frac{i m \pi}{\alpha+1}}G^{(\alpha)}(x,t),
\end{equation}
for all $m\in\mathbb{Z}$ (formulas~\eqref{14-dez-2023-4} and~\eqref{16-gennaio-2024-3} are obtained by making use of the analytic continuation formulas~\eqref{8-gennaio-2024-2} and~\eqref{8-gennaio-2024-3} of Appendix \ref{appendix-14-dez-2023}). Using the covariance property~\eqref{16-gennaio-2024-3} of $G^{(\alpha)}(x,t)$ and by making the change of variable $t\mapsto t e^{\frac{i m \pi}{\alpha+1}}$ into equation~\eqref{10-gennaio-2024-1}, we find
\[
\begin{aligned}
& Z_k^{(\alpha)}\left(x e^{\frac{i m \pi}{\alpha+1}};E,\ell\right)= \\ & 1+\int_{ e^{-\frac{i m \pi}{\alpha+1}} \gamma_k,\infty}^x \frac{\Psi^{(\alpha)}_{k-m}(t)}{\Psi_{k-m}^{(\alpha)}(x)}G^{(\alpha)}(x,t) F(t;E e^{\frac{2 i m \pi}{\alpha+1}},\ell) Z_{k}^{(\alpha)}\left(t e^{\frac{i m \pi}{\alpha+1}};E,\ell\right) dt,
\end{aligned}
\]
where $e^{-\frac{i m \pi}{\alpha+1}}\gamma_k$ is the path $\gamma_k$ rotated by $e^{-\frac{i m \pi}{\alpha+1}}$. This equation is satisfied also by $Z_{k-m}^{(\alpha)}\left(x;E e^{\frac{ 2 i m \pi}{\alpha+1}},\ell\right)$, hence, by uniqueness and analyticity of the solution, it follows that
\[
Z_{k}^{(\alpha)}\left(x e^{\frac{i m \pi}{\alpha+1}}; E,\ell\right)=Z_{k-m}^{(\alpha)}(x;E e^{\frac{2 i m \pi}{\alpha+1}},\ell),
\]
and thus
\begin{equation}
\label{4-aprile-2024-8}
\psi_k^{(\alpha)}\left(x e^{\frac{i m \pi}{\alpha+1}};E,\ell\right)=e^{-\frac{i m \pi}{2}}e^{\frac{i m \pi}{2 \alpha+2}}\psi^{(\alpha)}_{k-m}\left(x;E e^{\frac{2 i m \pi}{\alpha+1}},\ell\right)
\end{equation}
(recall that $Z_k^{(\alpha)}(x;E,\ell)=\frac{\psi_k^{(\alpha)}(x;E,\ell)}{\Psi_k^{(\alpha)}(x)}$) can be analytically continued on the whole universal covering $\widetilde{\mathbb{C}^*}$ of the puctured $x$-plane. 
\end{proof}

\subsubsection{Convergence to the spherical rigid well and Stokes multipliers}
\label{fixed-e-l}

In this section we show that for each $k\in\mathbb{Z}$ Sibuya's solution $\psi^{(\alpha)}_k(x;E,\ell)$ is asymptotic, in an annular sector whose interior has non-empty intersection with the curve $\left\{|x|=1\right\}\subset\widetilde{\mathbb{C}^*}$ and for large values of $\alpha$, to the solution $\widetilde{\psi}_k^{(\alpha)}(x;E,\ell)$ of the spherical rigid well~\eqref{spherical-potential-well} with Cauchy data
\begin{equation}
\label{9-aprile-2024-1}
\widetilde{\psi}_k^{(\alpha)}\left(e^{\frac{ik\pi}{\alpha+1}};E,\ell\right)=0,\quad\left.\frac{d}{dx}\widetilde{\psi}_k^{(\alpha)}(x;E,\ell)\right|_{x=e^{\frac{ik\pi}{\alpha+1}}}=-e^{-\frac{i k \pi}{2}}.
\end{equation}

\begin{remark}
\label{remark-22-aprile-2024-1}
The functions $\widetilde{\psi}_k^{(\alpha)}(x;E,\ell)$ can be written explicitly as
\begin{equation}
\label{22-aprile-2024-1}
\begin{aligned}
& \widetilde{\psi}_k^{(\alpha)}(x;E,\ell)=\frac{e^{-\frac{i k \pi}{2}}\pi}{2} e^{\frac{i k \pi}{2 \alpha+2}} x^{\frac{1}{2}} \times \\ & \times \left[ Y_{\ell+\frac{1}{2}}\left(E^{\frac{1}{2}} e^{\frac{i k \pi}{\alpha+1}}\right) J_{\ell+\frac{1}{2}}\left(E^{\frac{1}{2}} x\right)- J_{\ell+\frac{1}{2}}\left(E^{\frac{1}{2}} e^{\frac{i k \pi}{\alpha+1}}\right) Y_{\ell+\frac{1}{2}}\left(E^{\frac{1}{2}}x\right) \right],
\end{aligned}
\end{equation}
where $J_{\ell+\frac{1}{2}}(-)$ and $Y_{\ell+\frac{1}{2}}(-)$ are the Bessel functions of first and second kind, respectively, of order $\ell+\frac{1}{2}$ (see Appendix \ref{subappendix-bessel}).
\end{remark}

Before stating the main results, we introduce the following notation: for any $k\in\mathbb{Z}$, $\alpha>0$ and $0<R<1$ we denote by $\mathcal{D}_k^{(\alpha)}(R)$ the closed annular sector in $\widetilde{\mathbb{C}^*}$ defined by
\begin{equation}
\label{18-aprile-2024-2}
\mathcal{D}_k^{(\alpha)}(R):=\left\{ R \le|x|\le 1,\, \frac{(k-1)\pi}{\alpha+1} \le \arg(x)\le \frac{(k+1)\pi}{\alpha+1} \right\}\subset\widetilde{\mathbb{C}^*}.
\end{equation}
 
\begin{lemma}
\label{lemma-16-maggio-2024-1}
Fix a compact subset $W\in\mathbb{C}^2$. There exist a number $0<R_W<1$ and a constant $\alpha_W>0$ depending only on $W$ such that for all $k\in\mathbb{Z}$ and all $\alpha\ge \alpha_W$ the ratio $Z_k^{(\alpha)}(x;E,\ell)=\frac{\psi_k^{(\alpha)}(x;E,\ell)}{\Psi_k^{(\alpha)}(x)}$ restricted to the closed domain $\mathcal{D}_k^{(\alpha)}(R_W)$ (see~\eqref{18-aprile-2024-2}) is uniformly bounded with respect to $(E,\ell)\in W$, namely
\begin{equation}
\label{16-maggio-2024-1}
\left| Z_k^{(\alpha)}(x;E,\ell) \right|_{x\in\mathcal{D}_k^{(\alpha)}(R_W)} \lesssim_K 1,\quad\mbox{for all }\alpha\ge\alpha_W.
\end{equation} 
\end{lemma}

\begin{proof}
Let us start choosing a number $0<R<1$ and let $\mathcal{D}_k^{(\alpha)}(R)$ be as in~\eqref{18-aprile-2024-2}. Since by Proposition \ref{theorem-10-gennaio-2024-1} the function $\Psi_k^{(\alpha)}(x)$ has no zeros for $\frac{(k-1)\pi}{\alpha+1}\le\arg(x)\le\frac{(k+1)\pi}{\alpha+1}$, we can consider again the Volterra integral equation~\eqref{10-gennaio-2024-1} satisfied by the ratio $Z_k^{(\alpha)}(x;E,\ell)$. For all $x\in\mathcal{D}_k^{(\alpha)}(R)$ we can choose the integration path $\gamma_k$ in~\eqref{10-gennaio-2024-1} to be $\gamma_k:=\hat{\gamma}_k^{(1)}*\hat{\gamma}_k^{(2)}*\hat{\gamma}_k^{(3)}$, where

\begin{equation}
\label{11-agosto-2024-11}
\begin{aligned}
& \hat{\gamma}_k^{(1)}:=-\left\{r e^{i\frac{k\pi}{\alpha+1}},\,1\le r<+\infty\right\} \\[2ex]
& \hat{\gamma}_k^{(2)}:=
\begin{cases}
\left\{ e^{i\phi},\,\frac{k\pi}{\alpha+1}\le \phi\le \arg(x)\right\}, & \mbox{if }\frac{k\pi}{\alpha+1}\le  \arg(x) \\[2ex]
-\left\{ e^{i\phi},\,\arg(x)\le \phi\le \frac{k\pi}{\alpha+1}\right\}, & \mbox{if }  \arg(x)\le \frac{k\pi}{\alpha+1},
\end{cases}
\\[2ex]
& \hat{\gamma}_k^{(3)}:=-\left\{r e^{i\arg(x)},\,|x|\le r \le 1\right\},
\end{aligned}
\end{equation}
see Figure \ref{figure-11-agosto-2024-2} below.
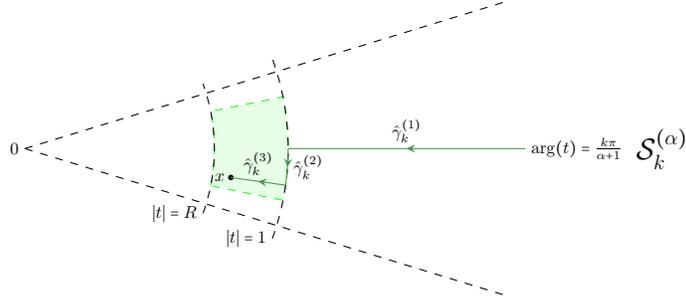
\begin{figure}[H]
\centering
\begin{tikzpicture}[decoration={markings, mark= at position 0.5 with {\arrow{stealth}}},scale=0.35, every node/.style={scale=0.65}]

\filldraw[color=green!90, fill=green!10, dashed] (-11.33:7.2) arc[start angle=-11.33, end angle=11.33, radius=7.2] -- (11.33:10) -- (11.33:10) arc[start angle=11.33, end angle=-11.33, radius=10] -- (-11.33:7.2);

\draw[dashed] (-20:10) arc[start angle=-20, end angle=20, radius=10];

\draw[dashed] (-20:7.2) arc[start angle=-20, end angle=20, radius=7.2];

\draw (-20:10) node[left] {$|t|=1$};

\draw (-20:7.2) node[left] {$|t|=R$};

\draw (0,0)  node[left] {$0$};

\fill (-8:7.9) circle[radius=3pt];

\draw (23,0) node[right] {\mbox{\LARGE $\mathcal{S}_k^{(\alpha)}$}};

\draw (-8:7.9) node[left] {$x$};
\draw (-4:10) node[right] {$\hat{\gamma}_k^{(2)}$};
\draw (0:14.5) node[above] {$\hat{\gamma}_k^{(1)}$};
\draw (-8:8.95) node[above] {$\hat{\gamma}_k^{(3)}$};

\draw[dashed] (0,0) -- (17:19);
\draw[dashed] (0,0) -- (-17:19);

\draw (19,0) node[right] {$\arg(t)=\frac{k\pi}{\alpha+1}$};

\draw[color=green!40!black!70, postaction=decorate] (0:19) -- (0:10);

\draw[color=green!40!black!70, postaction=decorate] (0:10) arc[start angle=0, end angle=-8, radius=10];

\draw[color=green!40!black!70, postaction=decorate] (-8:10) -- (-8:7.9);

\end{tikzpicture}
\caption{\small The paths $\hat{\gamma}_k^{(1)}$, $\hat{\gamma}_k^{(2)}$ and $\hat{\gamma}_k^{(3)}$ (dark green lines) defined in~\eqref{11-agosto-2024-11}. The green area represents the annular sector $\mathcal{D}_k^{(\alpha)}(R)$ defined in~\eqref{18-aprile-2024-2}.}
\label{figure-11-agosto-2024-2}
\end{figure}
At this point we follow the same procedure we used to prove Proposition \ref{proposition-3-aprile-2024-1}, namely we study uniform upper bounds for the quantity
\begin{equation}
\label{11-agosto-2024-12}
R_k^{(\alpha)}(x;E,\ell):=\left(\int_{\hat{\gamma}_k^{(1)}}dt+\int_{\hat{\gamma}_k^{(2)}}dt+\int_{\hat{\gamma}_k^{(3)}}dt \right)\frac{\Psi_k^{(\alpha)}(t)}{\Psi_k^{(\alpha)}(x)}G^{(\alpha)}(x,t) F(t;E,\ell).
\end{equation}
Upper bounds for the integrals along $\hat{\gamma}_k^{(1)}$ and $\hat{\gamma}_k^{(2)}$ in~\eqref{11-agosto-2024-12} are obtained in the same way given in the proof for Proposition \ref{proposition-3-aprile-2024-1}, thus we omit the details.

Let us study the integral along $\hat{\gamma}_k^{(3)}$: in this case, from the series representation~\eqref{modified-bessel-first-kind} and formula~\eqref{modified-bessel-second-kind} for the modified Bessel functions, we see that
\[
\left| \frac{\Psi_k^{(\alpha)}(t)}{\Psi_k^{(\alpha)}(x)}G^{(\alpha)}(x,t) \right|\lesssim 1,
\]
for all sufficiently big values of $\alpha$, thus, since the length of $\hat{\gamma}_k^{(3)}$ is
\begin{equation}
\label{11-agosto-2024-20}
\int_{\hat{\gamma}_k^{(3)}}|dt|\le 1- R,
\end{equation}
we find
\begin{equation}
\label{11-agosto-2024-21}
\left|\int_{\hat{\gamma}_k^{(3)}} \frac{\Psi_k^{(\alpha)}(t)}{\Psi_k^{(\alpha)}(x)}G^{(\alpha)}(x,t) F(t;E,\ell) \right|\lesssim_W \frac{1-R}{R^2}
\end{equation}
(the term $\frac{1}{R^2}$ as well as the dependence on $W$ of the implicit constant in the previous inequality come from the forcing term $F(t;E,\ell)$). Summing up, taking into account also the other pieces in~\eqref{11-agosto-2024-12}, there exist a constant $\alpha_W>1$ and a constant $\mathcal{C}_W>0$ depending only on $W$, such that
\begin{equation}
\label{11-agosto-2024-22}
\left|R_k^{(\alpha)}(x;E,\ell)\right|\le \mathcal{C}_W\left(\frac{1-R}{R^2}+\frac{1}{\alpha+1}\right).
\end{equation}
Choosing $R:=R_W<1$ to be sufficiently close to $1$ (depending only on $W$) so that
\[
\mathcal{C}_W\frac{1-R}{R^2}<\frac{1}{4}
\]
and $\alpha_W>1$ to be sufficiently big so that
\[
\frac{\mathcal{C}_W}{\alpha+1}<\frac{1}{4},
\]
from the integral equation~\eqref{10-gennaio-2024-1} it follows that
\[
\sup_{\alpha\ge\alpha_W}\sup_{(E,\ell)\in W}\sup_{x\in\mathcal{D}_k^{(\alpha)}(R_W)}\left|Z_k^{(\alpha)}(x;E,\ell)\right|\le 2
\]
and we conclude.
\end{proof}

\begin{lemma}
\label{remark-17-aprile-2024-1}
Let us fix numbers $0<R<1$ and $s\ge 1$. The following inequalities hold for all sufficiently big values of $\alpha$ and all $|x|\in\left[r,1+\frac{1}{(\alpha+1)^s}\right]$:

\begin{equation}
\label{9-aprile-2024-2}
\begin{aligned}
& \left|\frac{\Psi_k^{(\alpha)}(x)}{(-1)^{\frac{k}{2}}\left[e^{\frac{ik\pi}{2\alpha+2}}(2\alpha+2)^{\frac{1}{2\alpha+2}}\Gamma\left(1+\frac{1}{2\alpha+2}\right)-x e^{-\frac{ik\pi}{2\alpha+2}}(2\alpha+2)^{-\frac{1}{2\alpha+2}}\Gamma\left(1-\frac{1}{2\alpha+2}\right)\right]}-1\right| \\ & \lesssim\frac{1}{(\alpha+1)\log(2\alpha+2)}
\end{aligned}
\end{equation}
and
\begin{equation}
\label{9-aprile-2024-3}
\left|\frac{(-1)^{-\frac{k}{2}} e^{\frac{i k \pi}{2\alpha+2}}(2\alpha+2)^{\frac{1}{2\alpha+2}}}{\Gamma\left(1-\frac{1}{2\alpha+2}\right)}\frac{d\Psi^{(\alpha)}_k(x)}{dx}+1\right| \lesssim \frac{1}{\alpha+1}.
\end{equation}
\end{lemma}
\begin{proof}
It follows from the series representation~\eqref{modified-bessel-first-kind} and formula~\eqref{modified-bessel-second-kind}. 
\end{proof}

We can now compare the solution $\psi_k^{(\alpha)}(x;E,\ell)$  with the solution $\widetilde{\psi}_k^{(\alpha)}(x;E,\ell)$ of the Cauchy problem~\eqref{9-aprile-2024-1}. 

\begin{proposition}
\label{proposition-9-aprile-2024-1}
Let $k\in\mathbb{Z}$ be a fixed integer number and let us fix a compact subset $W\subset\mathbb{C}^2$. Let $0<R_W<1$ be a number (depending only on $W$) as in Lemma \ref{lemma-16-maggio-2024-1}. Sibuya's solution $\psi_k^{(\alpha)}(x;E,\ell)$ restricted to the closed domain $\mathcal{D}_k^{(\alpha)}(R_W)$ (see~\eqref{18-aprile-2024-2}), along with its derivative, is asymptotic to $\widetilde{\psi}_k^{(\alpha)}(x;E,\ell)$ as $\alpha\to+\infty$ uniformly with respect to $x\in\mathcal{D}_k^{(\alpha)}(R_W)$ and $(E,\ell)\in W$. More precisely, there exists a constant $\alpha_{W}(k)>0$, depending only on $W$ and $k$, such that
\begin{equation}
\label{9-aprile-2024-5}
\left|\psi_k^{(\alpha)}(x;E,\ell)-\widetilde{\psi}_k^{(\alpha)}(x;E,\ell)\right|\lesssim_{W,k} \frac{\log(\alpha+1)}{\alpha+1}
\end{equation}
and
\begin{equation}
\label{9-aprile-2024-6}
\left|\frac{d}{dx}\psi_{k}^{(\alpha)}(x;E,\ell)-\frac{d}{dx}\widetilde{\psi}_k^{(\alpha)}(x;E,\ell)\right|\lesssim_{W,k}\frac{\log( \alpha+1)}{\alpha+1}
\end{equation}
hold for all $x\in \mathcal{D}_k^{(\alpha)}(R_W)$, $\alpha\ge\alpha_{W}(k)$ and all $(E,\ell)\in W$.
\end{proposition}
\begin{proof}
We use the Volterra integral equation~\eqref{10-gennaio-2024-1} written in terms of $\psi_k^{(\alpha)}(x;E,\ell)$, namely
\begin{equation}
\label{10-gennaio-2024-1-bis}
\psi_k^{(\alpha)}(x;E,\ell)=\Psi_k^{(\alpha)}(x)+\int_{\gamma_k,\infty}^{x}G^{(\alpha)}(x,t)F(t;E,\ell)\psi_k^{(\alpha)}(t;E,\ell)dt
\end{equation} 
and the Volterra integral equation
\begin{equation}
\label{18-aprile-2024-4}
\widetilde{\psi}_k^{(\alpha)}(x;E,\ell)=e^{-\frac{i k \pi}{2}} e^{\frac{i k \pi }{\alpha+1}}-e^{-\frac{i k \pi}{2}}x-\int_{\mathring{\gamma}_k,x}^{e^{\frac{i k \pi}{\alpha+1}}} (x-t) F(t;E,\ell) \widetilde{\psi}_k^{(\alpha)}(t;E,\ell) dt
\end{equation}
satisfied by $\widetilde{\psi}_k^{(\alpha)}(x;E,\ell)$. The paths of integration in~\eqref{10-gennaio-2024-1-bis} are chosen to be $\gamma_k:=\hat{\gamma}_k^{(1)}*\hat{\gamma}_k^{(2)}*\hat{\gamma}_k^{(3)}$ and $\mathring{\gamma}_k:=\hat{\gamma}_k^{(2)}*\hat{\gamma}_k^{(3)}$, where $\hat{\gamma}_k^{(j)}$, $j=1,2,3$, are as in~\eqref{11-agosto-2024-11} (see also Figure \ref{figure-11-agosto-2024-2}). Let $\alpha_W>0$ be a constant depending on $W$ as in Lemma \ref{lemma-16-maggio-2024-1}. 

First of all, from inequality~\eqref{9-aprile-2024-2} of Lemma \ref{remark-17-aprile-2024-1}, we notice that, by eventually redefining $\alpha_{W}=\alpha_{W}(k)$ (depending also on the integer number $k$) so that $\frac{k \pi}{\log(2\alpha+2)}\le 1$ for all $\alpha\ge\alpha_W(k)$, we have
\begin{equation}
\label{18-aprile-2024-5}
\left|e^{-\frac{i k \pi}{2}}\left(e^{\frac{ik\pi}{\alpha+1}}-x\right)-\Psi_k^{(\alpha)}(x)\right|\lesssim_{W,k} \frac{\log(2 \alpha+2)}{\alpha+1}
\end{equation}
for all $x\in\mathcal{D}_k^{(\alpha)}(R_W)$ and all $\alpha\ge\alpha_{W}(k)$, while, making use of the series representation~\eqref{modified-bessel-first-kind}, we have
\begin{equation}
\label{19-aprile-2024-1}
\left|G^{(\alpha)}(x,t)-(x-t)\right|\lesssim_{W,k}\frac{1}{(\alpha+1)\log(2\alpha+2)},
\end{equation}
for all $x\in\mathcal{D}_k^{(\alpha)}(R_W)$ and all $\alpha\ge\alpha_{W}(k)$. From the integral equation~\eqref{10-gennaio-2024-1-bis} satisfied by $\psi_k^{(\alpha)}(x;E,\ell)$ and the integral equation~\eqref{18-aprile-2024-4} satisfied by $\widetilde{\psi}_k^{(\alpha)}(x;E,\ell)$, using inequalities~\eqref{18-aprile-2024-5},~\eqref{19-aprile-2024-1}, together with boundedness of the ratio $Z_k^{(\alpha)}(x;E,\ell)$ proved in Lemma \ref{lemma-16-maggio-2024-1}, we can find a constant $\tilde{\mathcal{C}}_{W}(k)>0$, depending only on $W$ and $k$, such that
\begin{equation}
\label{19-aprile-2024-2}
\begin{aligned}
& \sup_{x\in\mathcal{D}_k^{(\alpha)}(R_W)}\left|\psi_k^{(\alpha)}(x;E,\ell)-\widetilde{\psi}_k^{(\alpha)}(x;E,\ell)\right|\le \\ & \le \tilde{\mathcal{C}}_W(k)\left(\frac{1-R_{W}}{R_W^2}+\frac{1}{\alpha+1}\right)\sup_{x\in\mathcal{D}_k^{(\alpha)}(R_{W})}\left|\psi_k^{(\alpha)}(x;E,\ell)-\widetilde{\psi}_k^{(\alpha)}(x;E,\ell)\right| \\ & +\tilde{\mathcal{C}}_{W}(k)\frac{\log(2\alpha+2)}{\alpha+1}.
\end{aligned}
\end{equation}
By eventually redefining $0<R_{W}<1$ and $\alpha_{W}(k)>1$ so that
\[
\tilde{\mathcal{C}}_W(k)\frac{1-R_{W}}{R_W^2}\le\frac{1}{4},\quad \tilde{\mathcal{C}}_W(k)\frac{1}{\alpha+1}\le\frac{1}{4},
\]
we obtain
\[
\sup_{\substack{x\in\mathcal{D}_{k}^{(\alpha)}(R_{W}) \\ (E,\ell)\in W}}\left|\psi_k^{(\alpha)}(x;E,\ell)-\psi_k^{(\infty)}(x;E,\ell)\right|\le 2 \tilde{\mathcal{C}}_{W}(k)\frac{\log(2\alpha+2)}{\alpha+1},
\]
and we conclude.

To prove inequality~\eqref{9-aprile-2024-6} we proceed similarly: we have
\begin{equation}
\label{21-aprile-2024-1}
\frac{d}{dx}\psi_{k}^{(\alpha)}(x;E,\ell)=\frac{d}{dx}\Psi_x^{(\alpha)}(x)+\int_{\gamma_k,\infty}^x\frac{\partial}{\partial x}G^{(\alpha)}(x,t) F(t;E,\ell) \psi_k^{(\alpha)}(t;E,\ell)dt
\end{equation}
and
\begin{equation}
\label{21-aprile-2024-2}
\frac{d}{dx}\widetilde{\psi}_k^{(\alpha)}(x;E,\ell)=-e^{-\frac{i k \pi}{2}}-\int_{\mathring{\gamma}_k,x}^{e^{\frac{ik\pi}{\alpha+1}}} F(t;E,\ell)\widetilde{\psi}_k^{(\alpha)}(t;E,\ell) dt;
\end{equation}
from~\eqref{9-aprile-2024-3} of Lemma \ref{remark-17-aprile-2024-1} we have
\begin{equation}
\label{21-aprile-2024-3}
\left|\frac{d}{dx}\Psi_k^{(\alpha)}(x)+e^{-\frac{i k \pi}{2}}\right|\lesssim_{W,k} \frac{\log(2 \alpha+2)}{2\alpha+2}+\frac{1}{\alpha+1}
\end{equation}
and using the series representation~\eqref{modified-bessel-first-kind} we find
\begin{equation}
\label{21-aprile-2024-4}
\left|\frac{\partial}{\partial x} G^{(\alpha)}(x,t)-1\right|\lesssim_{W,k} \frac{1}{\alpha+1},
\end{equation}
for all $\alpha\ge\alpha_W(k)$ and all $x,t\in\mathcal{D}_k^{(\alpha)}(R_{W})$. Putting together estimates~\eqref{21-aprile-2024-1}-\eqref{21-aprile-2024-4} and following the same procedure as before, we reach the sought inequality.
\end{proof}

\begin{corollary}
\label{corollary-13-agosto-2024-1}
With the same assumptions and notations of Proposition \ref{proposition-9-aprile-2024-1} specialized to $k=0$, inequalities
\begin{equation}
\label{13-agosto-2024-30}
\left|\psi_0^{(\alpha)}(1;E,\ell)\right|\lesssim_{W} \frac{\log(\alpha+1)}{\alpha+1}
\end{equation}
and
\begin{equation}
\label{13-agosto-2024-31}
\left|\left.\frac{d}{dx}\psi_0^{(\alpha)}(x;E,\ell)\right|_{x=1}+1\right|\lesssim_W\frac{\log(\alpha+1)}{\alpha+1}
\end{equation}
hold for all $\alpha\ge\alpha_W$ and $(E,\ell)\in W$.
\end{corollary}

\begin{proof}
It is an immediate consequence of Proposition \ref{proposition-9-aprile-2024-1}.
\end{proof}

Proposition \ref{proposition-9-aprile-2024-1} allows us to compute the Stokes multiplier $\sigma_k^{(\alpha)}(E;\ell)$ defined in~\eqref{31-agosto-2024-2} in the large $\alpha$ limit:

\begin{proposition}
\label{proposition-16-maggio-2024-10}
Let $k\in\mathbb{Z}$ be a fixed integer number and let $W\subset\mathbb{C}^2$ be a fixed compact subset.

The Stokes multiplier $\sigma_k^{(\alpha)}(E;\ell)$ defined in~\eqref{31-agosto-2024-2} converges to the constant $-2 i $ as $\alpha\to+\infty$, uniformly with respect to $(E,\ell)\in W$; more precisely, there exists a constant $\alpha_{W}(k)>1$ depending only on $W$ and $k$, such that
\begin{equation}
\label{21-aprile-2024-6}
\left|\frac{\sigma_k^{(\alpha)}(E;\ell)}{2 i}+1\right|\lesssim_{W,k} \frac{\log( \alpha+1)}{\alpha+1}
\end{equation}
hold for all $\alpha\ge\alpha_{W}(k)$ and all $(E,\ell)\in W$.
\end{proposition}

\begin{proof}
Let us define
\[
\widetilde{\sigma}_k^{(\alpha)}(E;\ell):=\frac{\mathcal{W}\left[ \widetilde{\psi}_{k-1}^{(\alpha)}(x;E,\ell),\widetilde{\psi}_{k+1}^{(\alpha)}(x;E,\ell) \right]}{\mathcal{W}\left[\widetilde{\psi}_{k}^{(\alpha)}(x;E,\ell),\widetilde{\psi}_{k+1}^{(\alpha)}(x;E,\ell)\right]};
\]
from the explicit expression~\eqref{22-aprile-2024-1} of Remark \ref{remark-22-aprile-2024-1} it follows that
\[
\begin{aligned}
& \widetilde{\sigma}_k^{(\alpha)}(E,\ell)=-i e^{-\frac{i\pi}{\alpha+1}}\times \\
& \times \frac{ J_{\ell+\frac{1}{2}}\left(E^{\frac{1}{2}} e^{\frac{i(k+1)\pi}{\alpha+1}}\right) Y_{\ell+\frac{1}{2}}\left(E^{\frac{1}{2}} e^{\frac{i(k-1)\pi}{\alpha+1}}\right)- Y_{\ell+\frac{1}{2}}\left(E^{\frac{1}{2}} e^{\frac{i(k+1)\pi}{\alpha+1}}\right) J_{\ell+\frac{1}{2}}\left(E^{\frac{1}{2}} e^{\frac{i(k-1)\pi}{\alpha+1}}\right) }{ J_{\ell+\frac{1}{2}}\left(E^{\frac{1}{2}} e^{\frac{i(k+1)\pi}{\alpha+1}}\right) Y_{\ell+\frac{1}{2}}\left(E^{\frac{1}{2}} e^{\frac{ik\pi}{\alpha+1}}\right)-Y_{\ell+\frac{1}{2}}\left(E^{\frac{1}{2}} e^{\frac{i(k+1)\pi}{\alpha+1}}\right) J_{\ell+\frac{1}{2}}\left(E^{\frac{1}{2}} e^{\frac{ik\pi}{\alpha+1}}\right)}.
\end{aligned}
\]
From Proposition \ref{proposition-9-aprile-2024-1} it follows that there exist constants $\alpha_{W}(k)>1$ and $\widetilde{\mathcal{C}}_{W}(k)>0$ such that
\begin{equation}
\label{22-aprile-2024-7}
\left|\sigma_k^{(\alpha)}(E;\ell)-\widetilde{\sigma}_k^{(\alpha)}(E;\ell)\right|\le \widetilde{\mathcal{C}}_{W}(k)\frac{\log(2 \alpha+2)}{\alpha+1},
\end{equation}
for all $\alpha\ge\alpha_{W}(k)$ and all $(E,\ell)\in W$. Furthermore, from identity~\eqref{multiplication-theorem} of Appendix \ref{subappendix-bessel} (in particular, due to the uniform and absolute convergence of the series on the right hand side for $E$ taking values on a compact subset), by eventually taking a larger $\alpha_{W}(k)>1$, we can find a constant $\widetilde{\widetilde{\mathcal{C}}}_{W}(k)>0$, depending on $W$ and $k$, such that
\begin{equation}
\label{22-aprile-2024-8}
\left|\frac{\widetilde{\sigma}_k^{(\alpha)}(E;\ell)}{2 i}+1\right|\le \frac{\widetilde{\widetilde{\mathcal{C}}}_{W}(k)}{\alpha+1},
\end{equation}
for all $\alpha\ge\alpha_{W}(k)$, $(E,\ell)\in W$. Putting together inequalities~\eqref{22-aprile-2024-7} and~\eqref{22-aprile-2024-8}, the conclusion follows.
\end{proof}

\subsection{Local theory at the (formal) regular singularity, Frobenius solutions and proof of Theorem \ref{theorem-23-aprile-2024-1}}
\label{subsection-18-maggio-2024-1}
In this section, we study the local theory for equation~\eqref{eqn} around $x=0$. As remarked in the Introduction, due to the symmetry $\ell\mapsto -\ell-1$ of the potential $V^{(\alpha)}(x;E,\ell)$, we can restrict the domain of the angular parameter $\ell$ to be open half plane $\mathbb{H}_+=\left\{\operatorname{Re}(\ell)>-\frac{1}{2}\right\}$. The case $\operatorname{Re}(\ell)=-\frac{1}{2}$ will not be considered here.

Since at $x=0$ the centrifugal term $\frac{\ell(\ell+1)}{ x^2}-E$ dominates the potential, we expect that pairs of dominant/subdominant solutions can be (uniformly) approximated by pairs of dominant/subdominant solutions to equation
\begin{equation}
\label{bessel-u}
\frac{d^2X}{dx^2}=\left(\frac{\ell(\ell+1)}{ x^2}-E\right)X,
\end{equation}
whose solutions can be written in terms of Bessel functions (see equation~\eqref{eeeeeeeeeee} of Appendix \ref{subappendix-bessel}). We denote by $X_+^{(\alpha)}(x;E,\ell)$ and $X_-^{(\alpha)}(x;E,\ell)$ the solutions to~\eqref{bessel-u} defined by
\begin{equation}
\label{definition-u0}
X_+^{(\alpha)}(x;E,\ell):=\left(\frac{E^{\frac{1}{2}}}{2}\right)^{-\ell-\frac{1}{2}}\frac{\Gamma\left(\frac{3}{2}+\ell\right)}{\Gamma\left(1+\frac{\ell+\frac{1}{2}}{\alpha+1}\right)} x^{\frac{1}{2}} J_{\ell+\frac{1}{2}}\left(x E^{\frac{1}{2}}\right)
\end{equation}
and
\begin{equation}
\label{definition-u1}
X_-^{(\alpha)}(x;E,\ell):=\left(\frac{E^{\frac{1}{2}}}{2}\right)^{\ell+\frac{1}{2}}\frac{\Gamma\left(\frac{1}{2}-\ell\right)}{\Gamma\left(1-\frac{\ell+\frac{1}{2}}{\alpha+1}\right)} x^{\frac{1}{2}} J_{-\ell-\frac{1}{2}}\left(x E^{\frac{1}{2}}\right),\quad \mbox{for }\ell\notin\Lambda^{(\alpha)},
\end{equation}
where $\Lambda^{(\alpha)}$ is the discrete set defined in~\eqref{14-maggio-2024-10} (here we have used the standard notation for the Bessel functions of the first and second kind, see Appendix \ref{bessel-appendix}). All the functions in formulas~\eqref{definition-u0},~\eqref{definition-u1} are understood to take their principal value on the cut plane $\mathbb{C}\setminus(-\infty,0]$, according to the convention given in Appendix \ref{subappendix-bessel}.

\begin{remark}
\label{remark-17-maggio-2024-1}
The functions $X_+^{(\alpha)}(x;E,\ell)$ and $X_-^{(\alpha)}(x;E,\ell)$ defined in~\eqref{definition-u0} and~\eqref{definition-u1} enjoy the following properties:
\begin{itemize}
\item[i)] They are analytic functions of $x\in\widetilde{\mathbb{C}^*}$ and entire functions of $E\in\mathbb{C}$ for all $\alpha>0$ and all $\operatorname{Re}(\ell)>-\frac{1}{2}$ (with $\ell\notin\Lambda^{(\alpha)}$ for the second one);
\item[ii)] They are normalized so that
\[
\lim_{x\to 0} \Gamma\left(1+\frac{\ell+\frac{1}{2}}{\alpha+1}\right) x^{-\ell-1} X_+^{(\alpha)}(x;E,\ell)=1
\]
and
\[
\lim_{x\to 0}\Gamma\left(1-\frac{\ell+\frac{1}{2}}{\alpha+1}\right) x^{\ell} X_-^{(\alpha)}(x;E,\ell)=1;
\]
\item[iii)] The action of the operator $\mathcal{M}^{(\alpha)}$ defined in~\eqref{11-maggio-2024-2} is
\[
\mathcal{M}^{(\alpha)} X_+^{(\alpha)}(x;E,\ell)= e^{\frac{i  \pi}{\alpha+1}(\ell+1)} X_+^{(\alpha)}(x;E,\ell)
\]
and
\[
\mathcal{M}^{(\alpha)} X_-^{(\alpha)}(x;E,\ell)= e^{-\frac{i \pi}{\alpha+1}} X_-^{(\alpha)}(x;E,\ell),\quad \ell\notin\Lambda^{(\alpha)};
\]
\item[iv)] Their Wronskian is
\[
\mathcal{W}\left[X_-^{(\alpha)}(x;E,\ell),X_+^{(\alpha)}(x;E,\ell)\right]=\frac{2 \alpha+2}{\pi} \sin\left(\frac{\pi}{\alpha+1}\left(\ell+\frac{1}{2}\right)\right),\quad \ell\notin\Lambda^{(\alpha)}.
\]
\end{itemize}
\end{remark}

We show now that in the large $\alpha$ limit the subdominant Frobenius solution $\chi_+^{(\alpha)}(x;E,\ell)$ defined in~\eqref{19-agosto-2024-1} can be (uniformly) approximated with the function $X_+^{(\alpha)}(x;E,\ell)$ given in~\eqref{definition-u0}.

\begin{proposition}
\label{bessel-approximation-zero}
Let $W\subset \mathbb{C}\times \mathbb{H}_+$ be a compact subset and let $\mathcal{U}^{(\alpha)}$ be the cut neighborhood
\begin{equation}
\label{23-aprile-2024-9}
\mathcal{U}^{(\alpha)}:=\left\{|x|\le1+\frac{1}{\alpha+1}\right\}\setminus(-\infty,0].
\end{equation}
There exists a constant $\alpha_{W}>0$ depending only on $W$ such that
\begin{equation}
\label{23-aprile-2024-6}
\left|\chi_+^{(\alpha)}(x;E,\ell)-X_+^{(\alpha)}(x;E,\ell)\right|\lesssim_{W} \frac{1}{\alpha+1},
\end{equation}
and
\begin{equation}
\label{23-aprile-2024-7}
\left|\frac{d}{dx}\chi_+^{(\alpha)}(x;E,\ell)-\frac{d}{dx}X_+^{(\alpha)}(x;E,\ell)\right|\lesssim_W \frac{1}{\alpha+1}
\end{equation}
hold for all $x\in\mathcal{U}^{(\alpha)}$, all $\alpha\ge\alpha_{W}$ and all $(E,\ell)\in W$. 
\end{proposition}

\begin{proof}
By direct inspection, we can check that the distinguished Frobenius solution $\chi_+^{(\alpha)}(x;E,\ell)$ satisfies the following Volterra integral equation:
\begin{equation}
\label{volterra-bessel-zero}
\chi_+^{(\alpha)}(x;E,\ell)=X_+(x;E,\ell)+\int_{\gamma_0,0}^x H(x,t;E,\ell)t^{2 \alpha} \chi_+^{(\alpha)}(t;E,\ell)dt,
\end{equation}
where the kernel $H(x,t;E,\ell)$ is defined as
\begin{equation}
\label{23-aprile-2024-8}
H(x,t;E,n):= \frac{\pi}{2}(xt)^{\frac{1}{2}}\left[J_{\ell+\frac{1}{2}}\left(E^{\frac{1}{2}} t \right) Y_{\ell+\frac{1}{2}}\left(E^{\frac{1}{2}} x \right)-J_{\ell+\frac{1}{2}}\left(E^{\frac{1}{2}} x \right) Y_{\ell+\frac{1}{2}}\left(E^{\frac{1}{2}} t \right)\right],
\end{equation}
$X_+^{(\alpha)}(x;E,\ell)$ is the function defined in~\eqref{definition-u0}, and $\gamma_0$ is the ray joining $0$ to $x$. By formulas~\eqref{bessel-first-kind},~\eqref{bessel-second-kind-generic} and~\eqref{bessel-second-kind-nongeneric}, inequality
\begin{equation}
\label{23-aprile-2024-10}
\left|H(x,t;E,\ell)\right|\lesssim_W\left(|t|^{\operatorname{Re}(\ell)+1}|x|^{-\operatorname{Re}(\ell)}+|t|^{-\operatorname{Re}(\ell)}|x|^{\operatorname{Re}(\ell)+1}\right)
\end{equation}
holds for all $x,t\in \mathcal{U}^{(\alpha)}$ and all $(E,\ell)\in W$. Letting $\alpha_W>0$ to be a constant such that
\[
\alpha_W\ge \sup_{\ell\in W} \operatorname{Re}(\ell),
\]
we see that inequality
\begin{equation}
\label{23-aprile-2024-11}
\left|\int_{\gamma_0,0}^xH(x,t;E,\ell) t^{2 \alpha}dt \right| \lesssim_W \frac{1}{\alpha+1},
\end{equation}
hold for all $x\in \mathcal{U}^{(\alpha)}$, all $\alpha\ge\alpha_W$ and all $(E,\ell)\in W$. 

By eventually redefining $\alpha_W$ so that $\frac{3 e^2 \mathcal{C}_W}{2(\alpha+1)}\le\frac{1}{2}$ for all $\alpha\ge \alpha_{W}$, and defining
\[
\widehat{\mathcal{C}}_{W}:=\sup_{\alpha\ge\alpha_{W}}\sup_{\substack{x\in\mathcal{U}^{(\alpha)} \\ (E,\ell)\in W}}\left|X_+^{(\alpha)}(x;E,\ell)\right|<\infty,
\]
by making use of inequality~\eqref{23-aprile-2024-11} and the integral relation~\eqref{volterra-bessel-zero}, we find
\begin{equation}
\label{23-aprile-2024-12}
\sup_{\alpha\ge\alpha_{W}}\sup_{\substack{x\in\mathcal{U}^{(\alpha)} \\ (E,\ell)\in W}}\left|\chi^{(\alpha)}_+(x;E,\ell)\right|\le 2 \widehat{\mathcal{C}}_{W}.
\end{equation}
Finally, using again equation~\eqref{volterra-bessel-zero} together with~\eqref{23-aprile-2024-12} we obtain inequality~\eqref{23-aprile-2024-6}.

To prove inequality~\eqref{23-aprile-2024-7}, we just take the derivative of~\eqref{volterra-bessel-zero} and we proceed similarly (using also formula~\eqref{bessel-function-derivative} of Appendix \ref{subappendix-bessel} to compute and estimate the derivatives of the Bessel functions).
\end{proof}

Before stating the analogous approximation result for the dominant Frobenius solution, we want to remark that from the integral equation~\eqref{volterra-bessel-zero} we have used in the proof of Proposition \ref{bessel-approximation-zero} we can check that the subdominant Frobenius solution $\chi_+^{(\alpha)}(x;E,\ell)$ enjoy the Dorey-Tateo symmetry~\eqref{13-maggio-2024-4}:

\begin{corollary}
\label{corollary-17-maggio-2024-1}
With the same assumptions and notations of Proposition \ref{bessel-approximation-zero}, the following holds:
\[
\mathcal{M}^{(\alpha)}\chi_+^{(\alpha)}(x;E,\ell)= e^{\frac{i  \pi}{\alpha+1}(\ell+1)}\chi_+^{(\alpha)}(x;E,\ell),
\]
where $\mathcal{M}^{(\alpha)}$ is the quantum monodromy endomorphism defined in~\eqref{11-maggio-2024-2}.
\end{corollary}
\begin{proof}
It follows from point iii) of Remark \ref{remark-17-maggio-2024-1}, the covariance property
\begin{equation}
\label{cov-prop-h}
H\left( x e^{\frac{i m \pi}{\alpha+1}}, t e^{\frac{i m \pi}{\alpha+1}}; E,\ell \right)= e^{\frac{i m \pi}{\alpha+1}} H\left(x,t;E e^{\frac{2 i m \pi}{\alpha+1}},\ell\right),\quad\mbox{for all }m\in\mathbb{Z}
\end{equation} 
of the kernel and uniqueness of solutions to Volterra integral equations (the proof is similar to the one given for Proposition \ref{proposition-15-maggio-2024-1}).
\end{proof}

We show now that the distinguished Frobenius solution $\chi_-^{(\alpha)}(x;E,\ell)$ defined in~\eqref{19-agosto-2024-3} can be (uniformly) approximated by the function $X_-^{(\alpha)}(x;E,\ell)$ defined in ~\eqref{definition-u1}:

\begin{proposition}
\label{proposition-24-aprile-2024-1}
Let $W\subset\mathbb{C}\times\left(\mathbb{H}_+\setminus\Lambda^{(\alpha)}\right)$, with $\Lambda^{(\alpha)}$ the discrete set defined in~\eqref{14-maggio-2024-10}, be a compact subset and let $\mathcal{U}^{(\alpha)}$ be the cut neighborhood~\eqref{23-aprile-2024-9}. There exists a constant $\alpha_{W}>0$ depending only on $W$ such that
\begin{equation}
\label{24-aprile-2024-11}
|x|^\ell\left| \chi_-^{(\alpha)}(x;E,\ell)-X_-^{(\alpha)}(x;E,\ell)\right|\lesssim_W \frac{1}{\alpha+1}
\end{equation}
and
\begin{equation}
\label{24-aprile-2024-12}
|x|^{\ell+1}\left|\frac{d}{dx}\chi_-^{(\alpha)}(x;E,\ell)-\frac{d}{dx}X_-^{(\alpha)}(x;E,\ell)\right|\lesssim_W \frac{1}{\alpha+1}
\end{equation}
hold for all $x\in\mathcal{U}^{(\alpha)}$, $\alpha\ge\alpha_W$ and all $(E,\ell)\in W$.
\end{proposition}
\begin{proof}
We consider again the Volterra integral equation~\eqref{volterra-bessel-zero} with $X_+^{(\alpha)}(x;E,\ell)$ replaced by $X_-^{(\alpha)}(x;E,\ell)$ and $\chi_+^{(\alpha)}(x;E,\ell)$ replaced by $\chi_-^{(\alpha)}(x;E,\ell)$. Checking that it is actually satisfied by $\chi_-^{(\alpha)}(x;E,\ell)$ is a bit more subtle because, besides normalization condition~\eqref{12-maggio-2024-2} and analyticity with respect to $x$ and $E$, we have to check also condition~\eqref{31-agosto-2024-1} for the action of the quantum monodromy $\mathcal{M}^{(\alpha)}$ (see~\eqref{11-maggio-2024-2}). These properties follows from the fact that we cut the subset $\Lambda^{(\alpha)}$ from the complex half plane of $\ell$ (so that no logarithmic term appears) and from the covariance property~\eqref{cov-prop-h} of the kernel $H(x,t;E,\ell)$.

With the change of dependent variable $z_-(x)=x^\ell\chi(x)$ the integral equation becomes
\begin{equation}
\label{25-aprile-2024-2}
z_-(x)=x^\ell X_-^{(\alpha)}(x;E,\ell)+\int_{\gamma,x_1}^x x^\ell H(x,t;E,\ell) t^{2 \alpha-\ell}z_-(t) dt.
\end{equation}
At this point the procedure given in the proof of Proposition \ref{bessel-approximation-zero} applies verbatim and we obtain inequality~\eqref{24-aprile-2024-11}. Inequality~\eqref{24-aprile-2024-12} is proved similarly.
\end{proof}

Using the results in Propositions \ref{bessel-approximation-zero} and \ref{proposition-24-aprile-2024-1} together with the results of Proposition \ref{proposition-9-aprile-2024-1} specialized to the solution $\psi_0^{(\alpha)}(x;E,\ell)$ which is subdominant as $x\to+\infty$, we can prove Theorem \ref{theorem-23-aprile-2024-1}:

\begin{proof}[Proof of Theorem \ref{theorem-23-aprile-2024-1}]
We will write the proof of the Theorem for the \textquotedblleft+\textquotedblright\,case, being the \textquotedblleft-\textquotedblright\,case analogous.

Let us start with point i). Let $\mathcal{D}_0^{(\alpha)}(R_W)$ be the domain~\eqref{18-aprile-2024-2} with $0<R_W<1$ as in Proposition \ref{proposition-9-aprile-2024-1} and let $\mathcal{U}^{(\alpha)}$ be the domain~\eqref{23-aprile-2024-9} given in Proposition \ref{bessel-approximation-zero}. We see that for all $\alpha>0$ their intersection is non-empty, in particular
\[
1\in \mathcal{U}^{(\alpha)}\cap \mathcal{D}_0^{(\alpha)}(R_W),\quad\mbox{for all }\alpha>0.
\]
Writing
\[
\mathcal{Q}_+^{(\alpha)}(E;\ell)=\chi_+^{(\alpha)}(1;E,\ell)\left.\frac{d}{dx}\psi_0^{(\alpha)}(x;E,\ell)\right|_{x=1}-\psi_0^{(\alpha)}(1;E,\ell)\left.\frac{d}{dx}\chi_+^{(\alpha)}(x;E,\ell)\right|_{x=1},
\]
the result follows by Corollary \ref{corollary-13-agosto-2024-1} and Proposition \ref{bessel-approximation-zero}.

Point ii) is an immediate consequence of inequality~\eqref{uniform-convergence-spectral-determinants}: we have
\[
\left|\mathcal{Q}_+^{(\alpha)}(E;\ell)\right| - \left|\Gamma\left(1+\left(\frac{1}{2}+\ell\right)\right)\left(\frac{E^{\frac{1}{2}}}{2}\right)^{-\left(\ell+\frac{1}{2}\right)} J_{\ell+\frac{1}{2}}\left( E^{\frac{1}{2}} \right)\right|\gtrsim_{W_1,W_{2,+}}-\frac{\log( \alpha+1)}{\alpha+1},
\]
hence, for all sufficiently big values of $\alpha$, point iii) follows.

We are going now to prove point iii), from which point iv) follows, being a special case. Let $E^{(\infty)}_+(\ell)$ be a complex zero of $J_{\ell+\frac{1}{2}}\left(E^{\frac{1}{2}}\right)$, with $\ell$ ranging in a compact subset $W_{2,+}\subset\mathbb{H}_+$. Let $M>0$ be a fixed positive constant and let $\mathcal{D}$ be the disc
\begin{equation}
\label{23-aprile-2024-17}
\mathcal{D}:=\left\{\left|\frac{E}{E^{(\infty)}_{+}(\ell)}-1\right|\le M \frac{\log( \alpha+1)}{\alpha+1}\right\}
\end{equation}
which is a neighborhood of $E_{+}^{(\infty)}(\ell)$ contained in a compact subset $W_1\subset\mathbb{C}$, which is independent of $M$, for all $\ell\in W_{2,+}$ and all sufficiently big values of $\alpha$ depending only on $W_{2,+}$ and on $M$. Choosing $M=M(W_{2,+})>0$ so that
\[
M(W_{2,+})>\frac{2 \widehat{\mathcal{C}}_{1,2}}{\inf_{\ell\in W_{2,+}}\left| E_{+}^{(\infty)}(\ell) J^{(1)}_{\ell+\frac{1}{2}}\left(\left(E_{+}^{(\infty)}(\ell)\right)^{\frac{1}{2}}\right) \right|},
\]
where $\widehat{\mathcal{C}}_{1,2}>0$ is a constant so that inequality~\eqref{uniform-convergence-spectral-determinants} holds (which is independent of $M$), and by eventually taking larger values of $\alpha$ so that
\[
\sum_{s\ge 1}\frac{\left| E_{+}^{(\infty)}(\ell) \right|^s}{(s+1)!}\left|\frac{ J^{(s+1)}_{\ell+\frac{1}{2}}\left(\left(E_{+}^{(\infty)}(\ell)\right)^{\frac{1}{2}}\right) }{ J^{(1)}_{\ell+\frac{1}{2}}\left(\left(E_{+}^{(\infty)}(\ell)\right)^{\frac{1}{2}}\right) }\right| (M(W_{2,+}))^s\left(\frac{\log( \alpha+1)}{\alpha+1}\right)^s\le \frac{1}{2},
\]
where $J_{\ell+\frac{1}{2}}^{(s)}(-)$ denotes the $s$-th derivative (notice that by the continuous dependence of $E_{+}^{(\infty)}(\ell)$ on $\ell$, by the semplicity of $E_{+}^{(\infty)}(\ell)$ as a root of $J_{\ell+\frac{1}{2}}\left(E^{\frac{1}{2}}\right)$ and by the continuity of the $ J^{(s)}_{\ell+\frac{1}{2}}\left(\left(E_{+}^{(\infty)}(\ell)\right)^{\frac{1}{2}}\right)$'s with respect to $\ell$, all the previous expressions are well defined and finite), on the boundary of $\mathcal{D}$ we have
\[
\left|J_{\ell+\frac{1}{2}}\left(E^{\frac{1}{2}}\right)\right|_{E\in\partial\mathcal{D}}> \widehat{\mathcal{C}}_{1,2}\frac{\log( \alpha+1)}{\alpha+1},
\]
thus, by making use of inequality~\eqref{uniform-convergence-spectral-determinants} again, we have
\[
\begin{aligned}
\left|\Gamma\left(1+\left(\frac{1}{2}+\ell\right)\right)\left(\frac{E^{\frac{1}{2}}}{2}\right)^{-\left(\ell+\frac{1}{2}\right)} J_{\ell+\frac{1}{2}}\left( E^{\frac{1}{2}} \right)\right|_{\partial \mathcal{D}}> \\ \left|\mathcal{Q}_+^{(\alpha)}(E;\ell)-\Gamma\left(1+\left(\frac{1}{2}+\ell\right)\right)\left(\frac{E^{\frac{1}{2}}}{2}\right)^{-\left(\ell+\frac{1}{2}\right)} J_{\ell+\frac{1}{2}}\left( E^{\frac{1}{2}} \right)\right|_{\partial \mathcal{D}}
\end{aligned}
\]
for all $\ell\in W_{2,+}$ and all sufficiently big values of $\alpha$, depending only on $W_{2,+}$. By Rouché Theorem we conclude.
\end{proof}

\begin{remark}
\label{remark-20-agosto-2024-1}
Point iv) of Theorem \ref{theorem-23-aprile-2024-1} holds for any finite sequence $J_m=\left\{k_1,\ldots,k_m\right\}$, $m\in\mathbb{Z}_{>0}$, of nonnegative integers $k_i\in\mathbb{Z}_{\ge 0}$ and the corresponding finite sequence $j_{\pm\left(\ell+\frac{1}{2}\right),k_1},\ldots, j_{\pm\left(\ell+\frac{1}{2}\right),k_m}$. More precisely, with the same assumptions and notations of Theorem \ref{theorem-23-aprile-2024-1}, the following statement holds: fix $\ell>-\frac{1}{2}$; there are precisely $m$ simple zeros $E_{k_1,\pm}^{(\alpha)}(\ell),\ldots, E_{k_m,\pm}^{(\alpha)}(\ell)$ of $\mathcal{Q}_\pm^{(\alpha)}(E;\ell)$ such that
\begin{equation}
\label{23-aprile-2024-18}
\left|\frac{E_{k_i,\pm}^{(\alpha)}(\ell)}{j^2_{\pm\left(\ell+\frac{1}{2}\right),k_i}}-1\right|\lesssim_{J_p} \frac{\log(\alpha+1)}{\alpha+1}
\end{equation}
for all $i=1,\ldots,m$ and all sufficiently big values of $\alpha$ (depending on $J_m$). The proof is just a repetition of the same proof we gave before for all the indices in $J_p$. A similar observation holds for point iii). 
\end{remark}

\section{The large degree, energy and angular momentum limit}
\label{large-deg-en-mom}

In this section we study equation~\eqref{eqn} in the large $\alpha$, $E$ and $\ell$ regime. Let $\varepsilon,p$ be complex parameters, $p\ne 0$\footnote{The case $p=0$ is already included in the analysis of section \ref{large-alpha-fixed-e-l}, thus we can exclude it here without loss of generality.}, and let us rescale the energy parameter $E$ and the angular momentum parameter $\ell$ in~\eqref{eqn} as
\begin{equation}
\label{15-ott-2023-2}
\ell=2p(\alpha+1)-\frac{1}{2},\quad E= 4p^2(\alpha+1)^2\varepsilon^2,
\end{equation} 
so that equation~\eqref{eqn} becomes
\begin{equation}
\label{15-ott-2023-3}
\frac{d^2\psi}{dx^2}=\left[x^{2 \alpha}+\frac{16p^2(\alpha+1)^2-1}{4 x^2}-4p^2(\alpha+1)^2\varepsilon^2\right]\psi.
\end{equation}
We will refer to $\varepsilon$ and $p$ as the \textit{rescaled energy} and \textit{rescaled angular momentum} parameters, respectively.

Following the same scheme of section \ref{large-alpha-fixed-e-l}, we study separately the local problems at the (formal) irregular singularity $x=\infty$ and at the (formal) regular singularity $x=0$.

\subsection{Local theory at the (formal) irregular singularity and Stokes multipliers}
\label{subsection-all-large-infinity}

The analysis of equation~\eqref{15-ott-2023-3} as $x\to\infty$ is similar to the one performed in section \ref{large-alpha-fixed-e-l-infinity}. With some minor changes in the statements of the propositions,  the proofs are carried out in the same way by merely plugging in the new parameterizations~\eqref{15-ott-2023-2} for the energy and angular momentum parameters. For this reason, in this section we will just state the results and quote the few technical differences in the proofs.

\begin{proposition}
\label{proposition-3-aprile-2024-1-bis}
Let $W\subset\mathbb{C}\times\mathbb{C}^*$ be a compact subset. There exist a constant $\alpha_W>0$ depending only on $W$ and a number $r_0>1$ (independent of all the parameters) such that inequality
\begin{equation}
\label{3-aprile-2024-3-bis}
\left|\frac{\psi_k^{(\alpha)}\left(x;4p^2(\alpha+1)^2\varepsilon^2,2p(\alpha+1)-\frac{1}{2}\right)}{\Psi_k^{(\alpha)}(x)}-1\right|\lesssim_{W}\frac{\alpha+1}{r_0^{\alpha-1}}
\end{equation}
holds for all $x$ in the closed annular sector
\[
\left\{\frac{(k-1)\pi}{\alpha+1}\le \arg(x)\le\frac{(k+1)\pi}{\alpha+1},\,|x|\ge 1\right\}\subset\mathcal{S}_k^{(\alpha)},
\]
$\alpha\ge\alpha_W$ and $(\varepsilon,p)\in W$ (here $\Psi_k^{(\alpha)}(x)$ is the function~\eqref{14-dez-2023-1}, while $\mathcal{S}_k^{(\alpha)}$ is the sector~\eqref{3-aprile-2024-2}).
\end{proposition}

\begin{proof}
See the proof of Proposition \ref{proposition-3-aprile-2024-1}.
\end{proof}

Here the only difference with Proposition \ref{proposition-3-aprile-2024-1} is the upper bound in~\eqref{3-aprile-2024-3-bis}, which is due to the factor $(\alpha+1)^2$ introduced by the new parameterization.

Sibuya's solution $\psi_k^{(\alpha)}\left(x;4p^2(\alpha+1)^2\varepsilon^2,2p(\alpha+1)-\frac{1}{2}\right)$ is asymptotic to the function $\widetilde{\psi}_k^{(\alpha)}\left(x;4p^2(\alpha+1)^2\varepsilon^2,2p(\alpha+1)-\frac{1}{2}\right)$, being $\widetilde{\psi}_k^{(\alpha)}(x;E,\ell)$  the solution to~\eqref{spherical-potential-well} with Cauchy data~\eqref{9-aprile-2024-1}. To prove this, we use the same procedure we followed in the previous analogous case, but we have to take case of some subtleties. Indeed, due to the factor $(\alpha+1)^2$ of the new parameterization, the ratio 
\begin{equation}
\label{reparameterized-ratio}
Z_k^{(\alpha)}\left(x;4p^2(\alpha+1)^2\varepsilon^2,2p(\alpha+1)-\frac{1}{2}\right)=\frac{\psi_k^{(\alpha)}\left(x;4 p^2(\alpha+1)^2\varepsilon^2,2p(\alpha+1)-\frac{1}{2}\right)}{\Psi_k^{(\alpha)}(x)}
\end{equation}
is no more bounded, as $\alpha$ grows large, in any sector with $R\le|x|\le1$ for any fixed $R$, no matter how close is it to 1. In order to overcome this obstacle, we need to choose an $R$ depending on $\alpha$ so that it approaches $1$ from the left as $\alpha$ grows large. 

For any $k\in\mathbb{Z}$, $\alpha>0$ and $\widetilde{R}>0$ we denote by $\widetilde{\mathcal{D}}_k^{(\alpha)}(\widetilde{R})$ the closed annular subsector in $\widetilde{\mathbb{C}^*}$ defined by
\begin{equation}
\label{11-agosto-2024-24}
\widetilde{D}_k^{(\alpha)}(\widetilde{R}):=\left\{1-\frac{\widetilde{R}}{(\alpha+1)^2}\le |x|\le 1,\,\frac{(k-1)\pi}{\alpha+1}\le\arg(x)\le\frac{(k+1)\pi}{\alpha+1}\right\}\subset\widetilde{\mathbb{C}^*}.
\end{equation}

\begin{lemma}
\label{lemma-16-maggio-2024-1-bis}
Fix a compact subset $W\in\mathbb{C}\times\mathbb{C}^*$. There exist a number $\widetilde{R}_W>0$ and a constant $\alpha_W>0$ depending only on $W$ such that for all $k\in\mathbb{Z}$ and all $\alpha\ge \alpha_W$ the ratio $Z_k^{(\alpha)}\left(x;4 p^2(\alpha+1)^2\varepsilon^2,2p(\alpha+1)-\frac{1}{2}\right)$ (see~\eqref{reparameterized-ratio}) restricted to the closed domain $\widetilde{\mathcal{D}}_k^{(\alpha)}(\widetilde{R}_W)$ (see~\eqref{11-agosto-2024-24}) is uniformly bounded with respect to $(\varepsilon,p)\in W$, namely
\begin{equation}
\label{16-maggio-2024-1-bis}
\left| Z_k^{(\alpha)}\left(x;4 p^2(\alpha+1)^2\varepsilon^2,2p(\alpha+1)-\frac{1}{2}\right) \right|_{x\in\widetilde{\mathcal{D}}_k^{(\alpha)}(\widetilde{R}_W)} \lesssim_K 1,\quad\mbox{for all }\alpha\ge\alpha_W.
\end{equation} 
\end{lemma}

\begin{proof}
The proof is similar to the proof of the analogous Lemma \ref{lemma-16-maggio-2024-1}. In particular, it suffices to substitute $R$ with $1-\frac{\widetilde{R}}{(\alpha+1)^2}$ in the relevant inequalities~\eqref{11-agosto-2024-20},~\eqref{11-agosto-2024-21} and~\eqref{11-agosto-2024-22}. Notice that with this substitution a factor $\frac{1}{(\alpha+1)^2}$ in all the estimates appears, balancing the factor $(\alpha+1)^2$ introduced by the new parameterization.
\end{proof}

\begin{proposition}
\label{proposition-9-aprile-2024-1-bis}
Let $k\in\mathbb{Z}$ be a fixed integer number and let us fix a compact subset $W\subset\mathbb{C}\times\mathbb{C}^*$. Let $0<\widetilde{R}_W<1$ be a number (depending only on $W$) as in Lemma \ref{lemma-16-maggio-2024-1-bis}. Sibuya's solution $\psi_k^{(\alpha)}\left(x;4p^2(\alpha+1)^2\varepsilon^2,2p(\alpha+1)-\frac{1}{2}\right)$ restricted to the closed domain $\widetilde{\mathcal{D}}_k^{(\alpha)}(\widetilde{R}_W)$ (see~\eqref{11-agosto-2024-24}), along with its derivative, is asymptotic to $\widetilde{\psi}_k^{(\alpha)}\left(x;4p^2(\alpha+1)^2\varepsilon^2,2p(\alpha+1)-\frac{1}{2}\right)$ as $\alpha\to+\infty$ uniformly with respect to $x\in\widetilde{\mathcal{D}}_k^{(\alpha)}(\widetilde{R}_W)$ and $(\varepsilon,p)\in W$. More precisely, there exists a constant $\alpha_{W}(k)>0$, depending only on $W$ and $k$, such that
\begin{equation}
\label{9-aprile-2024-5-bis}
\begin{aligned}
& \left|\psi_k^{(\alpha)}\left(x;4p^2(\alpha+1)^2\varepsilon^2,2p(\alpha+1)-\frac{1}{2}\right)-\widetilde{\psi}_k^{(\alpha)}\left(x;4p^2(\alpha+1)^2\varepsilon^2,2p(\alpha+1)-\frac{1}{2}\right)\right| \\ & \lesssim_{W,k} \frac{\log(\alpha+1)}{\alpha+1}
\end{aligned}
\end{equation}
and
\begin{equation}
\label{9-aprile-2024-6}
\begin{aligned}
& \left|\frac{d}{dx}\psi_k^{(\alpha)}\left(x;4p^2(\alpha+1)^2\varepsilon^2,2p(\alpha+1)-\frac{1}{2}\right)-\frac{d}{dx}\widetilde{\psi}_k^{(\alpha)}\left(x;4p^2(\alpha+1)^2\varepsilon^2,2p(\alpha+1)-\frac{1}{2}\right)\right| \\ 
& \lesssim_{W,k}\frac{\log( \alpha+1)}{\alpha+1}
\end{aligned}
\end{equation}
hold for all $x\in \widetilde{\mathcal{D}}_k^{(\alpha)}(\widetilde{R}_W)$, $\alpha\ge\alpha_{W}(k)$ and all $(\varepsilon,p)\in W$.
\end{proposition}

\begin{proof}
See proof of Proposition \ref{proposition-9-aprile-2024-1}.
\end{proof}

\begin{corollary}
\label{corollary-13-agosto-2024-1-bis}
With the same assumptions and notations of Proposition \ref{proposition-9-aprile-2024-1-bis} specialized to $k=0$, inequalities
\begin{equation}
\label{13-agosto-2024-30}
\left|\psi_k^{(\alpha)}\left(1;4p^2(\alpha+1)^2\varepsilon^2,2p(\alpha+1)-\frac{1}{2}\right)\right|\lesssim_{W} \frac{\log(\alpha+1)}{\alpha+1}
\end{equation}
and
\begin{equation}
\label{13-agosto-2024-31}
\left|\left.\frac{d}{dx}\psi_k^{(\alpha)}\left(x;4p^2(\alpha+1)^2\varepsilon^2,2p(\alpha+1)-\frac{1}{2}\right)\right|_{x=1}+1\right|\lesssim_W\frac{\log(\alpha+1)}{\alpha+1}
\end{equation}
hold for all $\alpha\ge\alpha_W$ and $(\varepsilon,p)\in W$.
\end{corollary}

\begin{proof}
It is an immediate consequence of Proposition \ref{proposition-9-aprile-2024-1-bis}.
\end{proof}

Finally, we have the asymptotics for the Stokes multiplier as a function of the rescaled energy parameter $\varepsilon$ and the rescaled angular momentum parameter $p$:

\begin{proposition}
\label{proposition-16-maggio-2024-10-bis}
Let $k\in\mathbb{Z}$ be a fixed integer number and let $W\subset\mathbb{C}\times\mathbb{C}^*$ be a fixed compact subset.

The Stokes multiplier $\sigma_k^{(\alpha)}\left(4p^2(\alpha+1)^2\varepsilon^2;2p(\alpha+1)-\frac{1}{2}\right)$ converges to the constant $-2 i $ as $\alpha\to+\infty$, uniformly with respect to $(\varepsilon,p)\in W$; more precisely, there exists a constant $\alpha_{W}(k)>1$ depending only on $W$ and $k$, such that
\begin{equation}
\label{21-aprile-2024-6}
\left|\frac{\sigma_k^{(\alpha)}\left(4p^2(\alpha+1)^2\varepsilon^2;2p(\alpha+1)-\frac{1}{2}\right)}{2 i}+1\right|\lesssim_{W,k} \frac{\log( \alpha+1)}{\alpha+1}
\end{equation}
holds for all $\alpha\ge\alpha_{W}(k)$ and all $(\varepsilon,p)\in W$.
\end{proposition}

\begin{proof}
See proof of Proposition \ref{proposition-16-maggio-2024-10}.
\end{proof}

\subsection{Local theory at the (formal) regular singularity - A WKB approach}
\label{subsection-3.2}

In this section we study equation~\eqref{15-ott-2023-3} around the (formal) regular singular point $x=0$ by means of the WKB method. In particular, we will give some uniform asymptotics for the Frobenius subdominant solution $\chi_+^{(\alpha)}(x;E,\ell)$ defined in~\eqref{19-agosto-2024-1} with the parameterization~\eqref{15-ott-2023-2} of the energy and angular momentum parameters. Notice that the normalization condition~\eqref{12-maggio-2024-1} reads
\begin{equation}
\label{22-maggio-2024-12}
\lim_{x\to 0} \Gamma\left(1+2 p\right) x^{-\frac{1}{2}-2p(\alpha+1)}\chi_+^{(\alpha)}\left(x;4p^2(\alpha+1)^2\varepsilon^2,2p(\alpha+1)-\frac{1}{2}\right)=1.
\end{equation}
Let us start by writing equation~\eqref{15-ott-2023-3} as 
\begin{equation}
\label{18-gennaio-2024-1}
\frac{d^2\psi}{dx^2}=\left[4p^2(\alpha+1)^2 \widehat{V}^{(\alpha)}(x;\varepsilon,p)-\frac{1}{4 x^2}\right]\psi
\end{equation}
where the $\widehat{V}^{(\alpha)}(x;\varepsilon,p)$ is the function defined by
\begin{equation}
\label{3-febbraio-2024-1}
\widehat{V}^{(\alpha)}(x;\varepsilon,p):=\frac{1-(\varepsilon x)^2}{x^2}+\frac{ x^{2 \alpha}}{4 p^2(\alpha+1)^2}
\end{equation}
and is referred to as the \textit{rescaled potential} of the equation. The term $r(x):=-\frac{1}{4 x^2}$ is usually called the \textit{Langer correction}. Due to the symmetry $p\to-p$, we can assume that the quasi momentum $p$ ranges in the right half plane
\begin{equation}
\label{18-maggio-2024-10}
\mathbb{H}_r:=\left\{|\arg(p)|<\frac{\pi}{2},\,p\ne 0\right\}.
\end{equation}
The cases $\arg(p)=\pm\frac{\pi}{2}$ are not considered here.

In order to perform a WKB analysis to equation~\eqref{18-gennaio-2024-1}, we need first of all to localize the zeros of the rescaled potential $\widehat{V}^{(\alpha)}(x;\varepsilon,p)$, which in this context are called the \textit{turning points} of equation~\eqref{18-gennaio-2024-1}, namely we have to study the following equation:
\begin{equation}
\label{20-luglio-2024-1}
\widehat{V}^{(\alpha)}(x;\varepsilon,p)=\frac{1-(\varepsilon x)^2}{x^2}+\frac{x^{2 \alpha}}{4 p^2(\alpha+1)^2}=0.
\end{equation}
A turning point is said to be \textit{simple} if it is a simple zero of the rescaled potential. Since our aim is to study the connection problem between the subdominant Frobenius solution at $x=0$ and Sibuya's solution which is subdominant as $x\to+\infty$, we restrict the analysis of equation~\eqref{20-luglio-2024-1} to the sector $\Sigma_0^{(\alpha)}$ of the punctured $x$-plane (see~\eqref{3-aprile-2024-2-bis} for the definition of $\Sigma_0^{(\alpha)}$). 

\begin{proposition}
\label{proposition-22-maggio-2024-1}

Let us fix constants $L>0$ and $0<\Theta<\frac{\pi}{2}$. There exists a constant $\alpha_{L}>0$ depending only on $L$ (and $\Theta$) such that the following holds for all $\alpha\ge\alpha_L$, $|\varepsilon|\ge 1$, $|\arg(\varepsilon)|\le\frac{\Theta}{\alpha+1}$ and all $|p|\ge L$, $|\arg(p)|\le \Theta$:
\begin{itemize}
\item[i)] There exists precisely one simple turning point $x_+^{(\alpha)}(\varepsilon,p)$ satisfying
\begin{equation}
\label{18-gennaio-2024-3}
\left|\pm\varepsilon x_\pm^{(\alpha)}(\varepsilon,p)-1\right|\lesssim_{L,\Theta} \frac{1}{|p|^2(\alpha+1)^2};
\end{equation}
\item[ii)] There exists precisely one simple turning point $x_0^{(\alpha)}(\varepsilon,p)$ satisfying
\begin{equation}
\label{19-gennaio-2024-1}
\left|\left(2(\alpha+1)p\varepsilon\right)^{-\frac{1}{\alpha}} x_0^{(\alpha)}(\varepsilon,p)-1\right|\lesssim_{L,\Theta} \frac{1}{\alpha};
\end{equation}
\item[iii)] The points $x_+^{(\alpha)}(\varepsilon,p)$ and $x_0^{(\alpha)}(\varepsilon,p)$ are the only turning points in $\Sigma_0^{(\alpha)}$ (see Figure \ref{turning-points-sigma-0}).
\end{itemize}
\end{proposition}
\begin{proof}
The proof is postponed to Appendix \ref{appendix-technical-proofs}.
\end{proof}

\begin{figure}[H]
\centering
\begin{tikzpicture}[decoration={markings, mark= at position 0.5 with {\arrow{stealth}}},scale=0.35, every node/.style={scale=0.65}]

\draw[dashed] (0,0) -- (14:19);
\draw[dashed] (0,0) -- (-14:19);
\draw[dashed] (0,0) -- (0:19);

\filldraw[color=green!90, fill=green!10, dashed] (5:8) circle (1.0);

\filldraw[color=violet!90, fill=violet!10, dashed] (-2:14) circle (1.3);

\fill (0,0)  circle[radius=3pt];

\fill (5:8)  circle[radius=3pt];

\fill (-2:14)  circle[radius=3pt];

\fill (0:12)  circle[radius=3pt];

\draw (0,0)  node[left] {$0$};

\draw (5:8)  node[above right] {\mbox{$x_+^{(\alpha)}(\varepsilon,p)$}};

\draw (0:12)  node[below] {$1$};

\draw (-2:14) node[below right] {\mbox{$x_0^{(\alpha)}(\varepsilon,p)$}};

\draw (20,0) node[right] {\mbox{\LARGE $\Sigma_0^{(\alpha)}$}};

\end{tikzpicture}
\caption{\small The turning points $x_+^{(\alpha)}(\varepsilon,p)$, $x_0^{(\alpha)}(\varepsilon,p)$ in the sector $\Sigma_0^{(\alpha)}$ described in Proposition \ref{proposition-22-maggio-2024-1}. The green area corresponds to the disc to which $x_+^{(\alpha)}(\varepsilon,p)$ can be localized according to inequality~\eqref{18-gennaio-2024-3}, while the violet area corresponds to the disc to which $x_0^{(\alpha)}(\varepsilon,p)$ can be localized according to inequality~\eqref{19-gennaio-2024-1}. }
\label{turning-points-sigma-0}
\end{figure}
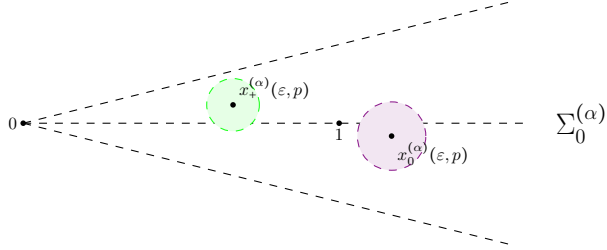

The second ingredient for a WKB analysis to equation~\eqref{18-gennaio-2024-1} is the \textit{action integral} of the rescaled potential. For all $\alpha,\varepsilon,p$ as in Proposition \ref{proposition-22-maggio-2024-1}, the action integral is the multiple valued function $S^{(\alpha)}(x;\varepsilon,p)$ defined by
\begin{equation}
\label{19-gennaio-2024-2}
\frac{2}{3}\left[-S^{(\alpha)}(x;\varepsilon,p) \right]^{\frac{3}{2}}:=\int_{\gamma,x_+^{(\alpha)}(\varepsilon,p)}^{x}\left(\widehat{V}^{(\alpha)}(t;\varepsilon,p)\right)^{\frac{1}{2}}dt,
\end{equation}
where $x_+^{(\alpha)}(\varepsilon,p)$ is as in point i) of Proposition \ref{proposition-22-maggio-2024-1} and $\gamma$ is any oriented, piecewise differentiable path joining $x_+^{(\alpha)}(\varepsilon,p)$ to $x$. In order for~\eqref{19-gennaio-2024-2} to be a well defined single valued function on $\Sigma_0^{(\alpha)}$ we have to introduce a cut joining the turning points $x_+^{(\alpha)}(\varepsilon,p)$ and $x_0^{(\alpha)}(\varepsilon,p)$. We cut the sector $\Sigma_0^{(\alpha)}$ along the curve $\mathfrak{T}^{(\alpha)}(\varepsilon,p)$ defined by 
\begin{equation}
\label{27-luglio-2024-1}
\begin{aligned}
 \mathfrak{T}^{(\alpha)}(\varepsilon,p) := & \left\{ x=|x_+^{(\alpha)}(\varepsilon,p)| e^{i \theta},\,\arg(x_+^{(\alpha)}(\varepsilon,p))\le\theta\le \frac{\pi}{2\alpha+2} \right\} \cup \\ & \left\{x=r e^{\frac{i \pi}{2\alpha+2}},\,|x_+^{(\alpha)}(\varepsilon,p)|\le r\le |x_0^{(\alpha)}(\varepsilon,p)|\right\} \cup \\ & \left\{x=|x_0^{(\alpha)}(\varepsilon,p)| e^{i \phi},\, \arg(x_0^{(\alpha)}(\varepsilon,p)) \le\phi \le \frac{\pi}{2 \alpha+2} \right\},
\end{aligned}
\end{equation}
see Figure \ref{taglio-sigma-0} below.

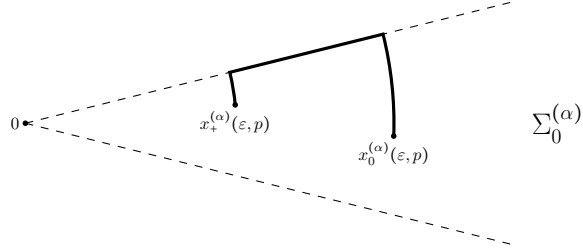
\begin{figure}[H]
\centering
\begin{tikzpicture}[decoration={markings, mark= at position 0.5 with {\arrow{stealth}}},scale=0.35, every node/.style={scale=0.65}] 

\draw (0,0)  node[left] {$0$};

\draw (5:8)  node[below] {$x_+^{(\alpha)}(\varepsilon,p)$};

\draw (-2:14) node[below] {$x_0^{(\alpha)}(\varepsilon,p)$};

\fill (0,0)  circle[radius=3pt];

\fill (5:8)  circle[radius=3pt];

\fill (-2:14)  circle[radius=3pt];

\draw (19,0) node[right] {\mbox{\LARGE $\Sigma_0^{(\alpha)}$}};

\draw[dashed] (0,0) -- (14:19);
\draw[dashed] (0,0) -- (-14:19);

\draw[very thick] (5:8) arc[start angle=5, end angle=14, radius=8] -- (14:14) -- (14:14) arc[start angle=14, end angle=-2, radius=14];

\end{tikzpicture}
\caption{The cut $\mathfrak{T}^{(\alpha)}(\varepsilon,p)$ defined in~\eqref{27-luglio-2024-1}.}
\label{taglio-sigma-0}
\end{figure}

\begin{proposition}
\label{proposition-20-luglio-2024-1}
Let us fix constants $L>0$ and $0<\Theta<\frac{\pi}{2}$. There exists a constant $\alpha_{L}>0$ depending only on $L$ (and $\Theta$) such that the following holds for all $\alpha\ge\alpha_L$, $|\varepsilon|\ge 1$, $|\arg(\varepsilon)|\le\frac{\Theta}{\alpha+1}$ and all $|p|\ge L$, $|\arg(p)|\le \Theta$:
\begin{itemize}
\item[i)] Inequality
\begin{equation}
\label{6-febbraio-2024-1}
\left|\frac{2}{3}\left[-S^{(\alpha)}(x;\varepsilon,p)\right]^{\frac{3}{2}} \pm\left[\sqrt{1-(\varepsilon x)^2}-\operatorname{arctanh}\sqrt{1-(\varepsilon x)^2}\right] \right| \lesssim_{L,\Theta} \frac{1}{|p|(\alpha+1)^2},
\end{equation}
holds for all $x$ in the region
\begin{equation}
\label{2-agosto-2024-3-bis}
\left(\Sigma_0^{(\alpha)}\setminus\mathfrak{T}^{(\alpha)}(\varepsilon,p) \right)\cap\left\{ |x|\le\frac{1}{|\varepsilon|}\left(1-\frac{1}{|p|^2(\alpha+1)^2}\right) \right\}
\end{equation}
(see Figure \ref{figure-2-agosto-2024-3-bis});
\item[ii)] Inequality 
\begin{equation}
\label{6-febbraio-2024-1-bis}
\left|\frac{2}{3}\left[-S^{(\alpha)}(x;\varepsilon,p)\right]^{\frac{3}{2}}\right|\lesssim_{L,\Theta}\frac{1}{|p|(\alpha+1)^2}
\end{equation}
holds for all $x$ in the region
\begin{equation}
\label{2-agosto-2024-1-bis}
\left(\Sigma_0^{(\alpha)}\setminus\mathfrak{T}^{(\alpha)}(\varepsilon,p)\right)\cap\left\{ ||\varepsilon x|-1|\le\frac{1}{ |p|^2(\alpha+1)^2}\right\}
\end{equation}
(see Figure \ref{figure-2-agosto-2024-1-bis}).
\end{itemize}
Here $\mathfrak{T}^{(\alpha)}(\varepsilon,p)$ is the cut~\eqref{27-luglio-2024-1} and \textquotedblleft$\pm$\textquotedblright\, refer to the choice of the positive and negative branch of the square root, respectively.
\end{proposition}

\begin{proof}
The proof is postponed to Appendix \ref{appendix-technical-proofs}.
\end{proof}

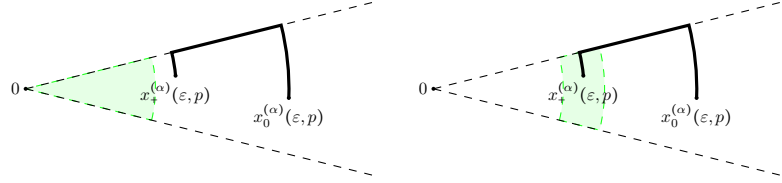
\begin{figure}[H]
\centering
\begin{subfigure}[t]{0.45\textwidth}
\raggedleft
\begin{tikzpicture}[decoration={markings, mark= at position 0.5 with {\arrow{stealth}}},scale=0.25, every node/.style={scale=0.65}] 
\filldraw[color=green!90, fill=green!10, dashed] (-14:6.9) arc[start angle=-14, end angle=14, radius=6.9] -- (0,0) -- (-14:6.9);

\draw (0,0)  node[left] {$0$};

\draw (5:8)  node[below] {$x_+^{(\alpha)}(\varepsilon,p)$};

\draw (-2:14) node[below] {$x_0^{(\alpha)}(\varepsilon,p)$};

\fill (0,0)  circle[radius=3pt];

\fill (5:8)  circle[radius=3pt];

\fill (-2:14)  circle[radius=3pt];

\draw[dashed] (0,0) -- (14:19);
\draw[dashed] (0,0) -- (-14:19);

\draw[very thick] (5:8) arc[start angle=5, end angle=14, radius=8] -- (14:14) -- (14:14) arc[start angle=14, end angle=-2, radius=14];

\end{tikzpicture}
\caption{\small The region (green area) defined in~\eqref{2-agosto-2024-3-bis}.}
\label{figure-2-agosto-2024-3-bis}
\end{subfigure}
\hspace{0.3cm}
\begin{subfigure}[t]{0.45\textwidth}
\begin{tikzpicture}[decoration={markings, mark= at position 0.5 with {\arrow{stealth}}},scale=0.25, every node/.style={scale=0.65}] 
\filldraw[color=green!90, fill=green!10, dashed] (-14:6.9) arc[start angle=-14, end angle=14, radius=6.9] -- (14:9.1) -- (14:9.1) arc[start angle=14, end angle=-14, radius=9.1] -- (-14:6.9);

\draw (0,0)  node[left] {$0$};

\draw (5:8)  node[below] {$x_+^{(\alpha)}(\varepsilon,p)$};

\draw (-2:14) node[below] {$x_0^{(\alpha)}(\varepsilon,p)$};

\fill (0,0)  circle[radius=3pt];

\fill (5:8)  circle[radius=3pt];

\fill (-2:14)  circle[radius=3pt];

\draw[dashed] (0,0) -- (14:19);
\draw[dashed] (0,0) -- (-14:19);

\draw[very thick] (5:8) arc[start angle=5, end angle=14, radius=8] -- (14:14) -- (14:14) arc[start angle=14, end angle=-2, radius=14];

\end{tikzpicture}
\caption{\small The region (green area) defined in~\eqref{2-agosto-2024-1-bis}.}
\label{figure-2-agosto-2024-1-bis}
\end{subfigure}

\caption{\small The regions (green areas) of the sector $\Sigma_0^{(\alpha)}$ defined in~\eqref{2-agosto-2024-3-bis} and~\eqref{2-agosto-2024-1-bis}. The solid black line represents the cut $\mathfrak{T}^{(\alpha)}(\varepsilon,p)$ defined in~\eqref{27-luglio-2024-1}. }
\label{regions}
\end{figure}

\begin{corollary}
\label{corollary-17-marzo-2024-1}
With the same assumptions and notations of Proposition \ref{proposition-20-luglio-2024-1}, the following holds for all $\alpha\ge\alpha_L$, $|\varepsilon|\ge 1$, $|\arg(\varepsilon)|\le\frac{\Theta}{\alpha+1}$ and all $|p|\ge L$, $|\arg(p)|\le \Theta$:
\begin{itemize}
\item[i)] For any fixed number $0<\nu<1$, by eventually taking a larger $\alpha_L$ (depending on $\nu$) we have
\begin{equation}
\label{1-agosto-2024-2}
\operatorname{Re}\left(\frac{2}{3}\left[-S^{(\alpha)}(x;\varepsilon,p)\right]^{\frac{3}{2}}\right)\gtrsim_{L,\Theta,\nu}\left|\log|\varepsilon x|\right|\quad\mbox{positive branch of square root}
\end{equation}
and
\begin{equation}
\label{1-agosto-2024-1}
\operatorname{Re}\left(\frac{2}{3}\left[-S^{(\alpha)}(x;\varepsilon,p)\right]^{\frac{3}{2}}\right)\lesssim_{L,\Theta,\nu}\log|\varepsilon x|\quad\mbox{negative branch of square root}
\end{equation}
for all $x\in\Sigma_0^{(\alpha)}\setminus\mathfrak{T}^{(\alpha)}(\varepsilon,p)$ with $|x|\le \frac{\nu}{|\varepsilon|}\left(1-\frac{1}{ |p|^2 (\alpha+1)^2}\right)$;
\item[ii)] With the choice of the positive branch of the square root, inequality
\begin{equation}
\label{17-marzo-2024-4}
\operatorname{Re}\left(\frac{2}{3}\left[-S^{(\alpha)}(x_2;\varepsilon,p)\right]^{\frac{3}{2}}\right)-\operatorname{Re}\left(\frac{2}{3}\left[-S^{(\alpha)}(x_1;\varepsilon,p)\right]^{\frac{3}{2}}\right)\lesssim_{L,\Theta} \frac{1}{|p|(\alpha+1)^2}
\end{equation}
holds for all $x_1,x_2\in \Sigma_0^{(\alpha)}\setminus\mathfrak{T}^{(\alpha)}(\varepsilon,p)$ with $|x_1|\le |x_2|\le\frac{1}{|\varepsilon|}\left(1-\frac{1}{|p|^2(\alpha+1)^2}\right)$ and $\arg(x_1)=\arg(x_2)$;
\item[iii)] For both choices of the branch of the square root
\begin{equation}
\label{19-marzo-2024-3}
\left|\operatorname{Im}\frac{2}{3}\left[-S^{(\alpha)}(x;\varepsilon,p)\right]^{\frac{3}{2}}\right|\lesssim_{L,\Theta} \frac{1}{|p|(\alpha+1)},
\end{equation}
for all $x\in\Sigma_0^{(\alpha)}\setminus\mathfrak{T}^{(\alpha)}(\varepsilon,p)$ with $|x|\le\frac{1}{|\varepsilon|}\left(1+\frac{1}{|p|^2(\alpha+1)^2}\right)$.
\end{itemize}
\end{corollary}

\begin{proof}
The proof is an immediate consequence of inequality~\eqref{6-febbraio-2024-1} of Proposition \ref{proposition-20-luglio-2024-1}; more precisely, properties i)-iii) follow by the analogous properties of the function $\sqrt{1-(\varepsilon x)^2}-\operatorname{arctanh}\sqrt{1-(\varepsilon x)^2}$, which can be verified by elementary methods.
\end{proof}

Before stating the main Proposition of this section, let us introduce some notation: let $S^{(\alpha)}(x;\varepsilon,p)$ be the action integral defined in~\eqref{19-gennaio-2024-2}, for the rescaled potential $\widehat{V}^{(\alpha)}(x;\varepsilon,p)$ defined in~\eqref{3-febbraio-2024-1}; we denote by $\widehat{X}_+^{(\alpha)}(x;\varepsilon,p)$ and $\widehat{X}_-^{(\alpha)}(x;\varepsilon,p)$ the functions defined by
\begin{equation}
\label{30-gennaio-2024-4}
\begin{aligned}
\widehat{X}_+^{(\alpha)}(x;\varepsilon,p):= & \frac{2 \sqrt{\pi} \left[2p(\alpha+1)\right]^{\frac{1}{6}}}{\Gamma\left(1+2p\right)}\varepsilon^{-2 p (\alpha+1)} \times \\
& \times \left(\frac{d}{dx}S^{(\alpha)}(x;\varepsilon,p)\right)^{-\frac{1}{2}}\operatorname{Ai}\left[-\left(2p(\alpha+1)\right)^{\frac{2}{3}}S^{(\alpha)}(x;\varepsilon,p)\right]
\end{aligned}
\end{equation}
and
\begin{equation}
\label{30-gennaio-2024-5}
\begin{aligned}
\widehat{X}_-^{(\alpha)}(x;\varepsilon,p):= & \frac{\Gamma\left(1+2p\right)}{2 \sqrt{\pi} \left[2p(\alpha+1)\right]^{\frac{1}{6}}}\varepsilon^{2 p (\alpha+1)} \times \\
& \times \left(\frac{d}{dx}S^{(\alpha)}(x;\varepsilon,p)\right)^{-\frac{1}{2}}\operatorname{Bi}\left[-\left(2p(\alpha+1)\right)^{\frac{2}{3}}S^{(\alpha)}(x;\varepsilon,p)\right]
\end{aligned}
\end{equation}

\begin{lemma}
\label{lemma-1-agosto-2024-1}
Let us fix constants $L>0$,  $0<\Theta<\frac{\pi}{2}$ and let us fix a number $0<\nu<1$. For the choice of the positive branch of the square root, there exists a constant $\alpha_{L}(\nu)>0$ depending only on $L$ (and $\Theta$, $\nu$) such that the following holds for all $\alpha\ge\alpha_L(\nu)$, $|\varepsilon|\ge 1$, $|\arg(\varepsilon)|\le\frac{\Theta}{\alpha+1}$ and all $|p|\ge L$, $|\arg(p)|\le \Theta$:
\begin{itemize}
\item[i)] $\left|\arg\left(p^{\frac{2}{3}}\left[-S^{(\alpha)}(x;\varepsilon,p)\right]\right)\right|<\frac{\pi}{3}$ for all $x\in\Sigma_0^{(\alpha)}\setminus\mathfrak{T}^{(\alpha)}(\varepsilon,p)$ with \\ $|x|\le\frac{\nu}{|\varepsilon|}\left(1-\frac{1}{|p|^2(\alpha+1)^2}\right)$;
\item[ii)] The functions $\widehat{X}_\pm^{(\alpha)}(x;\varepsilon,p)$ in~\eqref{30-gennaio-2024-4} and~\eqref{30-gennaio-2024-5} have no zeros for all $x\in\Sigma_0^{(\alpha)}\setminus\mathfrak{T}^{(\alpha)}(\varepsilon,p)$ with $|x|\le\frac{1}{|\varepsilon|}\left(1+\frac{1}{|p|^2(\alpha+1)^2}\right)$;
\item[iii)] $\widehat{X}_+^{(\alpha)}(x;\varepsilon,p)$ satisfies the normalization condition
\begin{equation}
\label{1-agosto-2024-3}
\lim_{x\to 0} \Gamma(1+2 p) x^{-\frac{1}{2}-2p(\alpha+1)} \widehat{X}_+^{(\alpha)}(x;\varepsilon,p)=1.
\end{equation} 
\end{itemize}
\end{lemma}
\begin{proof}
Point i) is an immediate consequence of Corollary \ref{corollary-17-marzo-2024-1} (inequalities~\eqref{1-agosto-2024-2} and~\eqref{19-marzo-2024-3}). Point ii) follows from the previous property, inequality~\eqref{6-febbraio-2024-1} of Proposition \ref{proposition-20-luglio-2024-1} and Proposition \ref{zeros-airy}. Point iii) follows from point i) and  representation~\eqref{29-gennaio-2024-1} for the Airy functions.
\end{proof}

\begin{proposition}
\label{proposition-30-gennaio-2024-1}
Let us fix constants $L>0$ and $0<\Theta<\frac{\pi}{2}$ and let us choose the positive branch of the square root in the definition of the action integral $S^{(\alpha)}(x;\varepsilon,p)$ (see~\eqref{19-gennaio-2024-2}).

There exists a constant $\alpha_{L}>0$ depending only on $L$ (and $\Theta$) such that 
\begin{equation}
\label{9-febbraio-2024-1}
\left|\frac{\chi_+^{(\alpha)}\left(x;4p^2(\alpha+1)^2\varepsilon^2,2p(\alpha+1)-\frac{1}{2}\right)}{\widehat{X}_+^{(\alpha)}(x;\varepsilon,p)}-1\right|\lesssim_{L,\Theta}\frac{1}{|p|(\alpha+1)} 
\end{equation}
for all
\[
x\in \Sigma_0^{(\alpha)}\setminus\mathfrak{T}^{(\alpha)}(\varepsilon,p) \quad\mbox{with}\quad |x|\le\frac{1}{|\varepsilon|}\left(1+\frac{1}{|p|^2(\alpha+1)^2}\right)
\]
all $\alpha\ge\alpha_{L}$, $|\varepsilon|\ge 1$, $|\arg(\varepsilon)|\le\frac{\Theta}{\alpha+1}$ and all $|p|\ge L$, $|\arg(p)|\le \Theta$(here $\mathfrak{T}^{(\alpha)}(\varepsilon,p)$ is as in~\eqref{27-luglio-2024-1}).
\end{proposition}

\begin{proof}
From Lemma \ref{lemma-1-agosto-2024-1}, point ii), the function $\widehat{X}_+^{(\alpha)}(x;\varepsilon,p)$ is different from zero for all $x\in\Sigma_0^{(\alpha)}\setminus\mathfrak{T}^{(\alpha)}(\varepsilon,p)$ with $|x|\le \frac{1}{|\varepsilon|}\left(1+\frac{1}{|p|^2(\alpha+1)^2}\right)$, thus we can define the ratio
\[
\hat{z}_+^{(\alpha)}(x;\varepsilon,p):=\frac{\chi_+^{(\alpha)}\left(x;4p^2(\alpha+1)^2\varepsilon^2,2p(\alpha+1)-\frac{1}{2}\right)}{\widehat{X}_{+}^{(\alpha)}(x;\varepsilon,p)}.
\]
By direct inspection, we can verify that it satisfies the Volterra integral equation
\begin{equation}
\label{19-marzo-2024-2}
\hat{z}_+^{(\alpha)}(x;\varepsilon,p)=1+\int_{\gamma,0}^{x}\frac{\widehat{X}_{+}^{(\alpha)}(t;\varepsilon,p)}{\widehat{X}_+^{(\alpha)}(x;\varepsilon,p)} \widehat{K}^{(\alpha)}(x,t;\varepsilon,p) \widehat{F}^{(\alpha)}(t;\varepsilon,p) \hat{z}_+^{(\alpha)}(t;\varepsilon,p) dt,
\end{equation}
where the kernel $\widehat{K}^{(\alpha)}(x,t;\varepsilon,p)$ is defined as
\begin{equation}
\label{31-gennaio-2024-2}
\begin{aligned}
&\frac{1}{\pi}  \left(2p(\alpha+1)\right)^{\frac{2}{3}} \widehat{K}^{(\alpha)}(x,t;\varepsilon,p):= \\ & = \left[\widehat{X}_+^{(\alpha)}(x;\varepsilon,p)\widehat{X}_-^{(\alpha)}(t;\varepsilon,p)-\widehat{X}_+^{(\alpha)}(t;\varepsilon,p)\widehat{X}_-^{(\alpha)}(x;\varepsilon,p)\right],
\end{aligned}
\end{equation} 
the forcing term $\widehat{F}^{(\alpha)}(t;\varepsilon,p)$ is
\begin{equation}
\label{31-gennaio-2024-3}
\widehat{F}^{(\alpha)}(t;\varepsilon,p):= \frac{1}{2}\left\{S^{(\alpha)}(t;\varepsilon,p),t\right\}-\frac{1}{4 t^2},
\end{equation}
being $\left\{-,t\right\}$ the Schwarzian derivative with respect to $t$, and $\gamma$ is any path joining $0$ to $x$ with support contained in $\left(\Sigma_0^{(\alpha)}\setminus\mathfrak{T}^{(\alpha)}(\varepsilon,p)\right)\cap\left\{|t|\le\frac{1}{|\varepsilon|}\left(1+\frac{1}{|p|^2(\alpha+1)^2}\right)\right\}$ and with finite length. In order to obtain an estimate for $\left|\hat{z}_+^{(\alpha)}(x;\varepsilon,p)\right|$, we are going to study the function
\begin{equation}
\label{9-febbraio-2024-5}
\widehat{R}^{(\alpha)}(x;\varepsilon,p):=\int_{\gamma,0}^{x}\left|\frac{\widehat{X}_{+}^{(\alpha)}(t;\varepsilon,p)}{\widehat{X}_+^{(\alpha)}(x;\varepsilon,p)} \widehat{K}^{(\alpha)}(x,t;\varepsilon,p) \widehat{F}^{(\alpha)}(t;\varepsilon,p)\right| |dt|.
\end{equation}
First of all, we notice that, due to Lemma \ref{lemma-1-agosto-2024-1}, points i) and ii), boundedness property~\eqref{useful-bounds} of the Airy functions and Corollary \ref{corollary-17-marzo-2024-1}, points ii) and iii), for any $t,x\in\Sigma_0^{(\alpha)}\setminus\mathfrak{T}^{(\alpha)}(\varepsilon,p)$ with $|t|\le|x|\le\frac{1}{|\varepsilon|}\left(1-\frac{1}{ |p|^2(\alpha+1)^2}\right)$ we have
\begin{equation}
\label{5-agosto-2024-1}
\left|\frac{\widehat{X}_+^{(\alpha)}(t;\varepsilon,p)}{\widehat{X}_+^{(\alpha)}(x;\varepsilon,p)}\widehat{K}^{(\alpha)}(x,t;\varepsilon,p)\right|\lesssim_{L,\Theta}\frac{1}{|p(\alpha+1)|^{\frac{2}{3}}}\left|\widehat{X}_+^{(\alpha)}(t;\varepsilon,p)\widehat{X}_-^{(\alpha)}(t;\varepsilon,p)\right|,
\end{equation}
for all $\alpha\ge\alpha_L>0$, $|\varepsilon|\ge 1$, $|\arg(\varepsilon)|\le\frac{\Theta}{\alpha+1}$ and all $|p|\ge L$, $|\arg(p)|\le \Theta$.

Let us fix a number $0<\nu<1$ and let us start considering points $x$ in $\Sigma_0^{(\alpha)}\setminus\mathfrak{T}^{(\alpha)}(\varepsilon,p)$ with $|x|\le \frac{\nu}{|\varepsilon|}\left(1-\frac{1}{|p|^2(\alpha+1)^2}\right)$. We can choose the path $\gamma$ to be the line segment joining $0$ to $x$. From representations~\eqref{29-gennaio-2024-1} for the Airy functions, with bounds~\eqref{29-gennaio-2024-2} for the remainder terms, from Corollary \ref{corollary-17-marzo-2024-1}, point i), and from inequality~\eqref{5-agosto-2024-1}  we receive a constant $\alpha_L(\nu)>0$ depending only on $L$ (and $\Theta$, $\nu$) such that
\[
\left|\frac{\widehat{X}_+^{(\alpha)}(t;\varepsilon,p)}{\widehat{X}_+^{(\alpha)}(x;\varepsilon,p)}\widehat{K}^{(\alpha)}(x,t)\right|\lesssim_{L,\Theta,\nu} \frac{|t|}{|p|(\alpha+1)},\quad t,x\in\gamma,\quad |t|\le |x|
\]
and
\begin{equation}
\label{5-agosto-2024-2}
\left|\widehat{F}^{(\alpha)}(t;\varepsilon,p)\right|\lesssim_{L,\Theta,\nu} \frac{1}{|t|^2|\log(|\varepsilon t|)|^2},\quad t\in\Sigma_0^{(\alpha)},\quad |t|\le\frac{\nu}{|\varepsilon|}\left(1-\frac{1}{|p|^2(\alpha+1)^2}\right),
\end{equation}
for all $\alpha\ge\alpha_L$, $|\varepsilon|\ge 1$, $|\arg(\varepsilon)|\le\frac{\Theta}{\alpha+1}$ and all $|p|\ge L$, $|\arg(p)|\le \Theta$. Plugging this inequalities into~\eqref{9-febbraio-2024-5} we obtain
\[
\left|\widehat{R}^{(\alpha)}(x;\varepsilon,p)\right|\lesssim_{L,\Theta,\nu} \frac{1}{|p|(\alpha+1)}.
\]
From this estimate and from the integral equation~\eqref{19-marzo-2024-2}, it follows that
\[
\sup_{\alpha\ge\alpha_{L}(\nu)}\sup_{\substack{|\varepsilon|\ge 1,\,|\arg(\varepsilon)|\le\frac{\Theta}{\alpha+1} \\ |p|\ge L,\,|\arg(p)|\le \Theta}}\sup_{\substack{x\in\Sigma_0^{(\alpha)} \\ |x|\le\frac{\nu}{|\varepsilon|}\left(1-\frac{1}{|p|^2(\alpha+1)^2}\right)}} \left|\hat{z}^{(\alpha)}(x;\varepsilon,p)\right|\le 2
\]
and hence, using again the integral equation~\eqref{19-marzo-2024-2}, the conclusion follows.

For 
\[
x\in\Sigma_0^{(\alpha)}\setminus\mathfrak{T}^{(\alpha)}(\varepsilon,p)\quad\mbox{with}\quad \frac{\nu}{|\varepsilon|}\left(1-\frac{1}{|p|^2(\alpha+1)^2}\right)\le |x| \le \frac{1}{|\varepsilon|}\left(1-\frac{1}{|p|^2(\alpha+1)^2}\right)
\]
we split the integral in~\eqref{9-febbraio-2024-5} as
\begin{equation}
\label{1-agosto-2024-4}
\begin{aligned}
\widehat{R}^{(\alpha)}(x;\varepsilon,p)= & \int_{\gamma_1} \left|\frac{\widehat{X}_{+}^{(\alpha)}(t;\varepsilon,p)}{\widehat{X}_+^{(\alpha)}(x;\varepsilon,p)} \widehat{K}^{(\alpha)}(x,t;\varepsilon,p) \widehat{F}^{(\alpha)}(t;\varepsilon,p)\right| |dt| \\
& + \int_{\gamma_2} \left|\frac{\widehat{X}_{+}^{(\alpha)}(t;\varepsilon,p)}{\widehat{X}_+^{(\alpha)}(x;\varepsilon,p)} \widehat{K}^{(\alpha)}(x,t;\varepsilon,p) \widehat{F}^{(\alpha)}(t;\varepsilon,p)\right| |dt|
\end{aligned}
\end{equation}
where $\gamma_1$ and $\gamma_2$ are the rays
\begin{equation}
\label{4-agosto-2024-1}
\begin{aligned}
& \gamma_1:=\left\{ r e^{i \arg\left(x\right)},\,0< r\le \frac{\nu}{|\varepsilon|}\left(1-\frac{1}{|p|^2(\alpha+1)^2}\right) \right\} \\[2ex]
& \gamma_2:=\left\{ r e^{i \arg(x)},\, \frac{\nu}{|\varepsilon|}\left(1-\frac{1}{|p|^2(\alpha+1)^2}\right)\le r\le |x| \right\},
\end{aligned}
\end{equation}
see Figure \ref{figure-4-agosto-2024-1} below.

\begin{figure}[H]
\centering
\begin{tikzpicture}[decoration={markings, mark= at position 0.5 with {\arrow{stealth}}},scale=0.35, every node/.style={scale=0.65}] 

\draw[dashed] (-19:7.5) arc[start angle=-19, end angle=19, radius=7.5];

\draw (-19:7.5) node[below left] {$|t|=\frac{\nu}{|\varepsilon|}\left(1-\frac{1}{|p|^2(\alpha+1)^2}\right)$};

\draw (0,0)  node[left] {$0$};

\draw (5:12)  node[below right] {$x_+^{(\alpha)}(\varepsilon,p)$};

\draw (8:10) node[above] {$x$};

\fill (0,0)  circle[radius=3pt];

\fill (5:12)  circle[radius=3pt];

\fill (8:10) circle[radius=3pt];

\draw (17,0) node[right] {\mbox{\LARGE $\Sigma_0^{(\alpha)}$}};

\draw[dashed] (0,0) -- (14:16);
\draw[dashed] (0,0) -- (-14:16);

\draw[very thick] (5:12) arc[start angle=5, end angle=14, radius=12] -- (14:16);

\draw[color=green!40!black!70, postaction=decorate] (0,0) -- (8:7.5);

\draw[color=green!40!black!70, postaction=decorate] (8:7.5) -- (8:10);

\draw (8:3.75) node[below] {$\gamma_1$};

\draw (8:8.75) node[below] {$\gamma_2$};

\end{tikzpicture}
\caption{\small The paths $\gamma_1$ and $\gamma_2$ (dark green line) defined in~\eqref{4-agosto-2024-1}. The solid black line represents the cut $\mathfrak{T}^{(\alpha)}(\varepsilon,p)$.}
\label{figure-4-agosto-2024-1}
\end{figure}
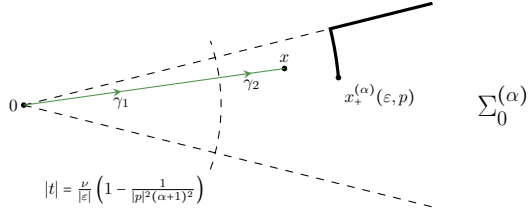
Along both $\gamma_1$ and $\gamma_2$ inequality~\eqref{5-agosto-2024-1} applies, thus the first integral in~\eqref{1-agosto-2024-4} is bounded from above as before, while to find an upper bound for the second integral in~\eqref{1-agosto-2024-4} we proceed as follows: we have the estimate
\[
\begin{aligned}
& \int_{\gamma_2} \left|\frac{\widehat{X}_{+}^{(\alpha)}(t;\varepsilon,p)}{\widehat{X}_+^{(\alpha)}(x;\varepsilon,p)} \widehat{K}^{(\alpha)}(x,t;\varepsilon,p) \widehat{F}^{(\alpha)}(t;\varepsilon,p)\right| |dt| \\ & \lesssim_{L,\Theta}\frac{1}{|p(\alpha+1)|^{\frac{2}{3}}}\int_{\frac{\nu}{|\varepsilon|}\left(1-\frac{1}{|p|^2(\alpha+1)^2}\right)}^{\frac{1}{|\varepsilon|}\left(1-\frac{1}{ |p|^2(\alpha+1)^2}\right)}\left|\widehat{X}_+^{(\alpha)}(t;\varepsilon,p)\widehat{X}_-^{(\alpha)}(t;\varepsilon,p)\widehat{F}^{(\alpha)}(t;\varepsilon,p)\right|d|t|,
\end{aligned}
\]
then we perform the change of integration variable $|t|\mapsto \frac{|t|}{|\varepsilon|}\left(1-\frac{1}{ |p|^2(\alpha+1)^2}\right)$ in the integral on the right hand side of the previous inequality and we apply the mean value theorem for definite integrals (in this case it is allowed) receiving a number $\nu<\nu^*<1$ (independent of all the parameters) such that
\[
\begin{aligned}
& \frac{1}{|p(\alpha+1)|^{\frac{2}{3}}}\int_{\frac{\nu}{|\varepsilon|}\left(1-\frac{1}{|p|^2(\alpha+1)^2}\right)}^{\frac{1}{|\varepsilon|}\left(1-\frac{1}{ |p|^2(\alpha+1)^2}\right)}\left|\widehat{X}_+^{(\alpha)}(t;\varepsilon,p)\widehat{X}_-^{(\alpha)}(t;\varepsilon,p)\widehat{F}^{(\alpha)}(t;\varepsilon,p)\right|d|t| \\
& \lesssim_{L,\Theta} \frac{1}{|\varepsilon||p(\alpha+1)|^{\frac{2}{3}}}\left|\widehat{X}_+^{(\alpha)}\left(w;\varepsilon,p\right)\widehat{X}_-^{(\alpha)}\left(w;\varepsilon,p\right)\widehat{F}^{(\alpha)}\left(w;\varepsilon,p\right) \right|_{w=\frac{\nu^*}{|\varepsilon|}\left(1-\frac{1}{ |p|^2(\alpha+1)^2}\right)} \\
& \lesssim_{L,\Theta,\nu^*} \frac{1}{|p|(\alpha+1)},
\end{aligned}
\]
where in the last line we have used again estimate~\eqref{5-agosto-2024-2} (with $\nu$ substituted with $\nu^*$) and representations~\eqref{29-gennaio-2024-1} of the Airy functions with bounds~\eqref{29-gennaio-2024-2}. Proceeding as before, we obtain again
\[
\left|\widehat{R}^{(\alpha)}(x;\varepsilon,p)\right|\lesssim_{\substack{L,\Theta \\ \nu,\nu^*}} \frac{1}{|p|(\alpha+1)}
\]
and
\[
\sup_{\alpha\ge\alpha_{L}(\nu)}\sup_{\substack{|\varepsilon|\ge 1,\,|\arg(\varepsilon)|\le\frac{\Theta}{\alpha+1} \\ |p|\ge L,\, |\arg(p)|\le\Theta}}\sup_{\substack{x\in\Sigma_0^{(\alpha)} \\ |x|\le\frac{1}{|\varepsilon|}\left(1-\frac{1}{|p|^2(\alpha+1)^2}\right)}} \left|\hat{z}^{(\alpha)}(x;\varepsilon,p)\right|\le 2,
\]
thus the conclusion follows from the integral relation~\eqref{19-marzo-2024-2}.

Finally, let us consider the points
\[
x\in\Sigma_0^{(\alpha)}\setminus\mathfrak{T}^{(\alpha)}(\varepsilon,p),\quad\mbox{with}\quad ||\varepsilon x|-1|\le\frac{1}{|p|^2(\alpha+1)^2}.
\]
We can choose the path $\gamma$ to be the product $\gamma:=\gamma_1*\gamma_2*\gamma_3$
 where 
\begin{equation}
\label{4-agosto-2024-2}
\gamma_{1}:=\left\{ r e^{i \arg\left(x_+^{(\alpha)}(\varepsilon,p)\right)},\, 0<r\le \frac{\nu}{|\varepsilon|}\left(1-\frac{1}{|p|^2(\alpha+1)^2}\right) \right\},
\end{equation}
\begin{equation}
\label{4-agosto-2024-2-bis}
\gamma_2:=\left\{r e^{i\arg\left(x_+^{(\alpha)}(\varepsilon,p)\right)},\, \frac{\nu}{|\varepsilon|}\left(1-\frac{1}{ |p|^2(\alpha+1)^2}\right)\le r \le \frac{1}{|\varepsilon|}\left(1-\frac{1}{ |p|^2(\alpha+1)^2}\right)\right\}
\end{equation}
and $\gamma_3:=\gamma_3^{(1)}*\gamma_3^{(2)}$ with
\begin{equation}
\label{4-agosto-2024-2-ter}
\begin{aligned}
& \gamma_{3}^{(1)}:=\left\{r e^{i\arg\left(x_+^{(\alpha)}(\varepsilon,p)\right)},\, \frac{1}{|\varepsilon|}\left(1-\frac{1}{ |p|^2 (\alpha+1)^2}\right)\le r \le |x|\right\}, \\[2ex]
& \gamma_{3}^{(2)}:=
\begin{cases}
\left\{|x|e^{i\phi},\,\arg\left(x_+^{(\alpha)}(\varepsilon,p)\right)\le \phi\le \arg(x)\right\}, & \mbox{if }\arg\left(x_+^{(\alpha)}(\varepsilon,p)\right)\le \arg(x), \\[2ex]
-\left\{|x| e^{i \phi},\,\arg(x)\le \phi \le \arg\left(x_+^{(\alpha)}(\varepsilon,p)\right)\right\}, & \mbox{if } \arg(x)\le \arg\left(x_+^{(\alpha)}(\varepsilon,p)\right),
\end{cases}
\end{aligned}
\end{equation}
see Figure \ref{figure-4-agosto-2024-2} below.

\begin{figure}[H]
\centering
\begin{tikzpicture}[decoration={markings, mark= at position 0.5 with {\arrow{stealth}}},scale=0.35, every node/.style={scale=0.65}] 

\draw[dashed] (-19:7.5) arc[start angle=-19, end angle=19, radius=7.5];

\draw[dashed] (-19:11) arc[start angle=-19, end angle=19, radius=9.75];

\draw (-19:7.5) node[below left] {$|t|=\frac{\nu}{|\varepsilon|}\left(1-\frac{1}{ |p|^2(\alpha+1)^2}\right)$};

\draw (-19:11) node[below left] {$|t|=\frac{1}{|\varepsilon|}\left(1-\frac{1}{ |p|^2(\alpha+1)^2}\right)$};

\draw (0,0)  node[left] {$0$};

\draw (5:11)  node[below right] {$x_+^{(\alpha)}(\varepsilon,p)$};

\draw (-2:15) node[right] {$x$};

\fill (0,0)  circle[radius=3pt];

\fill (5:12)  circle[radius=3pt];

\fill (-2:15) circle[radius=3pt];

\draw (17,0) node[right] {\mbox{\LARGE $\Sigma_0^{(\alpha)}$}};

\draw[dashed] (0,0) -- (14:16);
\draw[dashed] (0,0) -- (-14:16);

\draw[very thick] (5:12) arc[start angle=5, end angle=14, radius=12] -- (14:16);

\draw[color=green!40!black!70, postaction=decorate] (0,0) -- (5:7.5);

\draw[color=green!40!black!70, postaction=decorate] (5:7.5) -- (5:11);

\draw[color=green!40!black!70, postaction=decorate] (5:11) -- (5:15);

\draw[color=green!40!black!70, postaction=decorate] (5:15) arc[start angle=5, end angle=-2, radius=15];

\draw (5:3.75) node[below] {$\gamma_1$};

\draw (5:9.25) node[below] {$\gamma_2$};

\draw (5:13) node[above] {$\gamma_3^{(1)}$};

\draw (1.5:15) node[right] {$\gamma_3^{(2)}$};

\end{tikzpicture}
\caption{\small The paths $\gamma_1$, $\gamma_2$ and $\gamma_3$ (dark green line) defined in~\eqref{4-agosto-2024-2},~\eqref{4-agosto-2024-2-bis} and~\eqref{4-agosto-2024-2-ter} . The solid black line represents the cut $\mathfrak{T}^{(\alpha)}(\varepsilon,p)$.}
\label{figure-4-agosto-2024-2}
\end{figure}
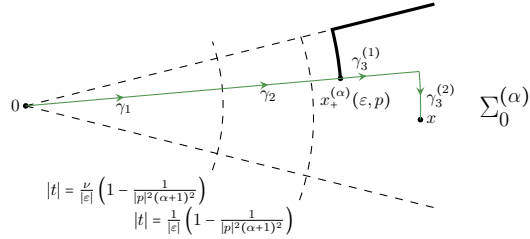
At this point we proceed as before: we split the integral in~\eqref{1-agosto-2024-4} into three integrals along $\gamma_1$, $\gamma_2$ and $\gamma_3$; the integrals along $\gamma_1$ and $\gamma_2$ are studied as before, while to obtain an estimate of the integral along $\gamma_3$ we make use of the following inequalities:
\begin{equation}
\label{6-agosto-2024-1}
\begin{aligned}
&\left|\frac{\widehat{X}_+^{(\alpha)}(t;\varepsilon,p)}{\widehat{X}_+^{(\alpha)}(x;\varepsilon,p)}\widehat{K}^{(\alpha)}(x,t;\varepsilon,p)\right|\lesssim_W \frac{1}{|p|(\alpha+1)},\\
& \left|\widehat{F}^{(\alpha)}(t;\varepsilon,p)\right|\lesssim_W 1 
\end{aligned}
\end{equation}
for all $||\varepsilon x|-1|,||\varepsilon t|-1|\le\frac{1}{|p|^2(\alpha+1)^2}$. Putting together these results, using the previous ones and again integral equation~\eqref{19-marzo-2024-2} we conclude the proof.
\end{proof}

We are going now to show that, when both $|\varepsilon|$ and $|p|$ are bounded and $|\arg(p)|=\mathcal{O}\left(\frac{1}{\alpha+1}\right)$, the distinguished subdominant Frobenius solution 

\noindent
$\chi_+^{(\alpha)}\left(x;4p^2(\alpha+1)^2\varepsilon^2,2p(\alpha+1)-\frac{1}{2}\right)$ can be approximated with the function $\widehat{X}_+^{(\alpha)}(x;\varepsilon,p)$ defined in~\eqref{30-gennaio-2024-4} till an open sector containing the point $x=1$. To this aim, we will need the following
\begin{proposition}
\label{proposition-24-marzo-2024-1}
Let us fix constants $L>0$ and $0<\Theta<\frac{\pi}{2}$. There exists a constant $\alpha_{L}>0$ depending only on $L$ (and $\Theta$) such that

\begin{equation}
\label{6-febbraio-2024-3-ter}
\left| \frac{2}{3}\left[S^{(\alpha)}(x;\varepsilon,p)\right]^{\frac{3}{2}}\mp \left[\sqrt{(\varepsilon x)^2-1}-\operatorname{arctan}\sqrt{(\varepsilon x)^2 -1}\right] \right|\lesssim_{L,\Theta} \frac{1}{|p|(\alpha+1)^2}
\end{equation}
holds for all $x\in\Sigma_0^{(\alpha)}\setminus\mathfrak{T}^{(\alpha)}(\varepsilon,p)$ with $\frac{1}{|\varepsilon|}\left(1+\frac{1}{|p|^2(\alpha+1)^2}\right)\le|x|\le 1+\frac{1}{\alpha+1}$, all $\alpha\ge\alpha_L$, $|\varepsilon|\ge 1$, $|\arg(\varepsilon)|\le\frac{\Theta}{\alpha+1}$ and all $|p|\ge L$, $|\arg(p)|\le \Theta$. Here \textquotedblleft$\mp$\textquotedblright\,correspond to the choices of the positive and negative branches of the square root, respectively.
\end{proposition}

\begin{proof}
The proof is postponed to Appendix \ref{appendix-technical-proofs}.
\end{proof}

An immediate consequence of Proposition \ref{proposition-24-marzo-2024-1} is the following

\begin{corollary}
\label{corollary-8-agosto-2024-1}
With the same notations and assumptions of Proposition \ref{proposition-24-marzo-2024-1} the following holds for all $x,x_1,x_2\in\Sigma_0^{(\alpha)}\setminus\mathfrak{T}^{(\alpha)}(\varepsilon,p)$ with $\frac{1}{|\varepsilon|}\left(1+\frac{1}{|p|^2(\alpha+1)^2}\right)\le|x|,|x_1|,|x_2|\le 1+\frac{1}{\alpha+1}$, all $\alpha\ge\alpha_L$, $|\varepsilon|\ge 1$, $|\arg(\varepsilon)|\le\frac{\Theta}{\alpha+1}$ and all $|p|\ge L$, $|\arg(p)|\le \Theta$:
\begin{itemize}
\item[i)] For the choice of the positive branch of the square root
\begin{equation}
\label{8-agosto-2024-1}
\operatorname{Re}\frac{2}{3}\left[S^{(\alpha)}(x_2;\varepsilon,p)\right]^{\frac{3}{2}}-\operatorname{Re}\frac{2}{3}\left[S^{(\alpha)}(x_1;\varepsilon,p)\right]^{\frac{3}{2}}\gtrsim_{L,\Theta} \frac{1}{|p|(\alpha+1)}
\end{equation}
for $|x_1|\le|x_2|$ with $\arg(x_1)=\arg(x_2)$;
\item[ii)] For both choices of the branch of the square root
\begin{equation}
\label{8-agosto-2024-2}
\left|\operatorname{Im} \frac{2}{3}\left[S^{(\alpha)}(x;\varepsilon,p)\right]^{\frac{3}{2}}\right|\lesssim_{L,\Theta}\frac{1}{|p|(\alpha+1)}.
\end{equation} 
\end{itemize}
\end{corollary}

\begin{proof}
It follows from inequality~\eqref{6-febbraio-2024-3-ter} of Proposition \ref{proposition-24-marzo-2024-1} and the analogous properties of the function $\sqrt{(\varepsilon x)^2-1}-\operatorname{arctan}\sqrt{(\varepsilon x)^2-1}$ which can be verified by elementary methods.
\end{proof}

\begin{lemma}
\label{lemma-8-agosto-2024-1}
Let $I\subset \mathbb{R}_{\ge 1}\times \mathbb{R}_{>0}$ be a compact subset and let us fix a constant $0<\Theta<\frac{\pi}{2}$. For the choice of the positive branch of the square root in the definition of the action $S^{(\alpha)}(x;\varepsilon,p)$ (see~\eqref{19-gennaio-2024-2}) there exists a constant $\alpha_{I}>0$ depending only on $I$ (and $\Theta$) such that
\begin{equation}
\label{8-agosto-2024-3}
\begin{aligned}
&\left|\operatorname{Ai}\left[-(2p(\alpha+1))^{\frac{2}{3}}S^{(\alpha)}(x;\varepsilon,p)\right]\right|,\left|\operatorname{Bi}\left[-(2p(\alpha+1))^{\frac{2}{3}}S^{(\alpha)}(x;\varepsilon,p)\right]\right| \\ &\lesssim_{I,\Theta} \frac{1}{2 |p|^{\frac{1}{6}}(\alpha+1)^{\frac{1}{6}}\left|S^{(\alpha)}(x;\varepsilon,p)\right|^{\frac{1}{4}}}
\end{aligned}
\end{equation}
and
\begin{equation}
\label{8-agosto-2024-4}
\begin{aligned}
&\left|\operatorname{Ai}'\left[-(2p(\alpha+1))^{\frac{2}{3}}S^{(\alpha)}(x;\varepsilon,p)\right]\right|,\left|\operatorname{Bi}'\left[-(2p(\alpha+1))^{\frac{2}{3}}S^{(\alpha)}(x;\varepsilon,p)\right]\right| \\ &\lesssim_{I,\Theta} 2 |p|^{\frac{1}{6}}(\alpha+1)^{\frac{1}{6}}\left|S^{(\alpha)}(x;\varepsilon,p)\right|^{\frac{1}{4}}
\end{aligned}
\end{equation}
hold for all $x\in\Sigma_0^{(\alpha)}\setminus\mathfrak{T}^{(\alpha)}(\varepsilon,p)$ with $\frac{1}{|\varepsilon|}\left(1+\frac{1}{|p|^2(\alpha+1)^2}\right)\le |x|\le 1+\frac{1}{\alpha+1}$, all $\alpha\ge\alpha_I$ and all $(|\varepsilon|,|p|)\in I$, $|\arg(\varepsilon)|,|\arg(p)|\le\frac{\Theta}{\alpha+1}$.
\end{lemma}

\begin{proof}
It follows from inequality~\eqref{6-febbraio-2024-3-ter} of Proposition \ref{proposition-24-marzo-2024-1}, inequalities~\eqref{8-agosto-2024-1} and~\eqref{8-agosto-2024-2} of Corollary \ref{corollary-8-agosto-2024-1}, the series representations~\eqref{taylor-ai},~\eqref{taylor-bi} and representations~\eqref{29-gennaio-2024-1-bis} with bounds~\eqref{29-gennaio-2024-1-ter} for the remainder terms given in Appendix \ref{appendix-29-gennaio-2024}.
\end{proof}

We have now all the technical tools to prove the last Proposition of this section.

\begin{proposition}
\label{proposition-8-agosto-2024-1}
Let $I\subset \mathbb{R}_{\ge 1}\times \mathbb{R}_{>0}$ be a compact subset and let us fix a constant $0<\Theta<\frac{\pi}{2}$. For the choice of the positive branch of the square root in the definition of the action $S^{(\alpha)}(x;\varepsilon,p)$ (see~\eqref{19-gennaio-2024-2}) there exists a constant $\alpha_{I}>0$ depending only on $I$ (and $\Theta$) such that
\begin{equation}
\label{8-agosto-2024-5}
\begin{aligned}
& \left|\chi_+^{(\alpha)}\left(x;4p^2(\alpha+1)^2\varepsilon^2,2p(\alpha+1)-\frac{1}{2}\right)-\widehat{X}_+^{(\alpha)}(x;\varepsilon,p)\right| \\ 
& \lesssim_{I,\Theta}\frac{1}{\alpha+1} \left|\frac{2 \sqrt{\pi} \left[2p(\alpha+1)\right]^{\frac{1}{6}}}{\Gamma\left(1+2p\right)}\varepsilon^{-2 p (\alpha+1)}\right|
\end{aligned}
\end{equation}
and
\begin{equation}
\label{8-agosto-2024-6}
\begin{aligned}
& \left|\frac{d}{dx} \chi_+^{(\alpha)}\left(x;4p^2(\alpha+1)^2 \varepsilon^2,2p(\alpha+1)-\frac{1}{2}\right)-\frac{d}{dx} \widehat{X}_+^{(\alpha)}(x;\varepsilon,p)\right|\\ 
& \lesssim_{I,\Theta} \left|\frac{2 \sqrt{\pi} \left[2p(\alpha+1)\right]^{\frac{1}{6}}}{\Gamma\left(1+2p\right)}\varepsilon^{-2 p (\alpha+1)} \right|
\end{aligned}
\end{equation}
hold for all $x\in \Sigma_0^{(\alpha)}\setminus\mathfrak{T}^{(\alpha)}(\varepsilon,p)$ with $\frac{1}{|\varepsilon|}\left(1+\frac{1}{|p|^2(\alpha+1)^2}\right)\le |x|\le 1+\frac{1}{\alpha+1}$, all $\alpha\ge\alpha_I$ and all $(|\varepsilon|,|p|)\in I$, $|\arg(\varepsilon)|,|\arg(p)|\le\frac{\Theta}{\alpha+1}$.
\end{proposition}

\begin{proof}
The proof is similar to the proof given for Proposition~\eqref{proposition-30-gennaio-2024-1}. We consider the Volterra integral equation
\begin{equation}
\label{8-agosto-2024-7}
\begin{aligned}
& \chi_+^{(\alpha)}\left(x;4p^2(\alpha+1)^2\varepsilon^2,2p(\alpha+1)-\frac{1}{2}\right)=\widehat{X}_+^{(\alpha)}(x;\varepsilon,p)+\\
& +\int_{\gamma,0}^x \widehat{K}^{(\alpha)}(x,t;\varepsilon,p)\widehat{F}^{(\alpha)}(t;\varepsilon,p)\chi_+^{(\alpha)}\left(t;4p^2(\alpha+1)^2\varepsilon^2, 2p(\alpha+1)-\frac{1}{2}\right) dt
\end{aligned}
\end{equation}
satisfied by $\chi_+^{(\alpha)}\left(x;4p^2(\alpha+1)^2\varepsilon^2,2p(\alpha+1)-\frac{1}{2}\right)$, with kernel $\widehat{K}^{(\alpha)}(x,t;\varepsilon,p)$ as in~\eqref{31-gennaio-2024-2}, forcing term $\widehat{F}^{(\alpha)}(t;\varepsilon,p)$ as in~\eqref{31-gennaio-2024-3} and $\gamma$ an oriented path of finite length with support contained in 
\[
\Sigma_0^{(\alpha)}\setminus\mathfrak{T}^{(\alpha)}(\varepsilon,p)\cap\left\{|t|\le 1+\frac{1}{\alpha+1}\right\}.
\]
For $x\in\Sigma_0^{(\alpha)}\setminus\mathfrak{T}^{(\alpha)}(\varepsilon,p)$ with $\frac{1}{|\varepsilon|}\left(1+\frac{1}{|p|^2(\alpha+1)^2}\right)\le |x|\le 1+\frac{1}{\alpha+1}$, we can take $\gamma:=\mathring{\gamma}_1*\mathring{\gamma}_2*\mathring{\gamma}_3$, with
\begin{equation}
\label{8-agosto-2024-8}
\begin{aligned}
& \mathring{\gamma}_1:=\left\{r e^{i \arg\left(x_+^{(\alpha)}(\varepsilon,p)\right)},\,0<r\le\frac{1}{|\varepsilon|}\left(1+\frac{1}{|p|^2(\alpha+1)^2}\right)\right\}, \\[2ex]
& \mathring{\gamma}_2:=
\begin{cases}
\left\{\frac{1}{|\varepsilon|}\left(1+\frac{1}{|p|^2(\alpha+1)^2}\right) e^{i\phi},\,\arg\left(x_+^{(\alpha)}(\varepsilon,p)\right)\le \phi\le\arg(x)\right\}, & \mbox{if }\arg\left(x_+^{(\alpha)}(\varepsilon,p)\right)\le\arg(x),\\[2ex]
-\left\{\frac{1}{|\varepsilon|}\left(1+\frac{1}{|p|^2(\alpha+1)^2}\right) e^{i\phi},\, \arg(x)\le\phi\le\arg\left(x_+^{(\alpha)}(\varepsilon,p)\right)\right\}, & \mbox{if }\arg(x)\le \arg\left(x_+^{(\alpha)}(\varepsilon,p)\right),
\end{cases}
\\[2ex]
& \mathring{\gamma}_3:=\left\{ r e ^{i\arg(x)},\, \frac{1}{|\varepsilon|}\left(1+\frac{1}{|p|^2(\alpha+1)^2}\right)\le r\le|x|  \right\},
\end{aligned}
\end{equation}
see Figure \ref{figure-8-agosto-2024-1} below.
\begin{figure}[H]
\centering
\begin{tikzpicture}[decoration={markings, mark= at position 0.5 with {\arrow{stealth}}},scale=0.35, every node/.style={scale=0.65}] 

\draw[dashed] (-17:10) arc[start angle=-17, end angle=17, radius=10];

\draw[dashed] (-17:13.6) arc[start angle=-17, end angle=17, radius=13.6];

\draw (-17:10) node[below left] {$|t|=\frac{1}{|\varepsilon|}\left(1+\frac{1}{|p|^2(\alpha+1)^2}\right)$};

\draw (-17:13.6) node[below left] {$|t|=1+\frac{1}{\alpha+1}$};

\draw (0,0)  node[left] {$0$};

\draw (5:8)  node[below] {$x_+^{(\alpha)}(\varepsilon,p)$};

\draw (-2:14) node[below right] {$x_0^{(\alpha)}(\varepsilon,p)$};

\fill (0,0)  circle[radius=3pt];

\fill (5:8)  circle[radius=3pt];

\fill (-2:14)  circle[radius=3pt];

\draw (19,0) node[right] {\mbox{\LARGE $\Sigma_0^{(\alpha)}$}};

\draw[color=green!40!black!70, postaction=decorate] (0,0) -- (5:10);

\draw[color=green!40!black!70, postaction=decorate] (5:10) arc[start angle=5,end angle=0, radius=10];

\draw[color=green!40!black!70, postaction=decorate] (10,0) -- (12.7,0);

\draw (5:5) node[above] {$\mathring{\gamma}_1$};

\draw (2.5:10) node[right] {$\mathring{\gamma}_2$};

\draw (0:11.85) node[below] {$\mathring{\gamma}_3$};

\draw[dashed] (0,0) -- (14:19);
\draw[dashed] (0,0) -- (-14:19);

\draw[very thick] (5:8) arc[start angle=5, end angle=14, radius=8] -- (14:14) -- (14:14) arc[start angle=14, end angle=-2, radius=14];

\end{tikzpicture}
\caption{\small The paths $\mathring{\gamma}_1$, $\mathring{\gamma}_2$ and $\mathring{\gamma}_3$ (dark green lines) defined in~\eqref{8-agosto-2024-8}.}
\label{figure-8-agosto-2024-1}
\end{figure}
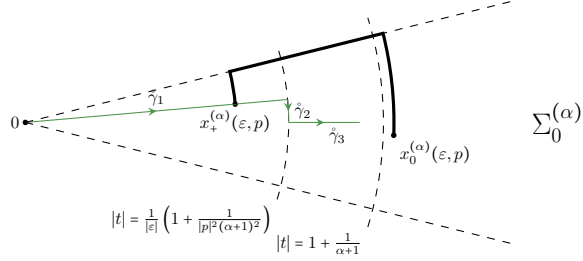
To obtain estimate~\eqref{8-agosto-2024-5} we have to study (uniform) upper bounds for the functions
\[
\int_{\mathring{\gamma}_1*\mathring{\gamma}_2} \left|\widehat{K}^{(\alpha)}(x,t;\varepsilon,p)\widehat{F}^{(\alpha)}(t;\varepsilon,p)\chi_+^{(\alpha)}\left(t;4p^2(\alpha+1)^2\varepsilon^2, 2p(\alpha+1)-\frac{1}{2}\right)\right| |dt|,
\]
and
\[
\int_{\mathring{\gamma}_3} \left|\widehat{K}^{(\alpha)}(x,t;\varepsilon,p)\widehat{F}^{(\alpha)}(t;\varepsilon,p)\right| |dt|.
\]
For the integral along $\mathring{\gamma}_1*\mathring{\gamma}_2$ the results of Proposition \ref{proposition-30-gennaio-2024-1} apply, while for the integral along $\mathring{\gamma}_3$ we make use of Lemma \ref{lemma-8-agosto-2024-1}, inequality~\eqref{8-agosto-2024-3} , obtaining
\[
\int_{\mathring{\gamma}_3} \left|\widehat{K}^{(\alpha)}(x,t;\varepsilon,p)\widehat{F}^{(\alpha)}(t;\varepsilon,p)\right| |dt|\lesssim_{I,\Theta} \frac{1}{\alpha+1},
\]
for $\alpha\ge\alpha_{I}$, being $\alpha_{I}>0$ sufficiently big (depending on $I$ and $\Theta$). Putting together these results and using the integral relation~\eqref{8-agosto-2024-7} we find
\[
\left|\widehat{\chi}_+^{(\alpha)}(x;\varepsilon,p)\right|\lesssim_{I,\Theta} \left|\frac{2 \sqrt{\pi} \left[2p(\alpha+1)\right]^{\frac{1}{6}}}{\Gamma\left(1+2p\right)}\varepsilon^{-2 p (\alpha+1)}\right|
\]
for all  $x\in\Sigma_0^{(\alpha)}\setminus\mathfrak{T}^{(\alpha)}(\varepsilon,p)$ with $\frac{1}{|\varepsilon|}\left(1+\frac{1}{|p|^2(\alpha+1)^2}\right)\le |x|\le 1+\frac{1}{\alpha+1}$, all $\alpha\ge\alpha_{I}$ and all $(|\varepsilon|,|p||)\in I$, $|\arg(\varepsilon)|,|\arg(p)|\le\frac{\Theta}{\alpha+1}$, thus, using again the integral equation, inequality~\eqref{8-agosto-2024-5} follows.

To prove inequality~\eqref{8-agosto-2024-6}, we take the $x$-derivative of equation~\eqref{8-agosto-2024-7} and we proceed in a similar way (using also inequality~\eqref{8-agosto-2024-4} of Lemma \ref{lemma-8-agosto-2024-1}).
\end{proof}

We have now all the ingredients to study the spectral determinant 
\[
\mathcal{Q}_+^{(\alpha)}\left(4p^2(\alpha+1)^2\varepsilon^2;2p(\alpha+1)-\frac{1}{2}\right)
\]
and its zeros. It turns out that there are two different interesting regimes according to the values of the rescaled energy parameter $\varepsilon$: the one with $|\varepsilon|$ bounded from $1$, which we call the \textit{high rescaled energy} case and the one with $|\varepsilon|$ close to $1$, which we call the \textit{low rescaled energy} case. We will treat them separately in the next two sections, giving the precise restrictions on $\varepsilon$.

\subsubsection{Spectral determinant - The high rescaled energy case and proof of Theorem \ref{theorem-13-agosto-2024-1}}
\label{high-energy-case}

In this section, we will restrict the values of $|\varepsilon|$ to be strictly greater than $1$, namely it will take values on (compact subsets of ) $\mathbb{R}_{>1}$. Furthermore, the choice of the positive branch of the square root of the rescaled potential $\widehat{V}^{(\alpha)}(x;\varepsilon,p)$ will be always understood. 

\begin{lemma}
\label{lemma-9-agosto-2024-1}
Let $I\subset\mathbb{R}_{>1}\times \mathbb{R}_{>0}$ be a compact subset and let us fix a number $0<\Theta<\frac{\pi}{2}$. There exists a constant $\alpha_{I}>0$ depending only on $I$ (and $\Theta$) such that
\begin{equation}
\label{9-agosto-2024-1}
\begin{aligned}
& \left|\chi_+^{(\alpha)}\left(1;4p^2(\alpha+1)^2\varepsilon^2,2p(\alpha+1)-\frac{1}{2}\right) \right. \\
& \left.-\frac{2 \varepsilon^{-2p(\alpha+1)}}{\Gamma(1+2p)(\varepsilon^2-1)^{\frac{1}{4}}}\cos\left(2p(\alpha+1)\left[\sqrt{\varepsilon^2-1}-\operatorname{arctan}\sqrt{\varepsilon^2 -1}\right] -\frac{\pi}{4}\right)\right|\\
&\lesssim_{I,\Theta} \frac{1}{\alpha+1} \left| \frac{2 \sqrt{\pi} \left[2p(\alpha+1)\right]^{\frac{1}{6}}}{\Gamma\left(1+2p\right)}\varepsilon^{-2 p (\alpha+1)} \right|
\end{aligned}
\end{equation}
and
\begin{equation}
\label{9-agosto-2024-2}
\begin{aligned}
& \left|\frac{d}{dx}\chi_+^{(\alpha)}\left(x;4p^2(\alpha+1)^2\varepsilon^2,2p(\alpha+1)-\frac{1}{2}\right)\right|_{x=1} \\
& \lesssim_{I,\Theta} (\alpha+1)^{\frac{5}{6}}\left|\frac{2 \sqrt{\pi} \left[2p(\alpha+1)\right]^{\frac{1}{6}}}{\Gamma\left(1+2p\right)}\varepsilon^{-2 p (\alpha+1)}\right|
\end{aligned}
\end{equation} 
hold for all $\alpha\ge\alpha_{I}$ and all $(|\varepsilon|,|p|)\in I$, $|\arg(\varepsilon)|,|\arg(p)|\le\frac{\Theta}{\alpha+1}$.
\end{lemma} 

\begin{proof}
It follows from definition~\eqref{30-gennaio-2024-4} of the function $\widehat{X}_+^{(\alpha)}(x;\varepsilon,p)$, Proposition \ref{proposition-8-agosto-2024-1} and representations~\eqref{29-gennaio-2024-1-bis},~\eqref{derivative-ai-bi-bis} for the Airy function $\operatorname{Ai}$ and its derivative.
\end{proof}

We can now prove Theorem \ref{theorem-13-agosto-2024-1}:

\begin{proof}[Proof of Theorem \ref{theorem-13-agosto-2024-1}]
Point i) is a direct consequence of Corollary \ref{corollary-13-agosto-2024-1-bis} and Lemma \ref{lemma-9-agosto-2024-1}.

Let us prove point ii). Let us denote by $W^{(\alpha)}$ the sector in the $\varepsilon$-plane defined by
\[
W^{(\alpha)}:=\left\{ |\varepsilon|\in I_1,\,|\arg(\varepsilon)|\le\frac{\Theta}{\alpha+1} \right\}
\]
($\Theta<\frac{\pi}{2}$ is a fixed constant as in the statement of point i)). From inequality~\eqref{13-agosto-2024-2} it follows that
\[
\begin{aligned}
&\left|\frac{2 \varepsilon^{-2p(\alpha+1)}}{\Gamma(1+2p)(\varepsilon^2-1)^{\frac{1}{4}}}\cos\left(2p(\alpha+1)\left[\sqrt{\varepsilon^2-1}-\operatorname{arctan}\sqrt{\varepsilon^2 -1}\right] -\frac{\pi}{4}\right)\right|_{\varepsilon\in\partial W^{(\alpha)}}> \\
& \left|\mathcal{Q}_+^{(\alpha)}\left(4p^2(\alpha+1)^2\varepsilon^2,2p(\alpha+1)-\frac{1}{2}\right)\right. \\
&  \left. -\frac{2 \varepsilon^{-2p(\alpha+1)}}{\Gamma(1+2p)(\varepsilon^2-1)^{\frac{1}{4}}}\cos\left(2p(\alpha+1)\left[\sqrt{\varepsilon^2-1}-\operatorname{arctan}\sqrt{\varepsilon^2 -1}\right] -\frac{\pi}{4}\right)\right|_{\varepsilon\in\partial W^{(\alpha)}},
\end{aligned}
\]  
hence, by Rouché theorem, it follows that the number of zeros (counted with multiplicities) of $\mathcal{Q}^{(\alpha)}(4p^2(\alpha+1)^2\varepsilon^2;2p(\alpha+1)-\frac{1}{2})$ in $W^{(\alpha)}$ is equal to the number of zeros of $\cos\left(2p(\alpha+1)\left[\sqrt{\varepsilon^2-1}-\operatorname{arctan}\sqrt{\varepsilon^2 -1}\right] -\frac{\pi}{4}\right)$ in $W^{(\alpha)}$. For real $\varepsilon\in I_1$ and $p\in I_2$, due to monotonicity of the function $\sqrt{\varepsilon^2-1}-\operatorname{arctan}\sqrt{\varepsilon^2-1}$ and relation
\[
\frac{d}{d\varepsilon} \left[\sqrt{\varepsilon^2-1}-\operatorname{arctan}\sqrt{\varepsilon^2-1}\right]=\frac{\sqrt{\varepsilon^2-1}}{\varepsilon},
\]
we find
\[
\left\lfloor \frac{2p(\alpha+1)}{\pi}\int_{I_1}\frac{\sqrt{\varepsilon^2-1}}{\varepsilon} d\varepsilon \right\rfloor\le n_{I_1}^{(\alpha)}(p)\le \left\lceil \frac{2p(\alpha+1)}{\pi}\int_{I_1}\frac{\sqrt{\varepsilon^2-1}}{\varepsilon} d\varepsilon \right\rceil, 
\]
where $\lfloor -\rfloor$ and $\lceil - \rceil$ denote the floor and ceiling functions, respectively\footnote{Recall that for a real number $r$ the number $\lfloor r \rfloor$ is the integer satisfying $r-1\le \lfloor r \rfloor\le r$, while $\lceil r\rceil$ is the integer satisfying $r\le\lceil r\rceil\le r+1 $.}, and $n_{I_1}^{(\alpha)}(p)$ is the number of zeros in $I_1$. As a consequence, we have
\[
\left|\frac{n_{I_1}^{(\alpha)}(p)}{\alpha+1}-\frac{2p}{\pi}\int_{I_1}\frac{\sqrt{\varepsilon^2-1}}{\varepsilon}d\varepsilon\right|\le \frac{1}{\alpha+1}
\]
and the result follows.
\end{proof}

\subsubsection{Spectral determinant - The low rescaled energy case and proof of Theorem \ref{theorem-airy-det}}
\label{low-energy-case}

In this section we reparameterize the rescaled energy parameter $\varepsilon$ as
\begin{equation}
\label{13-agosto-2024-5}
\varepsilon=1+\frac{\eta}{(\alpha+1)^{\frac{2}{3}}},
\end{equation}
where $\eta\ne 0$ is a complex parameter with $|\eta|$ bounded and $|\arg(\eta)|=\mathcal{O}\left(\frac{1}{(\alpha+1)^{\frac{1}{3}}}\right)$.

\begin{lemma}
\label{lemma-13-agosto-2024-1}
Let $J_1\subset\mathbb{R}_{>0}$ be a compact subset and let us fix numbers $L>0$ and $0<\Theta<\frac{\pi}{2}$. There exists a constant $\alpha_{J_1,L}>0$ depending only on $J_1$ and $L$ (and $\Theta$) such that inequality
\begin{equation}
\label{13-agosto-2024-4}
\left|\frac{2}{3}\left[S^{(\alpha)}\left(1;1+\frac{\eta}{(\alpha+1)^{\frac{2}{3}}},p\right)\right]^{\frac{3}{2}}\mp \frac{2^{\frac{3}{2}}}{3}\frac{\eta^{\frac{3}{2}}}{\alpha+1}\right|\lesssim_{J_1,L,\Theta}\frac{1}{(\alpha+1)^{\frac{5}{3}}}
\end{equation}
holds for all $\alpha\ge\alpha_{J_1,L}$ and all $|\eta|\in J_1$, $|\arg(\eta)|\le\frac{\Theta}{(\alpha+1)^{\frac{1}{3}}}$, $|p|\ge L$, $|\arg(p)|\le \Theta$. Here \textquotedblleft$\mp$\textquotedblright\, refer to the choice of the positive or negative branch of the square root of the rescaled potential, respectively, in the definition of the action $S^{(\alpha)}(x;\varepsilon,p)$ (see~\eqref{19-gennaio-2024-2}).
\end{lemma}
\begin{proof}
It is an immediate consequence of Proposition \ref{proposition-24-marzo-2024-1} with the new parameterization of $\varepsilon$ given in~\eqref{13-agosto-2024-5}.  
\end{proof}

From here on, the choice of the positive branch of the square root of the rescaled potential will be always assumed.

\begin{lemma}
\label{lemma-13-agosto-2024-2}
Let $J\subset\mathbb{R}_{>0}\times \mathbb{R}_{>0}$ be a compact subset and let us fix a number $0<\Theta<\frac{\pi}{2}$. There exists a constant $\alpha_{J}>0$ depending only on $J$ (and $\Theta$) such that
\begin{equation}
\label{13-agosto-2024-6}
\begin{aligned}
& \left|\chi_+^{(\alpha)}\left(1;4 p^2(\alpha+1)^2\left(1+\frac{\eta}{(\alpha+1)^{\frac{2}{3}}}\right)^2,2p(\alpha+1)-\frac{1}{2}\right)\right. \\
 & \left.- \frac{2 \sqrt{\pi} \left[p(\alpha+1)\right]^{\frac{1}{6}}}{\Gamma\left(1+2p\right)}\left(1+\frac{\eta}{(\alpha+1)^{\frac{2}{3}}}\right)^{-2 p (\alpha+1)} \operatorname{Ai}\left(-2 p^{\frac{2}{3}}\eta\right)\right| 
 \\ 
 & \lesssim_{J,\Theta}\frac{1}{(\alpha+1)^{\frac{2}{3}}} \frac{2 \sqrt{\pi} \left[2p(\alpha+1)\right]^{\frac{1}{6}}}{\Gamma\left(1+2p\right)}\left(1+\frac{\eta}{(\alpha+1)^{\frac{2}{3}}}\right)^{-2 p (\alpha+1)}
\end{aligned}
\end{equation}
and
\begin{equation}
\label{13-agosto-2024-7}
\begin{aligned}
& \left|\frac{d}{dx}\chi_+^{(\alpha)}\left(x;4 p^2(\alpha+1)^2\left(1+\frac{\eta}{(\alpha+1)^{\frac{2}{3}}}\right)^2,2p(\alpha+1)-\frac{1}{2}\right)\right|_{x=1} \\
& \lesssim_{J,\Theta}  (\alpha+1)^{\frac{2}{3}} \frac{2 \sqrt{\pi} \left[2p(\alpha+1)\right]^{\frac{1}{6}}}{\Gamma\left(1+2p\right)}\left(1+\frac{\eta}{(\alpha+1)^{\frac{2}{3}}}\right)^{-2 p (\alpha+1)}
\end{aligned}
\end{equation}
hold for all $\alpha\ge\alpha_{J}$ and all $(|\eta|,|p|)\in J$, $|\arg(\eta)|\le \frac{\Theta}{(\alpha+1)^{\frac{1}{3}}}$, $|\arg(p)|\le\frac{\Theta}{\alpha+1}$.
\end{lemma}

\begin{proof}
It follows from Proposition \ref{proposition-8-agosto-2024-1} and Lemma \ref{lemma-13-agosto-2024-1}.
\end{proof}

We can now prove Theorem \ref{theorem-airy-det}:

\begin{proof}[Proof of Theorem \ref{theorem-airy-det}]
Point i), inequality~\eqref{23-marzo-2024-1}, is an immediate consequence of Corollary \ref{corollary-13-agosto-2024-1-bis} and Lemma \ref{lemma-13-agosto-2024-1}. 
Under the assumptions of point ii), the statement is readily proved from the previous point: we have
\[
\begin{aligned}
& \left| \mathcal{Q}_+^{(\alpha)}\left(4p^2(\alpha+1)^2\left(1+\frac{\eta}{(\alpha+1)^{\frac{2}{3}}}\right)^2;2p(\alpha+1)-\frac{1}{2}\right) \right|\gtrsim_{J_1,J_2,\Theta} \\
&\left|\frac{2 \sqrt{\pi} \left[2p(\alpha+1)\right]^{\frac{1}{6}}}{\Gamma\left(1+2p\right)}\left(1+\frac{\eta}{(\alpha+1)^{\frac{2}{3}}}\right)^{-2 p (\alpha+1)}\right|\left(\left|\operatorname{Ai}\left(-2p^{\frac{2}{3}}\eta\right)\right|-\frac{\log(\alpha+1)}{(\alpha+1)^{\frac{1}{3}}}\right)>0,
\end{aligned}
\]
where the last inequality (which holds for all sufficiently big values of $\alpha$) follows from compactness of $J_1,J_2$ and the assumption that $\frac{|a_s|}{2|p|^{\frac{2}{3}}}\notin J_1$ for all $s\in\mathbb{Z}_{\ge 0}$ and all $|p|\in J_2$, $|\arg(p)|\le\frac{\Theta}{\alpha+1}$.

Let us now prove point iii). Let $\mathring{M}>0$ and let us consider the disc
\[
\mathring{\mathcal{D}}:=\left\{\left|\frac{2p^{\frac{2}{3}}\eta}{|a_k|}-1\right|\le \mathring{M} \frac{\log(\alpha+1)}{(\alpha+1)^{\frac{1}{3}}}\right\},
\]
which is a neighborhood of $\frac{|a_k|}{2p^{\frac{2}{3}}}$ contained in a subset of the $\eta$-plane with $|\eta|\in J_1$, $|\arg(\eta)|\le \frac{\Theta}{(\alpha+1)^{\frac{1}{3}}}$, for some compact $J_1\subset\mathbb{R}_{>0}$ and some $0<\Theta<\frac{\pi}{2}$ independent of $\mathring{M}$. On the boundary of $\mathring{\mathcal{D}}$ we have
\[
\begin{aligned}
\left|\operatorname{Ai}\left(-2 p^{\frac{2}{3}}\eta\right)\right|_{\eta\in\partial\mathring{\mathcal{D}}} \ge &  |a_k| \left|\operatorname{Ai}^{(1)}(-a_k)\right| \mathring{M}\frac{\log(\alpha+1)}{(\alpha+1)^{\frac{1}{3}}} \times \\ & \times \left(1-\mathring{M}\frac{|a_k|}{2}\frac{\left|\operatorname{Ai}^{(2)}(-a_k)\right|}{\operatorname{Ai}^{(1)}\left(-a_k\right)}\frac{\log(\alpha+1)}{(\alpha+1)^{\frac{1}{3}}}\widehat{\mathcal{L}}_{J_2}(k)\right),
\end{aligned}
\]
where
\[
\widehat{\mathcal{L}}_{J_2}(k):=\sup_{|p|\in J_2}\frac{1}{|p|^{\frac{2}{3}}} \frac{|a_k|}{\left|\operatorname{Ai}^{(2)}(-a_k)\right|}\sum_{s\ge 3} \frac{1}{s!}\left|\operatorname{Ai}^{(s)}\left(-a_k\right)\right|<\infty
\]
(here $\operatorname{Ai}^{(k)}(z)=\frac{d^k}{dz^k}\operatorname{Ai}(z)$), thus, letting 
\[
\mathring{M}=\mathring{M}_k(J_2) > \frac{2^{\frac{7}{6}}}{|a_k|\left|\operatorname{Ai}^{(1)}(-a_k)\right|}\widehat{\mathcal{C}}_{J_1,J_2},
\]
where $\widehat{\mathcal{C}}_{J,I_2}>0$ is the implicit constant in inequality~\eqref{23-marzo-2024-1} of point \textit{i)} (which is independent of $\mathring{M}$), and taking sufficiently big values of $\alpha$ (depending on $J_2$ and $k$) so that
\[
\mathring{M}_k(J_2)\frac{|a_k|}{2}\frac{\left|\operatorname{Ai}^{(2)}(-a_k)\right|}{\operatorname{Ai}^{(1)}\left(-a_k\right)}\frac{\log(\alpha+1)}{(\alpha+1)^{\frac{1}{3}}}\widehat{\mathcal{L}}_{J_2}(k)\le\frac{1}{2},
\] 
we find
\[
\begin{aligned}
& \left| \mathcal{Q}_+^{(\alpha)}\left(4p^2(\alpha+1)^2\left(1+\frac{\eta}{(\alpha+1)^{\frac{2}{3}}}\right)^2;2p(\alpha+1)-\frac{1}{2}\right) \right. \\
& \left.- \frac{2 \sqrt{\pi} \left[p(\alpha+1)\right]^{\frac{1}{6}}}{\Gamma\left(1+2p\right)}\left(1+\frac{\eta}{(\alpha+1)^{\frac{2}{3}}}\right)^{-2 p (\alpha+1)} \operatorname{Ai}\left(-2 p^{\frac{2}{3}}\eta\right)\right|_{\eta\in\partial\mathring{\mathcal{D}}}< \\
& <\left|\frac{2 \sqrt{\pi} \left[p(\alpha+1)\right]^{\frac{1}{6}}}{\Gamma\left(1+2p\right)}\left(1+\frac{\eta}{(\alpha+1)^{\frac{2}{3}}}\right)^{-2 p (\alpha+1)} \operatorname{Ai}\left(-2 p^{\frac{2}{3}}\eta\right)\right|_{\eta\in\partial\mathring{\mathcal{D}}}.
\end{aligned}
\]
The conclusion follows from Rouché theorem and the construction of the neighborhood $\mathring{\mathcal{D}}$.
\end{proof}
\begin{remark}
The same observations of Remark \ref{remark-20-agosto-2024-1} apply also for point iii) of Theorem \ref{theorem-airy-det} with the obvious changes.
\end{remark}

\begin{ackn}
I am indebted to my advisor Prof. Davide Masoero, who introduced me to the problems addressed in this paper and constantly gave me advice and suggestions. I also thank Giulio Ruzza for many stimulating discussions.
\\

This work is supported by the FCT Ph.D. scholarship UI/BD/152215/2021 and the FCT Projects UIDB/00208/2020, DOI: \url{https://doi.org/10.54499/UIDB/00208/2020}, UIDP/00208/2020, DOI: \url{https://doi.org/10.54499/UIDP/00208/2020}, 2022.03702.PTDC (GENIDE), DOI: \url{https://doi.org/10.54499/2022.03702.PTDC}. The author is member of the COST Action CA21109 CaLISTA.
\end{ackn}

\appendix

\section{Asymptotic computations}
\label{appendix-technical-proofs}

In this Appendix we collect the proofs that were not given in the main part of the paper.

\begin{proof}[Proof of Proposition \ref{proposition-22-maggio-2024-1}]
Let us start considering the disc
\begin{equation}
\label{21-agosto-2024-1}
\mathbb{D}_+:=\left\{\frac{1}{\varepsilon}+\frac{re^{i \theta}}{|\varepsilon||p|^2(\alpha+1)^2},\,0\le r\le R,\,0\le\theta<2\pi\right\}\subset\Sigma_0^{(\alpha)},
\end{equation}
for a constant $R$ to be chosen. We have
\[
\left|1-(\varepsilon x)^2\right|_{\partial \mathbb{D}_+}\ge \frac{2 R}{|p|^2(\alpha+1)^2}\left(1-\frac{R}{2 |p|^2(\alpha+1)^2}\right)
\]
and
\[
\left|\frac{x^{2 \alpha+2}}{4 p^2(\alpha+1)^2}\right|_{\partial \mathbb{D}_+}\le\frac{e^{\frac{2 R}{|p|^2(\alpha+1)}}}{4|p|^2(\alpha+1)^2};
\]
choosing $R>\frac{1}{8}$, for all sufficiently big values of $\alpha$ we have
\[
\left|1-(\varepsilon x)^2\right|_{\partial \mathbb{D}_+}>\left|\frac{x^{2 \alpha+2}}{4 p^2(\alpha+1)^2}\right|_{\partial \mathbb{D}_+},
\]
thus, by Rouché theorem, point i) of Proposition \ref{proposition-22-maggio-2024-1} follows. To prove point ii) we follow the same reasoning considering the disc
\[
\mathbb{D}_0:=\left\{\left(2p(\alpha+1)\varepsilon\right)^{\frac{1}{\alpha}}+\frac{r e^{i \theta}}{\alpha} \left(2|p|(\alpha+1)|\varepsilon|\right)^{\frac{1}{\alpha}},\,0\le r\le \widetilde{R},\,0\le\theta<2\pi \right\},
\]
with $\widetilde{R}>\frac{1}{2}$. Point iii) follows by repeated application of Rouché theorem in subsets of $\Sigma_0^{(\alpha)}\setminus \left(\mathbb{D}_+\cup\mathbb{D}_0\right)$.
\end{proof}

\begin{proof}[Proof of Proposition \ref{proposition-20-luglio-2024-1}]
For all $|\varepsilon|\ge 1$, $|\arg(\varepsilon)|\le\frac{\Theta}{\alpha+1}$ and all $|p|\ge L$, $|\arg(p)|\le \Theta$, let us start considering points $x$ in the region
\begin{equation}
\label{2-agosto-2024-1}
\left(\Sigma_0^{(\alpha)}\setminus\mathfrak{T}^{(\alpha)}(\varepsilon,p)\right)\cap\left\{ ||\varepsilon x|-1|\le\frac{1}{ |p|^2(\alpha+1)^2}\right\},
\end{equation}
where $\mathfrak{T}^{(\alpha)}(\varepsilon,p)$ is the cut~\eqref{27-luglio-2024-1} (see Figure \ref{figure-2-agosto-2024-1} below).

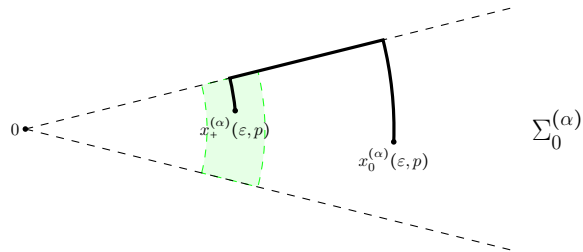
\begin{figure}[H]
\centering
\begin{tikzpicture}[decoration={markings, mark= at position 0.5 with {\arrow{stealth}}},scale=0.35, every node/.style={scale=0.65}] 

\filldraw[color=green!90, fill=green!10, dashed] (-14:6.9) arc[start angle=-14, end angle=14, radius=6.9] -- (14:9.1) -- (14:9.1) arc[start angle=14, end angle=-14, radius=9.1] -- (-14:6.9);

\draw (0,0)  node[left] {$0$};

\draw (5:8)  node[below] {$x_+^{(\alpha)}(\varepsilon,p)$};

\draw (-2:14) node[below] {$x_0^{(\alpha)}(\varepsilon,p)$};

\fill (0,0)  circle[radius=3pt];

\fill (5:8)  circle[radius=3pt];

\fill (-2:14)  circle[radius=3pt];

\draw (19,0) node[right] {\mbox{\LARGE $\Sigma_0^{(\alpha)}$}};

\draw[dashed] (0,0) -- (14:19);
\draw[dashed] (0,0) -- (-14:19);

\draw[very thick] (5:8) arc[start angle=5, end angle=14, radius=8] -- (14:14) -- (14:14) arc[start angle=14, end angle=-2, radius=14];

\end{tikzpicture}
\caption{\small The region (green area) defined in~\eqref{2-agosto-2024-1}.}
\label{figure-2-agosto-2024-1}
\end{figure}

For the integral~\eqref{19-gennaio-2024-2} defining the action $S^{(\alpha)}(x;\varepsilon,p)$ we can take the oriented path $\gamma$ starting at the turning point $x_+^{(\alpha)}(\varepsilon,p)$ and ending at $x$ as follows: $\gamma:=\gamma_1*\gamma_2$, where
\begin{equation}
\label{2-agosto-2024-2}
\begin{aligned}
& \gamma_1:=
\begin{cases}
\left\{ r e^{i \arg\left(x_+^{(\alpha)}(\varepsilon,p)\right)},\,\left|x_+^{(\alpha)}(\varepsilon,p)\right|\le r \le |x| \right\}, & \mbox{if } \left|x_+^{(\alpha)}(\varepsilon,p)\right|\le |x|, \\[2ex]
-\left\{ r e^{i \arg\left(x_+^{(\alpha)}(\varepsilon,p)\right)},\,|x| \le r \le \left|x_+^{(\alpha)}(\varepsilon,p)\right|  \right\}, & \mbox{if } |x| \le  \left|x_+^{(\alpha)}(\varepsilon,p)\right|,
\end{cases}
\\[2ex]
& \gamma_2:=
\begin{cases}
\left\{ |x| e^{i \phi},\,\arg\left(x_+^{(\alpha)}(\varepsilon,p)\right)\le\phi\le \arg(x) \right\}, & \mbox{if }\arg\left(x_+^{(\alpha)}(\varepsilon,p)\right)\le \arg(x), \\[2ex]
-\left\{ |x| e^{i \phi},\, \arg(x)\le\phi\le\arg\left(x_+^{(\alpha)}(\varepsilon,p)\right)    \right\}, & \mbox{if }\arg(x)\le\arg\left(x_+^{(\alpha)}(\varepsilon,p)\right),
\end{cases}
\end{aligned}
\end{equation} 
see Figure \ref{figure-2-agosto-2024-2} below. Notice that the support of $\gamma$ is fully contained in region~\eqref{2-agosto-2024-1}.

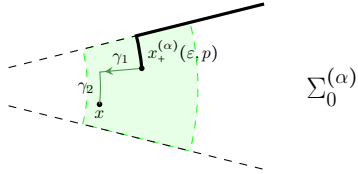
\begin{figure}[H]
\centering
\begin{tikzpicture}[decoration={markings, mark= at position 0.5 with {\arrow{stealth}}},scale=0.35, every node/.style={scale=0.65}] 

\filldraw[color=green!90, fill=green!10, dashed] (-14:5.9) arc[start angle=-14, end angle=14, radius=5.9] -- (14:10.1) -- (14:10.1) arc[start angle=14, end angle=-14, radius=10.1] -- (-14:5.9);

\draw (5:8)  node[above right] {$x_+^{(\alpha)}(\varepsilon,p)$};

\draw (-6:6.4) node[below] {$x$};

\draw (5:7.2) node[above] {$\gamma_1$};

\draw (-0.5:6.4) node[left] {$\gamma_2$};

\fill (-6:6.4)  circle[radius=3pt];

\draw[color=green!40!black!70, postaction=decorate] (5:8) -- (5:6.4) -- (5:6.4) arc[start angle=5, end angle=-6, radius=6.4];

\fill (5:8)  circle[radius=3pt];

\draw (14,0) node[right] {\mbox{\LARGE $\Sigma_0^{(\alpha)}$}};

\draw[dashed] (14:3) -- (14:13);
\draw[dashed] (-14:3) -- (-14:13);

\draw[very thick] (5:8) arc[start angle=5, end angle=14, radius=8] -- (14:13);

\end{tikzpicture}
\caption{\small The path $\gamma$ (dark green line) defined in~\eqref{2-agosto-2024-2} with support in the region (green area) defined in~\eqref{2-agosto-2024-1}.}
\label{figure-2-agosto-2024-2}
\end{figure}

We easily see that for all $t$ in region~\eqref{2-agosto-2024-1} the following inequalities hold:
\[
\left|\frac{1-(\varepsilon t)^2}{t^2}\right|\lesssim_{W} \frac{|\varepsilon|^2}{|p|^2(\alpha+1)^2},\quad\mbox{and}\quad \left|\frac{t^{2 \alpha}}{4 p^2 (\alpha+1)^2}\right|\lesssim_{W}\frac{1}{|p|^2(\alpha+1)^2}.
\]
Furthermore, the lengths of the paths $\gamma_1$ and $\gamma_2$ defined in~\eqref{2-agosto-2024-2} have the following bounds:
\[
\int_{\gamma_1}|dt|\lesssim_W \frac{1}{|\varepsilon|(\alpha+1)^2},\quad\mbox{and}\quad\int_{\gamma_2}|dt|\lesssim_W \frac{1}{|\varepsilon|(\alpha+1)}. 
\]
Putting these last results together, we obtain the bound
\[
\left|\frac{2}{3}\left[-S^{(\alpha)}(x;\varepsilon,p)\right]^{\frac{3}{2}}\right|\lesssim_W \frac{1}{|p|(\alpha+1)^2}
\]
for all $x$ belonging to region~\eqref{2-agosto-2024-1}, all sufficiently big values of $\alpha$ (depending on $L$) and all $|\varepsilon|\ge 1$, $|\arg(\varepsilon)|\le\frac{\Theta}{\alpha+1}$, $|p|\ge L$, $|\arg(p)|\le \Theta$.

We consider now points $x$ in the region
\begin{equation}
\label{2-agosto-2024-3}
\left(\Sigma_0^{(\alpha)}\setminus\mathfrak{T}^{(\alpha)}(\varepsilon,p) \right)\cap\left\{ |x|\le\frac{1}{|\varepsilon|}\left(1-\frac{1}{|p|^2(\alpha+1)^2}\right) \right\},
\end{equation}
see Figure \ref{figure-2-agosto-2024-3} below.

\begin{figure}[H]
\centering
\begin{tikzpicture}[decoration={markings, mark= at position 0.5 with {\arrow{stealth}}},scale=0.35, every node/.style={scale=0.65}] 

\filldraw[color=green!90, fill=green!10, dashed] (-14:6.9) arc[start angle=-14, end angle=14, radius=6.9] -- (0,0) -- (-14:6.9);

\draw (0,0)  node[left] {$0$};

\draw (5:8)  node[below] {$x_+^{(\alpha)}(\varepsilon,p)$};

\draw (-2:14) node[below] {$x_0^{(\alpha)}(\varepsilon,p)$};

\fill (0,0)  circle[radius=3pt];

\fill (5:8)  circle[radius=3pt];

\fill (-2:14)  circle[radius=3pt];

\draw (19,0) node[right] {\mbox{\LARGE $\Sigma_0^{(\alpha)}$}};

\draw[dashed] (0,0) -- (14:19);
\draw[dashed] (0,0) -- (-14:19);

\draw[very thick] (5:8) arc[start angle=5, end angle=14, radius=8] -- (14:14) -- (14:14) arc[start angle=14, end angle=-2, radius=14];

\end{tikzpicture}
\caption{\small The region (green area) defined in~\eqref{2-agosto-2024-3}.}
\label{figure-2-agosto-2024-3}
\end{figure}
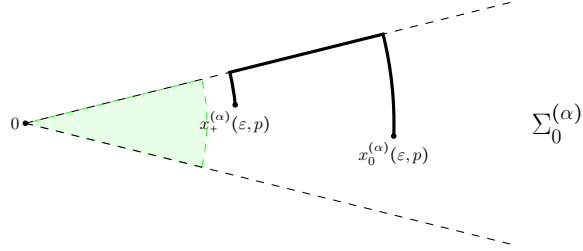

For the integral~\eqref{19-gennaio-2024-2} defining the action $S^{(\alpha)}(x;\varepsilon,p)$ we can take the oriented path $\gamma$ starting at $x_+^{(\alpha)}(\varepsilon,p)$ and ending at $x$ as follows: $\gamma:=\tilde{\gamma}_1*\tilde{\gamma}_2*\tilde{\gamma}_3$, where
\begin{equation}
\label{2-agosto-2024-4}
\begin{aligned}
& \tilde{\gamma}_1:=-\left\{r e^{i\arg\left(x_+^{(\alpha)}(\varepsilon,p)\right)},\,\left|z^{(\alpha)}(\varepsilon,p)\right|\le r\le \left|x_+^{(\alpha)}(\varepsilon,p)\right|  \right\}, \\[2ex]
& \tilde{\gamma}_2:= \begin{cases}
\left\{ \left|z^{(\alpha)}(\varepsilon,p)\right|e^{i\phi},\, \arg\left(x_+^{(\alpha)}(\varepsilon,p)\right) \le\phi \le \arg(x) \right\}, & \mbox{if } \arg\left(x_+^{(\alpha)}(\varepsilon,p)\right)\le\arg(x), \\[2ex]
 -\left\{ \left|z^{(\alpha)}(\varepsilon,p)\right| e^{i\phi},\,\arg(x)  \le\phi \le  \arg\left(x_+^{(\alpha)}(\varepsilon,p)\right)  \right\}, & \mbox{if }\arg(x)\le  \arg\left(x_+^{(\alpha)}(\varepsilon,p)\right),
\end{cases}
\\[2ex]
& \tilde{\gamma}_3:=
-\left\{r e^{i\arg\left(x\right)},\,|x|\le r\le \left|z^{(\alpha)}(\varepsilon,p)\right|  \right\},
\end{aligned}
\end{equation}
where $z^{(\alpha)}(\varepsilon,p):=\frac{1}{|\varepsilon|}\left(1-\frac{1}{|p|^2(\alpha+1)^2}\right) e^{i\arg\left(x_+^{(\alpha)}(\varepsilon,p)\right)}$; see Figure \ref{figure-2-agosto-2024-4} below.

\begin{figure}[H]
\centering
\begin{tikzpicture}[decoration={markings, mark= at position 0.5 with {\arrow{stealth}}},scale=0.35, every node/.style={scale=0.65}] 

\filldraw[color=green!90, fill=green!10, dashed] (-14:7.5) arc[start angle=-14, end angle=14, radius=7.5] -- (0,0) -- (-14:7.5);

\draw (0,0)  node[left] {$0$};

\draw (-8:4) node[above left] {$x$};

\draw (5:12)  node[above right] {$x_+^{(\alpha)}(\varepsilon,p)$};

\fill (0,0)  circle[radius=3pt];

\fill (5:12)  circle[radius=3pt];

\fill (-8:4) circle[radius=3pt];

\draw (17,0) node[right] {\mbox{\LARGE $\Sigma_0^{(\alpha)}$}};

\draw[dashed] (0,0) -- (14:16);
\draw[dashed] (0,0) -- (-14:16);

\draw[very thick] (5:12) arc[start angle=5, end angle=14, radius=12] -- (14:16);

\draw[color=green!40!black!70, postaction=decorate] (5:12) -- (5:7.5);

\draw[color=green!40!black!70, postaction=decorate] (5:7.5) arc[start angle=5, end angle=-8, radius=7.5];

\draw[color=green!40!black!70, postaction=decorate] (-8:7.5) -- (-8:4);

 -- (5:4) arc[start angle=5, end angle=-8, radius=4];

\draw (5:9.75) node[above] {$\tilde{\gamma}_1$};

\draw (-1.5:7.5) node[right] {$\tilde{\gamma}_2$};

\draw (-8:5.75) node[above] {$\tilde{\gamma}_3$};

\end{tikzpicture}
\caption{\small The paths $\tilde{\gamma}_1$, $\tilde{\gamma_2}$, $\tilde{\gamma}_3$ (dark green line) defined in~\eqref{2-agosto-2024-4}. The green area represents the region defined in~\eqref{2-agosto-2024-3}.}
\label{figure-2-agosto-2024-4}
\end{figure}
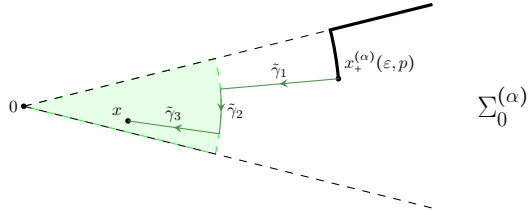

For the integration in~\eqref{19-gennaio-2024-2} along $\tilde{\gamma}_1$ and $\tilde{\gamma}_2$ we can apply the previous result. For the integration in-\eqref{19-gennaio-2024-2} along $\tilde{\gamma}_3$, first of all we notice that for all $t$ in region~\eqref{2-agosto-2024-3} we have
\[
\left|\frac{t^{2 \alpha+2}}{4 p^2 (\alpha+1)^2}\right|\le \frac{1}{2},
\]
thus we are allowed to (absolutely and uniformly) Taylor expand the square root of the rescaled potential $\widehat{V}^{(\alpha)}(x;\varepsilon,p)$ obtaining
\[
\begin{aligned}
&\pm \int_{\tilde{\gamma}_3,\left|z^{(\alpha)}(\varepsilon,p)\right| e^{i \arg(x)}}^x \sqrt{\widehat{V}^{(\alpha)}(t;\varepsilon,p)} dt= \\ & = \mp \left[\sqrt{1-(\varepsilon t)^2}-\operatorname{arctanh}\sqrt{1-(\varepsilon t)^2}\right]_{t=\left|z^{(\alpha)}(\varepsilon,p)\right| e^{i \arg(x)}}^{t=x}   \\ & \pm\sum_{m\ge 1} \binom{1/2}{m}\left(\frac{1}{4p^2(\alpha+1)^2}\right)^m \int_{\tilde{\gamma}_3,\left|z^{(\alpha)}(\varepsilon,p)\right| e^{i \arg(x)}}^x\frac{t^{2 m(\alpha+1)-1}}{\left(1-(\varepsilon t)^2\right)^{m-\frac{1}{2}}} dt,
\end{aligned}
\]
where the upper/lower signs correspond to the choices of the positive/negative branch of the square root. Since
\[
\left|\left[\sqrt{1-(\varepsilon t)^2}-\operatorname{arctanh}\sqrt{1-(\varepsilon t)^2}\right]_{t=z^{(\alpha)}(\varepsilon,p)}\right|\lesssim_{W}\frac{1}{|p|(\alpha+1)^2}
\]
and
\[
\left|\left(\frac{1}{4 p^2 (\alpha+1)^2}\right)^m \int_{\tilde{\gamma}_3,\left|z^{(\alpha)}(\varepsilon,p)\right| e^{i \arg(x)}}^x \frac{t^{2m(\alpha+1)-1}}{\left(1-(\varepsilon t)^2\right)^{m-\frac{1}{2}}} dt \right|\lesssim_{W}\frac{1}{|p|(\alpha+1)^2}\left(\frac{1}{2}\right)^{m},
\]
the sought inequality follows.

\end{proof}

\begin{proof}[Proof of Proposition \ref{proposition-24-marzo-2024-1}]
The proof is very similar to the one  of Proposition \ref{proposition-20-luglio-2024-1}. For all $|\varepsilon|\ge 1$, $|\arg(\varepsilon)|\le \frac{\Theta}{\alpha+1}$ and all $|p|\ge L$, $|\arg(p)|\le \Theta$ we consider points $x$ in the region
\begin{equation}
\label{6-agosto-2024-10}
\left(\Sigma_0^{(\alpha)}\setminus\mathfrak{T}^{(\alpha)}(\varepsilon,p)\right)\cap\left\{\frac{1}{|\varepsilon|}\left(1+\frac{1}{|p|^2(\alpha+1)^2}\right)\le |x|\le 1+\frac{1}{\alpha+1}\right\},
\end{equation}
see Figure \ref{figure-6-agosto-2024-1} below.

\begin{figure}[H]
\centering
\begin{tikzpicture}[decoration={markings, mark= at position 0.5 with {\arrow{stealth}}},scale=0.35, every node/.style={scale=0.65}] 

\draw[dashed] (16:9.1) arc[start angle=16, end angle=-16, radius=9.1];

\draw[dashed] (16:13) arc[start angle=16, end angle=-16, radius=13];

\draw (-16:9.1) node[left] {$|x|=\frac{1}{|\varepsilon|}\left(1+\frac{1}{|p|^2(\alpha+1)^2}\right)$};

\draw (-16:13) node[left] {$|x|=1+\frac{1}{\alpha+1}$};

\filldraw[color=green!90, fill=green!10, dashed] (14:9.1) arc[start angle=14, end angle=-14, radius=9.1] -- (-14:13) -- (-14:13) arc[start angle=-14, end angle=14, radius=13] -- (14:9.1);

\draw (0,0)  node[left] {$0$};

\draw (5:8)  node[below] {$x_+^{(\alpha)}(\varepsilon,p)$};

\draw (-2:14) node[below] {$x_0^{(\alpha)}(\varepsilon,p)$};

\fill (0,0)  circle[radius=3pt];

\fill (5:8)  circle[radius=3pt];

\fill (-2:14)  circle[radius=3pt];

\draw (19,0) node[right] {\mbox{\LARGE $\Sigma_0^{(\alpha)}$}};

\draw[dashed] (0,0) -- (14:19);
\draw[dashed] (0,0) -- (-14:19);

\draw[very thick] (5:8) arc[start angle=5, end angle=14, radius=8] -- (14:14) -- (14:14) arc[start angle=14, end angle=-2, radius=14];

\end{tikzpicture}
\caption{\small The region (green area) defined in~\eqref{6-agosto-2024-10}.}
\label{figure-6-agosto-2024-1}
\end{figure}
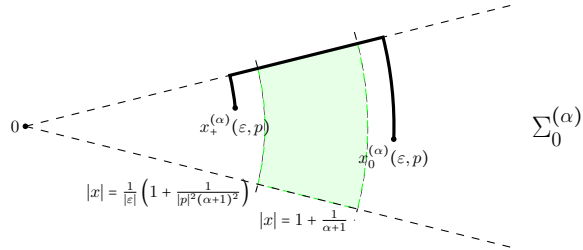
For the integral~\eqref{19-gennaio-2024-2} defining the action $S^{(\alpha)}(x;\varepsilon,p)$ we can take the oriented path $\gamma$ joining $x_+^{(\alpha)}(\varepsilon,p)$ to $x$ to be $\gamma:=\hat{\gamma}_1*\hat{\gamma}_2*\hat{\gamma}_3$, where
\begin{equation}
\label{6-agosto-2024-11}
\begin{aligned}
&\hat{\gamma}_1:=\left\{r e^{i\arg\left(x_+^{(\alpha)}(\varepsilon,p)\right)},\, \left|x_+^{(\alpha)}(\varepsilon,p)\right|\le r \le\frac{1}{|\varepsilon|}\left(1+\frac{1}{|p|^2(\alpha+1)^2}\right)\right\}, \\[2ex]
&\hat{\gamma}_2:=
\begin{cases}
\left\{\frac{1}{|\varepsilon|}\left(1+\frac{1}{|p|^2(\alpha+1)^2}\right) e^{i\phi},\,\arg\left(x_+^{(\alpha)}(\varepsilon,p)\right)\le\phi \le \arg(x)\right\}, &\mbox{if }\arg\left(x_+^{(\alpha)}(\varepsilon,p)\right)\le \arg(x), \\[2ex]
-\left\{\frac{1}{|\varepsilon|}\left(1+\frac{1}{|p|^2(\alpha+1)^2}\right) e^{i\phi},\,\arg(x)\le\phi \le \arg\left(x_+^{(\alpha)}(\varepsilon,p)\right)\right\}, &\mbox{if }\arg(x)\le \arg\left(x_+^{(\alpha)}(\varepsilon,p)\right),
\end{cases}
\\[2ex]
&\hat{\gamma}_3:=\left\{r e^{i\arg(x)},\,\frac{1}{|\varepsilon|}\left(1+\frac{1}{|p|^2(\alpha+1)^2}\right)\le r\le |x|\right\},
\end{aligned}
\end{equation}
see Figure \ref{figure-6-agosto-2024-2} below.

\begin{figure}[H]
\centering
\begin{tikzpicture}[decoration={markings, mark= at position 0.5 with {\arrow{stealth}}},scale=0.35, every node/.style={scale=0.65}] 

\draw[dashed] (16:9.1) arc[start angle=16, end angle=-16, radius=9.1];

\draw[dashed] (16:13) arc[start angle=16, end angle=-16, radius=13];

\filldraw[color=green!90, fill=green!10, dashed] (14:9.1) arc[start angle=14, end angle=-14, radius=9.1] -- (-14:13) -- (-14:13) arc[start angle=-14, end angle=14, radius=13] -- (14:9.1);

\draw (0,0)  node[left] {$0$};

\draw (5:8)  node[below left] {$x_+^{(\alpha)}(\varepsilon,p)$};

\draw (-2:14) node[below right] {$x_0^{(\alpha)}(\varepsilon,p)$};

\fill (0,0)  circle[radius=3pt];

\fill (5:8)  circle[radius=3pt];

\fill (-2:14)  circle[radius=3pt];

\draw (19,0) node[right] {\mbox{\LARGE $\Sigma_0^{(\alpha)}$}};

\draw[dashed] (0,0) -- (14:19);
\draw[dashed] (0,0) -- (-14:19);

\draw[very thick] (5:8) arc[start angle=5, end angle=14, radius=8] -- (14:14) -- (14:14) arc[start angle=14, end angle=-2, radius=14];

\draw[color=green!40!black!70, postaction=decorate] (5:8) -- (5:9.1);

\draw[color=green!40!black!70, postaction=decorate] (5:9.1) arc[start angle=5, end angle=-7, radius=9.1];

\draw[color=green!40!black!70, postaction=decorate] (-7:9.1) -- (-7:12);

\draw (5:8.55) node[below] {$\hat{\gamma}_1$};

\draw (-1:9.1) node[right] {$\hat{\gamma}_2$};

\draw (-7:10.55) node[below] {$\hat{\gamma}_3$};

\fill (-7:12) circle[radius=3pt];

\draw (-7:12) node[right] {$x$};

\end{tikzpicture}
\caption{\small The paths $\hat{\gamma}_1$, $\hat{\gamma}_2$ and $\hat{\gamma}_3$ (dark green lines) defined in~\eqref{6-agosto-2024-11}.}
\label{figure-6-agosto-2024-2}
\end{figure}
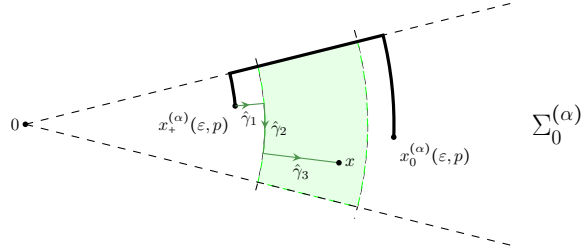

Along $\hat{\gamma}_1$ and $\hat{\gamma}_2$ inequality~\eqref{6-febbraio-2024-1-bis} of Proposition \ref{proposition-20-luglio-2024-1} applies, while along $\hat{\gamma}_3$ we proceed as in the proof of Proposition \ref{proposition-20-luglio-2024-1} by Taylor expanding the square root of the rescaled potential $\widehat{V}^{(\alpha)}(x;\varepsilon,p)$, thus we omit the details. 

\end{proof}

\section{Special functions}
\label{bessel-appendix}
We summarize here some properties of the special functions that are used in this paper. The results that are collected here are mainly taken from \cite{erdelyi-special-functions} and \cite{NIST:DLMF}.

\subsection{Bessel functions}
\label{subappendix-bessel}
Bessel functions are solutions to the differential equation
\begin{equation}
\label{general-bessel-equation}
z^2\frac{d^2 w}{dz^2}+z\frac{dw}{dz}+(z^2-\nu^2)w=0,\quad \nu\in\mathbb{C}.
\end{equation}
An independent pair of standard solutions is given by $\left\{J_\nu(z),\,Y_{\nu}(z)\right\}$, where
\begin{equation}
\label{bessel-first-kind}
J_{\nu}(z):=\left(\frac{z}{2}\right)^{\nu}\sum_{k=0}^\infty \frac{(-1)^k}{k!\Gamma\left(\nu+1+k\right)}\left(\frac{z}{2}\right)^{2k}
\end{equation}
is known as the Bessel function of the first kind, and
\begin{equation}
\label{bessel-second-kind-generic}
Y_\nu(z):=\cot(\nu\pi)J_\nu(z)-\csc(\nu\pi) J_{-\nu}(z),\quad\nu\notin\mathbb{Z}
\end{equation}
or
\begin{equation}
\label{bessel-second-kind-nongeneric}
\begin{aligned}
Y_n(z)= & -\frac{1}{\pi}\left(\frac{z}{2}\right)^{-n}\sum_{k=0}^{n-1}\frac{(n-k-1)!}{k!}\left(\frac{z}{2}\right)^{2k}+\frac{2}{\pi}J_n(z)\log\left(\frac{z}{2}\right) \\ & -\frac{1}{\pi}\left(\frac{z}{2}\right)^n\sum_{k=0}^\infty \frac{\psi(k+1)+\psi(n+k+1)}{k!(n+k)!}\left(-\frac{z}{2}\right)^{2k},\quad \nu=n\in\mathbb{Z}
\end{aligned}
\end{equation}
(here $\psi$ denotes the digamma function) is known as the Bessel function of the second kind. The principal branch of $J_{\nu}(z)$ is the principal branch of $z^\nu$, and defines an analytic function on $\mathbb{C}\setminus(-\infty,0]$, while for $\nu=n\in\mathbb{Z}$ the principal branch of $Y_n(z)$ is the principal branch of $\log(z)$ and defines an analytic function on $\mathbb{C}\setminus(-\infty,0]$. For each $z\ne 0$, $J_\nu(z)$ is an entire function of $\nu$.

The analytic continuation formulas for the Bessel functions of the first and second kind are the following: for $m\in\mathbb{Z}$,
\begin{equation}
\label{17-maggio-2024-1}
J_\nu\left( z e^{i m \pi} \right)=e^{i m \nu \pi} J_\nu(z),
\end{equation}
for $\nu\notin\mathbb{Z}$
\begin{equation}
\label{17-maggio-2024-2}
Y_\nu\left(z e^{i m \pi}\right)=e^{- i m \nu \pi} Y_\nu(z)+2 i \sin(m \nu \pi) \cot(\nu \pi) J_\nu(z),
\end{equation}
while for $\nu=n\in\mathbb{Z}$
\begin{equation}
\label{17-maggio-2024-3}
Y_n\left( z e^{i m \pi} \right)=(-1)^{mn}\left( Y_n(z)+ 2 i m J_n(z) \right).
\end{equation}

Another standard pair of independent solutions is $\left\{H_\nu^{(1)}(z),H_\nu^{(2)}(z)\right\}$, where $H_\nu^{(1)}(z)$ and $H_\nu^{(2)}(z)$ are called Hankel functions (or Bessel functions of the third kind) and are defined by the normalizations
\[
\begin{aligned}
& \lim_{\substack{ z\to\infty \\ -\pi<\arg(z)<2\pi }} \sqrt{\frac{\pi z}{2}} e^{-i\left(z-\frac{\nu\pi}{2}-\frac{\pi}{4}\right)} H_\nu^{(1)}(z)=1, \\ & \lim_{\substack{ z\to \infty \\ -2\pi<\arg(z)<\pi}} \sqrt{\frac{\pi z}{2}} e^{i\left(z-\frac{\pi\nu}{2}-\frac{\pi}{4}\right)} H_\nu^{(2)}(z)=1.
\end{aligned}
\]
Their principal branch corresponds to the principle branch of the square root appearing in the previous normalization conditions and defines analytic functions on $\mathbb{C}\setminus\left(-\infty,0\right]$. For each $z\ne 0$, $H_\nu^{(1)}(z)$ and $H_\nu^{(2)}(z)$ are entire functions of $\nu$.

Denoting by $\mathcal{C}_\nu(z)$ either the Bessel function of the first and second kind or the Hankel functions, the following formulas hold:
\begin{equation}
\label{bessel-function-derivative}
\frac{1}{z}\frac{d}{dz}\left(z^\nu\mathcal{C}_\nu(z)\right)=z^{\nu-1}\mathcal{C}_{\nu-1}(z)
\end{equation}
and
\begin{equation}
\label{multiplication-theorem}
\mathcal{C}_\nu(\lambda z)=\lambda^{\pm \nu}\left[\mathcal{C}_\nu(z)+\sum_{m\ge 1}\frac{(\mp1)^{m} (\lambda^2-1)^m }{m!}\left(\frac{z}{2}\right)^m \mathcal{C}_{\nu\pm m}(z) \right],\quad |\lambda^2-1|<1
\end{equation}
(the second identity is known as the multiplication theorem for the Bessel functions). Furthermore, the function $z^{\frac{1}{2}} \mathcal{C}_\nu(\lambda z)$, $\lambda\ne 0$, satisfies the differential equation
\begin{equation}
\label{eeeeeeeeeee}
\frac{d^2 w}{dz^2}=\left( \frac{\nu^2-\frac{1}{4}}{z^2} -\lambda^2 \right) w.
\end{equation}

Finally, we quote here the following results about the zeros of $J_\nu(z)$ and $H_\nu^{(1)}(z)$, $H_\nu^{(2)}(z)$:

\begin{proposition}
\label{proposition-15-aprile-2024-1}
Let $\nu\ge 0$. Then
\begin{itemize}
\item[i)] $J_\nu(z)$ has only positive real simple roots which are denoted $j_{\nu,s}$, $s\in\mathbb{Z}_{\ge 0}$;
\item[ii)] $H_\nu^{(1)}(z)$ has no roots in the closed sector $\left\{0\le\arg(z)\le\pi\right\}$ and (hence) $H_\nu^{(2)}(z)$ has no roots in the closed sector $\left\{-\pi\le\arg(z)\le 0\right\}$.
\end{itemize}
\end{proposition}
\begin{proof}
Point i) is known as Lommel's Theorem and the proof can be found in \cite{watson-bessel-functions}. For point ii) see \cite{erdelyi-special-functions} and \cite{article:hilbfalckenberg1916}, \cite{article:falckenberg1932} for a wider study.
\end{proof}

\subsection{Modified Bessel functions}
\label{appendix-14-dez-2023}
The modified Bessel functions are solutions to the differential equation
\begin{equation}
\label{general-modified-bessel-equation}
z^2 \frac{d^2 w}{dz^2}+z\frac{dw}{dz}-(z^2+\nu^ 2)w=0, \quad \nu\in\mathbb{C}.
\end{equation}
An independent pair of standard solutions is given by
\begin{equation}
\label{modified-bessel-first-kind}
I_\nu (z):=\left(\frac{z}{2}\right)^\nu \sum_{k=0}^\infty \frac{1}{k!\Gamma\left(\nu+1+k\right)} \left(\frac{z}{2}\right)^{2k}
\end{equation}
(known as modified Bessel function of the first kind), and by
\begin{equation}
\label{modified-bessel-second-kind}
K_{\nu}(z):=\frac{\pi}{2}\csc(\pi\nu)\left(I_{-\nu}(z)-I_{\nu}(z)\right)
\end{equation}
when $\nu$ is not an integer, while for integer $\nu$ it is defined by a limiting operation of the previous formula (this function is known as the modified Bessel function of the third kind). For both functions, the principal branch is defined to be the principal branch of $z^\nu$ and it is an analytic function on $\mathbb{C}\setminus\left(-\infty,0\right]$. For each $z\ne 0$, all branches of $I_{\nu}(z)$ and $K_{\nu}(z)$ are analytic in $\nu\in\mathbb{C}$. Both $I_{\nu}(z)$ and $K_{\nu}(z)$ are real when $\nu$ and $z$ are real. Furthermore, for any nonnegative order $\nu$ the function $I_\nu(z)$ is positive and monotonically increasing on $(0,+\infty)$, while the function $K_\nu(z)$ is positive and monotonically decreasing on the same interval.

The analytic continuation of the modified Bessel functions of the first and third kind on the whole universal covering space $\widetilde{\mathbb{C}^*}$ is given by the formulas
\begin{equation}
\label{8-gennaio-2024-2}
K_\nu\left(z e^{im\pi}\right)=e^{-im\pi\nu} K_\nu(z)-i\pi \sin\left(m\nu\pi\right) \csc\left(\nu \pi\right) I_\nu(z)
\end{equation}
and
\begin{equation}
\label{8-gennaio-2024-3}
I_\nu\left(z e^{i m\pi}\right)=e^{i m\pi\nu} I_\nu(z),
\end{equation}
for all $m\in \mathbb{Z}$.

The modified Bessel functions $I_\nu(z)$ and $K_\nu(z)$ are related to the Bessel functions by the following connection formulas:

\begin{equation}
\label{15-aprile-2024-1}
I_\nu(z)=e^{-\frac{\nu\pi i}{2}} J_\nu\left( z e^{i\frac{\pi}{2}} \right),\quad -\pi\le\arg(z)\le\frac{\pi}{2},
\end{equation}
\begin{equation}
\label{15-aprile-2024-2}
I_\nu(z)=\frac{1}{2}e^{\mp\frac{\nu\pi i}{2}}\left[ H_\nu^{(1)}\left(z e^{\pm i \frac{\pi}{2}}\right) + H_\nu^{(2)}\left(z e^{\pm i \frac{\pi}{2}}\right) \right],\quad -\pi\le\pm\arg(z)\le\frac{\pi}{2},
\end{equation}
and
\begin{equation}
\label{15-aprile-2024-3}
K_\nu(z)=
\begin{cases}
\frac{i\pi}{2} e^{\frac{\nu\pi i}{2}} H_\nu^{(1)}\left( z e^{i\frac{\pi}{2}} \right), & -\pi\le\arg(z)\le \frac{\pi}{2}, \\
-\frac{i\pi}{2} e^{-\frac{\nu\pi i}{2}} H_\nu^{(2)}\left( ze^{-i\frac{\pi}{2}} \right), & -\frac{\pi}{2}\le\arg(z)\le \pi.
\end{cases}
\end{equation}

For what concerns the asymptotic expansions, it can be shown that $K_\nu(z)$ and $I_\nu(z)$ admit the representations
\begin{equation}
\label{8-gennaio-2024-4}
K_\nu(z)=\left(\frac{\pi}{2 z}\right)^{\frac{1}{2}} e^{-z} \left[\sum_{s=0}^{n-1}\frac{A_s(\nu)}{z^s}+R_n(z;\nu)\right]
\end{equation}
and
\begin{equation}
\label{8-gennaio-2024-5}
I_\nu(z)=\frac{e^z}{(2 \pi z)^{\frac{1}{2}}}\left[\sum_{s=0}^{n-1} (-1)^s \frac{A_s(\nu)}{z^s}+ \widetilde{R}_n(z;\nu)\right]\mp i e^{\mp i\pi \nu}\frac{e^{-z}}{(2 \pi z)^{\frac{1}{2}}}\left[ \sum_{s=0}^{n-1}\frac{A_s(\nu)}{z^s}+ R_n(z;\nu) \right],
\end{equation}
where $n\in\mathbb{N}$ is any fixed natural number,
\begin{equation}
\label{8-gennaio-2024-6}
A_s(\nu):= \frac{\left(\frac{1}{2}-\nu\right)_s\left(\frac{1}{2}+\nu\right)_s}{(-2)^s s!}
\end{equation}
(here $(-)_s$ denotes the Pochhammer symbol) and the remainders $R_n(z;\nu)$, $\widetilde{R}_n(z;\nu)$ have the bounds
\begin{equation}
\label{8-gennaio-2024-7}
\left| R_n(z;\nu) \right| \le
\begin{cases}
2 \left|\frac{A_n(\nu)}{z^n}\right| \exp\left|\frac{1}{z}\left(\nu^2-\frac{1}{4}\right)\right| & |\arg(z)|\le\frac{\pi}{2}, \\
2\sqrt{\pi}\frac{\Gamma\left(\frac{n}{2}+1\right)}{\Gamma\left(\frac{n}{2}+\frac{1}{2}\right)} \left|\frac{A_n(\nu)}{z^n}\right| \exp\left|\frac{\pi}{2 z}\left(\nu^2-\frac{1}{4}\right)\right| & \frac{\pi}{2} \le |\arg(z)|\le \pi, \\
4\sqrt{\pi}\frac{\Gamma\left(\frac{n}{2}+1\right)}{\Gamma\left(\frac{n}{2}+\frac{1}{2}\right)} \left|\frac{A_n(\nu)}{(\operatorname{Re}(z))^n}\right| \exp\left|\frac{\pi}{\operatorname{Re}(z)}\left(\nu^2-\frac{1}{4}\right)\right| & \pi \le |\arg(z)|< \frac{3\pi}{2}
\end{cases}
\end{equation}
and
\begin{equation}
\label{8-gennaio-2024-8}
\left|\widetilde{R}_n(z;\nu)\right|\le
\begin{cases}
4\sqrt{\pi}\frac{\Gamma\left(\frac{n}{2}+1\right)}{\Gamma\left(\frac{n}{2}+\frac{1}{2}\right)} \left|\frac{A_n(\nu)}{(\operatorname{Re}(z))^n}\right| \exp\left|\frac{\pi}{\operatorname{Re}(z)}\left(\nu^2-\frac{1}{4}\right)\right| & 0 \le  \pm\arg(z)<\frac{\pi}{2}, \\
2\sqrt{\pi}\frac{\Gamma\left(\frac{n}{2}+1\right)}{\Gamma\left(\frac{n}{2}+\frac{1}{2}\right)} \left|\frac{A_n(\nu)}{z^n}\right| \exp\left|\frac{\pi}{2 z}\left(\nu^2-\frac{1}{4}\right)\right| & -\frac{\pi}{2} \le \pm \arg(z)\le 0, \\
2 \left|\frac{A_n(\nu)}{z^n}\right| \exp\left|\frac{1}{z}\left(\nu^2-\frac{1}{4}\right)\right| & -\frac{3\pi}{2}\le \pm\arg(z)\le-\frac{\pi}{2},
\end{cases}
\end{equation}
respectively (these formulas can be found in \cite{olver-special-functions}, Exercise 13.2 of Chapter 7, while the proofs can be found in \cite{watson-bessel-functions}, Chapter VII). In~\eqref{8-gennaio-2024-8} the choices \textquotedblleft$\pm$\textquotedblright\, correspond to the choices \textquotedblleft$\mp$\textquotedblright\, in~\eqref{8-gennaio-2024-5}.

Denoting by $\mathcal{Z}_{\nu}(z)$ any solution to~\eqref{general-modified-bessel-equation}, the function $z^{\frac{1}{2}}\mathcal{Z}_{\frac{1}{p}}\left(\frac{2\lambda}{p} z^{\frac{p}{2}} \right)$, with $\lambda,\,p\ne 0$, satisfies the equation
\begin{equation}
\label{modified-bessel-equation-second-form}
\frac{d^2 w}{dz^2}-\lambda^2 z^{p-2} w=0
\end{equation}
while the first derivative of $\mathcal{Z}_\nu(z)$ is given by the formula
\begin{equation}
\label{derivative-modified-bessel-functions}
\frac{1}{z}\frac{d}{dz}\left(z^\nu \mathcal{Z}_\nu(z)\right)=z^{\nu-1}\mathcal{Z}_{\nu-1}(z).
\end{equation}

Finally, we quote here the following result about the zeros of the modified Bessel function $K_\nu(z)$:
\begin{proposition}
\label{theorem-10-gennaio-2024-1}
Let $\nu\ge 0$. The number of zeros of the modified Bessel function $K_\nu(z)$ in the sector $\left\{|\arg(z)|\le\pi\right\}$ is the even integer which is closer to $\nu-1/2$ unless $\nu-1/2$ is an integer, in which case the number is $\nu-1/2$.
\end{proposition}
\begin{proof}
See \cite{article:MacDonald1898} and \cite{watson-bessel-functions} for a wider study.
\end{proof}

\subsection{Airy functions}
\label{appendix-29-gennaio-2024}
The Airy functions are the entire functions $\operatorname{Ai}(z)$ and $\operatorname{Bi}(z)$ which solve the Airy equation
\[
\frac{d^2 w} {dz^2}=z w,\quad z\in\mathbb{C}
\]
and with Taylor series
\begin{equation}
\label{taylor-ai}
\operatorname{Ai}(z)=\frac{1}{3^{\frac{2}{3}}\pi}\sum_{k\ge 0}\frac{1}{k!}\Gamma\left(\frac{k+1}{3}\right)\sin\left(\frac{2 \pi (k+1)}{3}\right) \left(3^{\frac{1}{3} } z\right)^k
\end{equation}
and
\begin{equation}
\label{taylor-bi}
\operatorname{Bi}(z)=\frac{1}{3^{\frac{1}{6}}\pi}\sum_{k\ge 0}\frac{1}{k!}\Gamma\left(\frac{k+1}{3}\right)\left|\sin\left(\frac{2 \pi (k+1)}{3}\right)\right| \left(3^{\frac{1}{3} } z\right)^k.
\end{equation}
The Airy functions $\operatorname{Ai}(z)$ and $\operatorname{Bi}(z)$ admit the following representations:
\begin{equation}
\label{29-gennaio-2024-1}
\begin{aligned}
&\operatorname{Ai}(z)=\frac{e^{-\frac{2}{3}z^{\frac{3}{2}}}}{2\sqrt{\pi}z^{\frac{1}{4}}}\left(1  +R_1(z)\right),\quad \left|\arg(z)\right|<\pi, \\
& \operatorname{Bi}(z)=\frac{e^{\frac{2}{3} z^{\frac{3}{2}}}}{\sqrt{\pi} z^{\frac{1}{4}}}\left(1+\frac{R_1\left(z e^{-\frac{2 \pi i}{3}}\right)+ R_1\left(z e^{\frac{2\pi i}{3}}\right)}{2}\right),\quad \left|\arg(z)\right|<\frac{\pi}{3},
\end{aligned}
\end{equation}
where the remainder term $R_1(z)$ has bounds
\begin{equation}
\label{29-gennaio-2024-2}
\left|R_1(z)\right|\le
\begin{cases}
\frac{5}{72}\left|\frac{2}{3} z^{\frac{3}{2}}\right|^{-1}, & \left|\arg(z)\right|\le \frac{\pi}{3}, \\
\frac{5}{72}\left|\frac{2}{3} z^{\frac{3}{2}}\right|^{-1}\max\left\{\left|\csc\left(\frac{3}{2}\arg(z)\right)\right|, 1+7\sqrt{\pi}\frac{\Gamma\left(\frac{7}{12}\right)}{\Gamma\left(\frac{1}{12}\right)}\right\}, & \frac{\pi}{3}\le\left|\arg(z)\right|\le \frac{2\pi}{3}, \\
\frac{5}{72}\left|\frac{2}{3} z^{\frac{3}{2}}\right|^{-1}\left[\frac{\sqrt{7\pi}}{\sqrt{3}\left|\cos\left(\frac{3}{2}\arg(z)\right)\right|^{\frac{7}{6}}}\left( 1+7\sqrt{\pi}\frac{\Gamma\left(\frac{7}{12}\right)}{\Gamma\left(\frac{1}{12}\right)}\right)\right], & \frac{2\pi}{3}\le\left|\arg(z)\right|<\pi,
\end{cases}
\end{equation}
and
\begin{equation}
\label{29-gennaio-2024-1-bis}
\begin{aligned}
& \operatorname{Ai}(-z)=\frac{1}{\sqrt{\pi} z^{\frac{1}{4}}}\left[\cos\left(\frac{2}{3}z^{\frac{3}{2}}-\frac{\pi}{4}\right)+\widetilde{R}_1(z)\right]\quad|\arg(z)|<\frac{2\pi}{3}, \\
& \operatorname{Bi}(-z)=\frac{1}{\sqrt{\pi} z^{\frac{1}{4}}}\left[-\sin\left(\frac{2}{3}z^{\frac{3}{2}}-\frac{\pi}{4}\right)+\widetilde{R}_2(z)\right],\quad|\arg(z)|<\frac{2\pi}{3},
\end{aligned}
\end{equation}
where the remainder terms $\widetilde{R}_1(z)$ and $\widetilde{R}_1(z)$ are
\begin{equation}
\label{29-gennaio-2024-1-ter}
\begin{aligned}
& 2\widetilde{R}_1(z):= e^{-i\left(\frac{2}{3}z^{\frac{3}{2}}-\frac{\pi}{4}\right)} R_1\left(z e^{i\frac{\pi}{3}}\right)+e^{i\left(\frac{2}{3}z^{\frac{3}{2}}-\frac{\pi}{4}\right)} R_1\left(z e^{-i\frac{\pi}{3}}\right),\\
& 2\widetilde{R}_2(z):=e^{-i\left(\frac{2}{3}z^{\frac{3}{2}}+\frac{\pi}{4}\right)} R_1\left(z e^{i\frac{\pi}{3}}\right)+e^{i\left(\frac{2}{3}z^{\frac{3}{2}}+\frac{\pi}{4}\right)} R_1\left(z e^{-i\frac{\pi}{3}}\right).
\end{aligned}
\end{equation}
As a consequence of representation~\eqref{29-gennaio-2024-1} and~\eqref{taylor-ai},~\eqref{taylor-bi}, it follows that for any closed subsector $\mathcal{S}_1\subset\left\{|\arg(z)|<\pi\right\}$ and any closed subsector $\mathcal{S}_2\subset\left\{|\arg(z)|<\frac{\pi}{3}\right\}$ there exist constants $0<c_{\mathcal{S}_1}^{(1)}\le c_{\mathcal{S}_1}^{(2)}$ depending only on the opening angle of $\mathcal{S}_1$ and constants $0<d_{\mathcal{S}_2}^{(1)}\le d_{\mathcal{S}_2}^{(2)}$ depending only on the opening angle of $\mathcal{S}_2$ such that
\begin{equation}
\label{useful-bounds}
\begin{aligned}
& c_{\mathcal{S}_1}^{(1)} \le\left|\left(1+|z|^{\frac{1}{4}}\right)e^{\frac{2}{3}z^{\frac{3}{2}}} \operatorname{Ai}(z)\right|\le c_{\mathcal{S}_1}^{(2)},\quad z\in \mathcal{S}_1, \\
& c_{\mathcal{S}_2}^{(1)} \le\left|\left(1+|z|^{\frac{1}{4}}\right)e^{-\frac{2}{3}z^{\frac{3}{2}}} \operatorname{Bi}(z)\right|\le c_{\mathcal{S}_2}^{(2)},\quad z\in \mathcal{S}_2.
\end{aligned}
\end{equation}
The derivatives $\operatorname{Ai}'(z)$ and $\operatorname{Bi}'(z)$ admit the following representations:
\begin{equation}
\label{derivative-ai-bi}
\begin{aligned}
& \operatorname{Ai}'(z)=-\frac{z^{\frac{1}{4}} e^{-\frac{2}{3}z^{\frac{3}{2}}}}{2 \sqrt{\pi}}\left(1+P_1(z)\right),\quad |\arg(z)|<\pi \\
& \operatorname{Bi}'(z)=\frac{z^{\frac{1}{4}}e^{\frac{2}{3}z^{\frac{3}{2}}}}{\sqrt{\pi}}\left(1+\frac{P_1\left(z e^{-\frac{2\pi i}{3}}\right)+P_1\left(z e^{\frac{2 \pi i}{3}}\right)}{2}\right),\quad |\arg(z)|<\frac{\pi}{3}
\end{aligned}
\end{equation}
where the remainder term $P_1(z)$ has bounds
\begin{equation}
\label{derivative-ai-bi-bounds}
|P_1(z)|\le 
\begin{cases}
\frac{5}{72}\left|\frac{2}{3} z^{\frac{3}{2}}\right|^{-1}, & \left|\arg(z)\right|\le \frac{\pi}{3}, \\
\frac{5}{72}\left|\frac{2}{3} z^{\frac{3}{2}}\right|^{-1}\max\left\{\left|\csc\left(\frac{3}{2}\arg(z)\right)\right|, 1+\frac{\pi}{2}\right\}, & \frac{\pi}{3}\le\left|\arg(z)\right|\le \frac{2\pi}{3}, \\
\frac{5}{72}\left|\frac{2}{3} z^{\frac{3}{2}}\right|^{-1}\left[\frac{\sqrt{2\pi}}{\left|\cos\left(\frac{3}{2}\arg(z)\right)\right|} +\frac{\pi}{2}+1\right], & \frac{2\pi}{3}\le\left|\arg(z)\right|<\pi,
\end{cases}
\end{equation}
and
\begin{equation}
\label{derivative-ai-bi-bis}
\begin{aligned}
& \operatorname{Ai}'(-z)=\frac{z^{\frac{1}{4}}}{\sqrt{\pi}}\left[\sin\left(\frac{2}{3}z^{\frac{3}{2}}-\frac{\pi}{4}\right)+\widetilde{P}_1(z)\right],\quad |\arg(z)|<\frac{2 \pi}{3}, \\
& \operatorname{Bi}'(z)=\frac{z^{\frac{1}{4}}}{\sqrt{\pi}}\left[\cos\left(\frac{2}{3}z^{\frac{3}{2}}-\frac{\pi}{4}\right)+\widetilde{P}_2(z)\right],\quad |\arg(z)|<\frac{2\pi}{3}
\end{aligned}
\end{equation}
where the remainder term $\widetilde{P}_1(z)$ and $\widetilde{P}_2(z)$ are
\begin{equation}
\label{derivative-ai-bi-bounds-bis}
\begin{aligned}
& (-2i)\widetilde{P}_1(z):=e^{i\left(\frac{2}{3}z^{\frac{3}{2}}-\frac{\pi}{4}\right)}P_1\left(z e^{\frac{i\pi}{3}}\right)-e^{i\left(\frac{2}{3}z^{\frac{3}{2}}-\frac{\pi}{4}\right)} P_1\left(z e^{-\frac{i\pi}{3}}\right), \\
& 2\widetilde{P}_2(z):=e^{-i\left(\frac{2}{3}z^{\frac{3}{2}}-\frac{\pi}{4}\right)}P_1\left(z e^{\frac{i\pi}{3}}\right)+e^{i\left(\frac{2}{3}z^{\frac{3}{2}}-\frac{\pi}{4}\right)} P_1\left(z e^{-\frac{i\pi}{3}}\right).
\end{aligned}
\end{equation}
The last property of the Airy functions we want to quote  is the following
\begin{proposition}
\label{zeros-airy}
The functions $\operatorname{Ai}(z)$ and $\operatorname{Bi}(z)$ have infinitely many zeros, all of them are simple and lie on the negative real axis.
\end{proposition}

\begin{proof}
See e.g. \cite{watson-bessel-functions}.
\end{proof}

\addcontentsline{toc}{section}{References}

\printbibliography

\end{document}